%% file: main.tex
\documentclass[11pt]{article}
\usepackage{ifthen}
\usepackage{mdwlist}
\usepackage{amsmath,amssymb,amsfonts,amsthm}
\usepackage{bm}
\usepackage{bbm}
\usepackage{mdframed}
\usepackage[colorlinks=true,linkcolor=blue,citecolor=blue,urlcolor=blue]{hyperref}
\usepackage{enumitem}
\usepackage{graphicx}
\usepackage{xspace}
\usepackage{verbatim}
\usepackage{algorithm}
\usepackage{algpseudocode}
\usepackage[margin=1in]{geometry}
\usepackage{color}
\usepackage{thm-restate}
\usepackage{latexsym}
\usepackage{epsfig}
\usepackage{booktabs}
\usepackage{cleveref}
\usepackage[colorinlistoftodos,textsize=scriptsize]{todonotes}
\def\withcolors{1}
\def\withnotes{1}
\input{preamble}

\title{Efficient Mean Estimation with Pure Differential Privacy via a Sum-of-Squares Exponential Mechanism\thanks{Authors are in alphabetical order.}}

\author {
    Samuel B.\ Hopkins\thanks{. {\tt samhop@mit.edu}. Supported by a Miller Postdoctoral Fellowship and a Simons Postdoctoral Fellowship. Part of this work was conducted while visiting the Simons Institute for the Theory of Computing.}
\and
Gautam Kamath\thanks{Cheriton School of Computer Science, University of Waterloo. {\tt g@csail.mit.edu}. Supported by an NSERC Discovery Grant and a University of Waterloo Startup grant.}
\and
Mahbod Majid\thanks{Cheriton School of Computer Science, University of Waterloo. {\tt m2majid@uwaterloo.ca}. Supported by an NSERC Discovery Grant, a Graduate Excellence Award in Computer Science, a David R.\ Cheriton Graduate Scholarship, and a Waterloo CPI Cybersecurity and Privacy Excellence Graduate Scholarship.}
}

\begin{document}
\maketitle

\setcounter{tocdepth}{2}

\thispagestyle{empty}

\begin{abstract}
    We give the first polynomial-time algorithm to estimate the mean of a $d$-variate probability distribution with bounded covariance from $\tilde{O}(d)$ independent samples subject to \emph{pure} differential privacy.
    Prior algorithms for this problem either incur exponential running time, require $\Omega(d^{1.5})$ samples, or satisfy only the weaker \emph{concentrated} or \emph{approximate} differential privacy conditions.
    In particular, all prior polynomial-time algorithms require $d^{1+\Omega(1)}$ samples to guarantee small privacy loss with ``cryptographically'' high probability, $1-2^{-d^{\Omega(1)}}$, while our algorithm retains $\tilde{O}(d)$ sample complexity even in this stringent setting.

    Our main technique is a new approach to use the powerful \emph{Sum of Squares method (SoS)} to design differentially private algorithms.
    \emph{SoS proofs to algorithms} is a key theme in numerous recent works in high-dimensional algorithmic statistics -- estimators which apparently require exponential running time but whose analysis can be captured by \emph{low-degree Sum of Squares proofs} can be automatically turned into polynomial-time algorithms with the same provable guarantees.
    We demonstrate a similar \emph{proofs to private algorithms} phenomenon: instances of the workhorse \emph{exponential mechanism} which apparently require exponential time but which can be analyzed with low-degree SoS proofs can be automatically turned into polynomial-time differentially private algorithms.
    We prove a meta-theorem capturing this phenomenon, which we expect to be of broad use in private algorithm design.

    Our techniques also draw new connections between differentially private and robust statistics in high dimensions.
    In particular, viewed through our proofs-to-private-algorithms lens, several well-studied SoS proofs from recent works in algorithmic robust statistics directly yield key components of our differentially private mean estimation algorithm.
\end{abstract}

\newpage

{
\tableofcontents
\thispagestyle{empty}
 }

\newpage

\setcounter{page}{1}

\input{intro}

\input{prelims}
\input{metatheorem}

\input{coarse-estimation}
\input{fine-estimation}
\input{lbs}

\section*{Acknowledgements}
The authors would like to thank Adam Smith for his role in the conception and contributions in the early stages of this project.

\bibliographystyle{alpha}
\bibliography{biblio}

\appendix
\input{deferred-proofs}
\input{lemmata}

\end{document}

%% file: preamble.tex
\def\eps{\ve}
\renewcommand{\epsilon}{\ve}
\def\ve{\varepsilon}

\DeclareMathOperator{\E}{\mbox{\bf E}}

\newcommand{\pr}[2][]{\mathrm{Pr}\ifthenelse{\not\equal{}{#1}}{_{#1}}{}\!\left[#2\right]}
\newcommand{\R}{\mathbb{R}}
\newcommand{\N}{\mathbb{N}}
\newcommand{\Z}{\mathbb{Z}}

\newcommand{\cC}{\mathcal{C}}
\newcommand{\cK}{\mathcal{K}}
\newcommand{\cX}{\mathcal{X}}

\DeclareMathOperator{\argmax}{argmax}
\DeclareMathOperator{\Cov}{Cov}
\DeclareMathOperator{\vol}{vol}
\DeclareMathOperator{\diam}{diam}
\providecommand{\poly}{\operatorname*{poly}}

\newcommand{\e}{\epsilon}
\newcommand{\mper}{\, .}
\newcommand{\mcom}{\, ,}

\newcommand{\proves}{\vdash}

\newtheorem{theorem}{Theorem}
\newtheorem{assumption}[theorem]{Assumption}

\newtheorem{proposition}[theorem]{Proposition}

\newtheorem{remark}[theorem]{Remark}
\newtheorem{fact}[theorem]{Fact}
\newtheorem{lemma}[theorem]{Lemma}

\newtheorem{corollary}[theorem]{Corollary}

\newtheorem{definition}[theorem]{Definition}

\newtheorem{example}[theorem]{Example}
\numberwithin{theorem}{section} 
\numberwithin{assumption}{section}
\numberwithin{nontheorem}{section} 
\numberwithin{proposition}{section} 
\numberwithin{observation}{section} 
\numberwithin{remark}{section} 
\numberwithin{fact}{section} 
\numberwithin{lemma}{section} 
\numberwithin{claim}{section} 
\numberwithin{corollary}{section} 
\numberwithin{case}{section} 
\numberwithin{dfn}{section} 
\numberwithin{definition}{section} 
\numberwithin{question}{section} 
\numberwithin{openquestion}{section} 
\numberwithin{res}{section} 
\numberwithin{example}{section}

\newcommand{\red}[1]{{\color{red} {#1}}} 

\ifnum\withcolors=1
  \newcommand{\gcolor}[1]{{\color{red}#1}} 
\else

  \newcommand{\gcolor}[1]{{#1}}
\fi

\ifnum\withnotes=1
  \newcommand{\gnote}[1]{\par\gcolor{\textbf{G: }\sf #1}} 
  \newcommand{\gfootnote}[1]{\footnote{{\bf \gcolor{Gautam}}: {#1}}}

  \newcommand{\Snote}[1]{{\color{blue}{\textbf{S: }\sf #1}}} 

\else
  \newcommand{\gnote}[1]{}
  \newcommand{\Snote}[1]{}
  \newcommand{\gfootnote}[1]{}

\fi
\newcommand{\ignore}[1]{\leavevmode\unskip} 

\newcommand{\paren}[1]{(#1)}
\newcommand{\Paren}[1]{\left(#1\right)}

\newcommand{\brac}[1]{[#1]}
\newcommand{\Brac}[1]{\left[#1\right]}

\newcommand{\floor}[1]{\lfloor#1\rfloor}

\newcommand{\ceil}[1]{\lceil#1\rceil}

\newcommand{\abs}[1]{\lvert#1\rvert}
\newcommand{\Abs}[1]{\left\lvert#1\right\rvert}

\newcommand{\set}[1]{\{#1\}}
\newcommand{\Set}[1]{\left\{#1\right\}}

\newcommand{\norm}[1]{\lVert#1\rVert}
\newcommand{\Norm}[1]{\left\lVert#1\right\rVert}

\newcommand{\normt}[1]{\norm{#1}_2}
\newcommand{\Normt}[1]{\Norm{#1}_2}

\newcommand{\snormt}[1]{\norm{#1}^2_2}

\newcommand{\snorm}[1]{\norm{#1}^2}

\newcommand{\iprod}[1]{\langle#1\rangle}
\newcommand{\Iprod}[1]{\left\langle#1\right\rangle}

\newcommand{\normsch}[1]{\vert\kern-0.25ex\vert\kern-0.25ex\vert #1 
\vert\kern-0.25ex\vert\kern-0.25ex\vert}
\newcommand{\Normsch}[1]{\left\vert\kern-0.25ex\left\vert\kern-0.25ex\left\vert #1 
\right\vert\kern-0.25ex\right\vert\kern-0.25ex\right\vert}
\DeclareMathOperator{\pE}{\tilde{\mathbf{E}}}
\newcommand{\Psymb}{\mathbb{P}}

\DeclareMathOperator*{\ProbOp}{\Psymb}
\renewcommand{\Pr}{\ProbOp}

\newcommand{\suchthat}{\;\middle\vert\;}
\newcommand{\tensor}{\otimes}

\newcommand{\from}{\colon}

\newcommand\bdot\bullet
\newcommand{\trsp}[1]{{#1}^\dagger}

\DeclareMathOperator{\Tr}{Tr}
\DeclareMathOperator{\SDP}{SDP}

\DeclareMathOperator{\OPT}{OPT}

\newcommand{\iid}{i.i.d.\xspace}

\newcommand\naively{na\"{\i}vely\xspace}

\newcommand{\cB}{\mathcal B}
\newcommand{\cD}{\mathcal D}

\newcommand{\cG}{\mathcal G}
\newcommand{\cH}{\mathcal H}

\newcommand{\cL}{\mathcal L}
\newcommand{\cM}{\mathcal M}

\newcommand{\cO}{\mathcal O}
\newcommand{\cP}{\mathcal P}

\newcommand{\cR}{\mathcal R}
\newcommand{\cS}{\mathcal S}
\newcommand{\cT}{\mathcal T}
\newcommand{\cU}{\mathcal U}

\newcommand{\cY}{\mathcal Y}

\newcommand{\bbB}{\mathbb B}

\newcommand{\bbS}{\mathbb S}

\newcommand{\wh}[1]{\ensuremath{\widehat{#1}}}

\newcommand*{\transpose}[1]{{#1}{}^{\mkern-1.5mu\mathsf{T}}}

\DeclareMathOperator*{\argmin}{arg\,min}

\makeatletter
\newcommand{\distas}[1]{\mathbin{\overset{#1}{\kern\z@\sim}}}%
\newsavebox{\mybox}\newsavebox{\mysim}
\newcommand{\distras}[1]{%
	\savebox{\mybox}{\hbox{\kern3pt$\scriptstyle#1$\kern3pt}}%
	\savebox{\mysim}{\hbox{$\sim$}}%
	\mathbin{\overset{#1}{\kern\z@\resizebox{\wd\mybox}{\ht\mysim}{$\sim$}}}%
}
\makeatother

%% file: intro.tex
\section{Introduction}
\emph{Mean estimation} is perhaps the most elementary statistical task: given samples from a probability distribution $D$, estimate its expected value.
In this paper, we study mean estimation in $d$ dimensions subject to \emph{differential privacy (DP)}~\cite{DworkMNS06}, a rigorous notion of data privacy.
Privacy is of natural concern in high-dimensional statistics, where data may be sensitive and standard estimators like the empirical mean may leak information about individuals in a dataset (see, e.g., privacy attacks in~\cite{DinurN03, HomersrDTMPSNC08, BunUV14, DworkSSUV15} and the survey~\cite{DworkSSU17}).
On the other hand, privacy and statistical estimation seem largely compatible: in the large-sample limit, good statistical estimators will have vanishing dependence on each individual sample anyway.
Indeed, even though the empirical mean (the natural benchmark for accuracy in mean estimation) is not differentially private, estimators are known which match its accuracy guarantees while also providing differential privacy~\cite{BunKSW19, KamathSU20}.
Namely, an accurate and private estimate of the mean can be obtained using $n = O(d)$ samples.
However, known estimators achieving this sample complexity require exponential running time.

If one instead focuses on polynomial-time algorithms, existing estimators all face significant drawbacks: either they require $n \geq d^{1+\Omega(1)}$ samples, or they may leak private information with probability $2^{-d^{c}}$ for small $c > 0$ \cite{HardtT10, KamathLSU19, KamathSU20}.
(In privacy language, they satisfy only \emph{concentrated} or \emph{approximate} differential privacy.)
Since a major goal of privacy is to make strong assurances on leakage of sensitive information, best practices disallow private data leakage even with ``cryptographically'' small probability $2^{-\poly(d)}$.
One way to satisfy this stringent requirement is to provide \emph{pure} differential privacy, which disallows substantial leakage of private information with any nonzero probability.
Thus, the main question in our paper is:
\begin{center}
    \emph{Is there a polynomial-time pure DP algorithm for mean estimation using $n = O(d)$ samples?}
\end{center}

Our main result answers this question affirmatively, up to logarithmic factors.
To state our main result, we define differential privacy:
\begin{definition}[(Pure) differential privacy]
    For $\e > 0$, a (randomized) algorithm $A$ which takes $n$ inputs $X_1,\ldots,X_n$ is $\e$-differentially private ($\e$-\emph{DP}) if for every pair of inputs $X_1,\ldots,X_n$ and $X_1',\ldots,X_n'$ such that $X_i = X_i'$ for all except a single index $i \in [n]$ and for all possible events $S \subseteq \mathrm{Range}(A)$,
    \[
        \Pr(A(X_1,\ldots,X_n) \in S) \leq e^\e \cdot \Pr(A(X_1',\ldots,X_n') \in S)\mper
    \]
\end{definition}

\begin{theorem}[Main Theorem — Combination of \Cref{thm:coarse-estimation}, \Cref{thm:high-dimensional-fine-mean-estimation}]
    \label{thm:main-intro}
    For every $n,d \in \N$ and $R, \alpha, \e, \beta > 0$ there is a polynomial-time $\e$-DP algorithm such that for every distribution $D$ on $\R^d$ such that $\| \E_{X \sim D} X \| \leq R$ and $\Cov_{X \sim D}(X) \preceq I$, given $X_1,\ldots,X_n \sim D$, with probability $1-\beta$ the algorithm outputs $\hat{\mu}$ such that $\|\hat{\mu} - \E_{X \sim D} X \| \leq \alpha$, so long as
    \[
        n \geq \tilde{O} \Paren{\frac{d + \log(1/\beta)}{\alpha^2 \e} + \frac{d\log R + \min(d,\log R) \cdot \log(1/\beta)}{\e}}\mper
    \]
    Furthermore, if an $\eta$-fraction of the samples $X_1,\ldots,X_n$ are adversarially corrupted, the algorithm maintains the same guarantee, at the cost that now $\|\hat{\mu} - \E X\| \leq \alpha + O(\sqrt \eta)$.
\end{theorem}

The sample complexity of our algorithm is nearly linear in $d$, thereby answering our main question up to logarithmic factors.
Beyond this core goal, we highlight that our algorithm is \emph{also} robust to adversarial contaminations and enjoys sub-Gaussian confidence intervals when $n \gg \frac{d \log R + \min(d,\log R) \log(1/\beta)}{\e}$, both sought-after features in recent non-private algorithms~\cite{DiakonikolasKKLMS19,LaiRV16,Hopkins20,CherapanamjeriFB19}.
The sample complexity of our algorithm is nearly optimal, with the exception of the term $(\log(1/\beta) \cdot \min(d,\log R)) / \e$ -- see Theorem~\ref{thm:lb-bdd-mom} for corresponding lower bounds.
(Put differently, information-theoretically it is possible to achieve sub-Gaussian confidence intervals under the slightly milder assumption $n \gg \frac{d \log R + \log(1/\beta)}{\e}$.)

\paragraph{Beyond clip-and-noise}
Obtaining the guarantees of our algorithm requires going beyond existing techniques for private mean estimation, which we briefly review.
A common technique in differential privacy is to ``just add noise.''
More precisely, suppose that $f(X_1,\ldots,X_n)$ is a \emph{bounded sensitivity} function of $n$ inputs, meaning that replacing a single input $X_i$ with an arbitrary $X_i'$ changes the value of $f$ by at most $\Delta$ (in an appropriate choice of norm).
Then, $f(X_1,\ldots,X_n) + Z$ will be private, where $Z$ is an appropriate random variable whose magnitude depends on $\Delta$.

Existing polynomial-time algorithms for private mean estimation all take this approach.
First, to limit the sensitivity of the empirical mean, the algorithms clip the samples $X_1,\ldots,X_n \in \R^d$ to lie in an $\ell_2$ ball.
Considering the case $\Cov(X) \approx I$, we will have $\|X - \E X\| \approx \sqrt d$, so using a ball of radius at least $\sqrt d$ is unavoidable without introducing too much error in the clipping phase.
Then, the algorithms output $\overline{\mu} + Z$, where $\overline{\mu}$ is the empirical mean of the clipped samples.

If one takes the coordinates of $Z$ to be draws from a Laplace distribution, the resulting algorithm will satisfy $\e$-DP, but the relatively heavy tails of the Laplace distribution impose a cost that $n$ must be at least $d^{1.5}/\e$ to obtain nontrivial guarantees.
On the other hand, if $Z$ is Gaussian, the resulting algorithm appears to tolerate $n \approx d/\e$ samples, but it no longer satisfies \emph{pure} DP.
Instead, to guarantee privacy loss at most $\e$ with probability at least $\delta$ over the internal randomness in the algorithm, $n \geq d \sqrt{\log(1/\delta)}/\e$ is required (i.e. the algorithm satisfies \emph{concentrated} DP), meaning that for ``cryptographically small'' $\delta$, this algorithm, too, has super-linear sample complexity.

One might naturally suspect, therefore, that strong privacy guarantees like this simply require $\Omega(d^{1.5})$ samples.
Indeed, if nastier distributions than we consider here (violating bounded covariance) are allowed, these two bounds of $\Omega(d^{1.5})$ and $\Omega(d)$ are known to be tight for pure and concentrated DP, respectively~\cite{HardtT10, BunUV14, KamathLSU19}, and similar separations were previously conjectured even for bounded-covariance distributions~\cite{BunKSW19}.

However, \cite{KamathSU20} (building on related techniques of~\cite{BunKSW19}) show that in exponential time one can go beyond the clip-and-noise approach to mean estimation.
In particular, a tournament-based approach gives a pure-$DP$ algorithm using $O(d/\e)$ samples.
We obtain nearly-matching guarantees in polynomial time -- ours is the first polynomial-time private algorithm to go beyond the clip-and-noise approach.

\paragraph{Sum-of-Squares Proofs to Private Algorithms}
\emph{Proofs to Algorithms} has become a powerful algorithm-design technique in computational high-dimensional statistics.
Roughly speaking, the proofs to algorithms technique shows that statistical estimation problems which can be solved to a given accuracy by some (not necessarily polynomial-time) estimator can actually be solved in polynomial time with the same accuracy \emph{if the \emph{analysis} of that estimator can be captured by a certain restricted but powerful formal proof system, known as the SoS proof system.}
This insight has had major consequences in robust and heavy-tailed statistics, clustering, learning latent variable models, and beyond.
(For instance, \cite{HopkinsL18,kothari2017outlier,kothari2017better,potechin2017exact}, and the survey \cite{raghavendra2018high}.)

We show that this approach also applies to private algorithm design.
Our techniques give a generic method to turn (potentially) exponential-time instances of the workhorse \emph{exponential mechanism}, whose breadth of applicability is hard to overstate \cite{McSherryT07}, into polynomial time algorithms, when their analyses are captured by the same SoS proof system.
This gives us the following \emph{proofs-to-algorithms} principle for private algorithm design, which we anticipate will be widely applicable:
\begin{quote}
    \textbf{Proofs to Private Algorithms} (see Theorem~\ref{thm:metatheorem}): Instances of the exponential mechanism with low-degree SoS proofs of bounded sensitivity and utility automatically yield computationally-efficient private algorithms.
\end{quote}
(See Section~\ref{sec:techniques} for a description of the exponential mechanism and more on SoS proofs.)
We are able to capture this principle as a formal meta-theorem, Theorem~\ref{thm:metatheorem} -- since its statement requires a few more technical definitions than we are ready to state, we defer it for now.

\paragraph{Convex Programming and Private Algorithms}
Like nearly all applications of the Sum of Squares method, algorithms which result from our proofs-to-private-algorithms approach ultimately use convex programs -- semidefinite programs, in our case.
Our techniques give a substantially novel approach to convex programming in private algorithms.
Prior works largely do this in one of two ways.
Some privatize the input \emph{before} applying a convex program to it (e.g., privatizing a graph and then post-processing with convex programming \cite{BlockiBDS12,EliasKKL20}).
This can be limiting, as the convex program cannot itself be helpful in achieving privacy, only in some post-processing.
Other algorithms use a private solver for the convex program itself \cite{hsu2014privately,LiuKKO21}.
This can be quite technically challenging: few DP solvers for convex programs are known, and in particular, we are not aware of any generic DP algorithm for solving semidefinite programs.

Instead, our algorithms \emph{hide convex programs behind the exponential mechanism}.
Namely, we use convex programs as score functions.
(Again, see Section~\ref{sec:techniques} for an explanation of this terminology.)
Our key insight is: \emph{when convex programs are married to the exponential mechanism in the correct way, the resulting exponential mechanisms can be implemented in polynomial time using private log-concave sampling algorithms.}
This builds directly on \cite{BassilyST14}, who use a similar approach but for much simpler un-constrained convex optimization problems resulting from empirical risk minimization problems.
The analysis of this interplay of exponential mechanism and convex programming is completely generic -- meaning it has nothing to do with the particular setting of mean estimation -- and is captured in the proof of our (meta)-Theorem~\ref{thm:metatheorem}.

\paragraph{Robust Statistics and Privacy}
Robust statistics, the study of statistics in the presence of corrupted samples, has enjoyed a recent renaissance~\cite{DiakonikolasK19}.
Robustness and privacy are at least spiritual cousins -- both demand that a statistical estimator not change its behavior ``too much'' when one or several samples are replaced arbitrarily.
However, the formal requirements for robustness and privacy are rather different.
The output of a good robust mean estimator should not move far in Euclidean distance when $1\%$ of the samples it is given are corrupted by a malicious adversary.
Privacy, by contrast, demands that the \emph{distribution} of outputs of an algorithm not shift too much when any single sample is replaced by another.

In spite of this formal difference, \cite{DworkL09} was able to use robust estimators in one dimension to construct private estimators.
This suggests that the flurry of recent algorithms for high-dimensional robust statistics with strong provable guarantees should yield high-dimensional private estimators.

Our work makes good on this promise.
In particular, our algorithm for private mean estimation ultimately employs a well-studied Sum-of-Squares SDP from robust mean estimation, and our analysis applies several SoS proofs originally formulated to analyze that SDP in robust statistics.
Thus, our proofs-to-private-algorithms approach gives a lens through which algorithms from robust statistics can yield private algorithms.

Lastly, there is an even more direct connection between our particular result and recent developments in robust statistics.
The robustness renaissance was kicked off by \cite{DiakonikolasKKLMS16,LaiRV16}, which gave the first polynomial-time algorithms for robustly learning a Gaussian where corrupted samples can be tolerated with dimension-independent error.
Prior polynomial-time algorithms for the same problem have guarantees no better than what can be obtained by naive sample-clipping, and as a result incur error scaling as $\eta \cdot \poly(d)$ when an $\eta$-fraction of samples are corrupted.
Just as \cite{DiakonikolasKKLMS16,LaiRV16} gave the first algorithms with guarantees going beyond naive sample clipping in the robust setting, our work is the first to go beyond clip-and-noise in the private setting.

\subsection{Related Work}
\input{table}
A mature body of work focuses on private mean estimation.
The most relevant results restrict the underlying distribution to satisfy a moment bound (including sub-Gaussianity), examples include~\cite{BarberD14, KarwaV18, BunKSW19, BunS19, KamathLSU19, CaiWZ19, KamathSU20, WangXDX20, DuFMBG20, BiswasDKU20, AdenAliAK21, BrownGSUZ21, KamathLZ21, HuangLY21, KamathMSSU21, AshtianiL21}. 
All prior study of the multivariate case either focuses on concentrated or approximate DP, or provides computationally-inefficient algorithms for pure DP. 
We are the first to give an efficient $O(d)$-sample algorithm for multivariate mean estimation under pure DP.
This matches the sample complexity of the best known algorithms under concentrated DP, while simultaneously strengthening the privacy guarantee.

Other prior work studies private mean estimation of arbitrary (bounded) distributions~\cite{BunUV14, SteinkeU15, DworkSSUV15}.
Interestingly, in this case, the optimal sample complexity under pure and concentrated DP are separated by a factor of $\sqrt{d}$, leading to the naive Laplace mechanism being effectively optimal. 
It appears that our bounded moments assumption induces a qualitatively different structure, which eliminates the benefits of relaxing to concentrated DP.
Pinpointing the precise conditions under which a separation occurs remains an interesting direction for future work.

Other works study private statistical estimation in different and more general settings, including mixtures of Gaussians~\cite{KamathSSU19, AdenAliAL21}, graphical models~\cite{ZhangKKW20}, discrete distributions~\cite{DiakonikolasHS15}, and median estimation~\cite{AvellaMedinaB19, TzamosVZ20}.
Some recent directions involve guaranteeing user-level privacy~\cite{LiuSYKR20, LevySAKKMS21}, or a combination of local and central DP for different users~\cite{AventDK20}.
See~\cite{KamathU20} for further coverage of DP statistical estimation.

Several other works employ sampling-based methods to efficiently implement the exponential mechanism~\cite{KapralovT13, BassilyST14, AminDKMV19, LeakeMV21}.
\cite{KapralovT13,AminDKMV19, LeakeMV21} construct sophisticated sampling algorithms by hand to sample from a non-log-concave distributions; by contrast, we use convex programming to construct our score functions, ensuring that we always stay within the realm of log-concave sampling algorithms.
\cite{BassilyST14} constructs the private log-concave sampler we use in our work, but they employ it only to sample from comparatively simple log-concave distributions arising from empirical risk minimization of convex loss functions.
A recent work of~\cite{MangoubiV21} provides a faster algorithm for private log-concave sampling.

Another recent line of work studies sub-Gaussian confidence intervals for mean estimation of heavy-tailed distributions (i.e., assuming only bounded covariance).
Lugosi and Mendelson~\cite{LugosiM19a} proposed an inefficient algorithm, with an SoS-based algorithm coming in~\cite{Hopkins20}; see also the survey~\cite{LugosiM19b}.
Some works considered efficient algorithms for simultaneously robust and heavy-tailed mean estimation~\cite{DepersinL19,PrasadBR19,DiakonikolasKP20}.
A recent result shows that the core solution concepts for robust and heavy-tailed mean estimation can be considered equivalent~\cite{HopkinsLZ20}.
Our work demonstrates that the line of efficient estimators inspired by~\cite{Hopkins20} is simultaneously effective for robust, heavy-tailed, and private mean estimation.

Relatively limited work simultaneously considers privacy and the other constraints of robustness and heavy-tailed estimation.
The main result is~\cite{LiuKKO21} which considers robust and private mean estimation. 
Our result can be seen as improving on the running time or sample complexity of their two algorithms, and the privacy guarantee of both. 
While other prior works focus on heavy-tailed distributions~\cite{BarberD14, KamathSU20, WangXDX20, KamathLZ21}, these do not address the primary goal in the non-private setting, which is to achieve sub-Gaussian rates with respect to the error probability.
Instead, they generally prove guarantees with constant probability of success, which can be boosted at the cost of a multiplicative (rather than additive) logarithmic factor which is inverse in the failure probability.
Note that some of the previous (inefficient) cover- and median-based approaches to private estimation~\cite{BunKSW19, KamathSU20, LiuKKO21, BrownGSUZ21, AdenAliAK21} have varying levels of robustness and sub-Gaussian rates for heavy-tailed settings. 
However, none of their results give computationally efficient estimators for pure DP.
A simultaneous and independent work of~\cite{LiuKO21} demonstrates an interesting connection between resilience~\cite{SteinhardtCV18} and private estimation. 
They exploit this connection to design robust and private algorithms for a variety of settings, including mean estimation, covariance estimation, PCA, and more. 
However, their focus is on providing inefficient algorithms under the constraint of approximate DP, while our goal is to give a framework for efficient algorithms under pure DP.
For a summarized comparison of our work with recent private algorithms, see Table~\ref{tab:comparison}.

\section{Techniques}
\label{sec:techniques}

\subsection{The SoS Exponential Mechanism}
Before we turn to our algorithm for mean estimation, we offer a little more detail on the proofs-to-private-algorithms approach and its core component, which we call the SoS exponential mechanism.
We begin with a review of the exponential mechanism itself.

\paragraph{The Exponential Mechanism}
Consider the general problem of privately selecting one among a set of candidates $\cC$ given a dataset $X$, where the quality of a candidate $x \in \cC$ depends on the dataset $X$.
The candidates $\cC$ could represent many different things depending on the context.
In a statistical setting one may often think of $\cC$ as a class of probability distributions, $X$ as a list of samples from one of those distributions, and the goal is to select the distribution from which $X$ came (up to small error).

To apply the exponential mechanism for this problem, one first finds a \emph{score function} $s(X,x)$ which assigns a (real-valued) score to each dataset-candidate pair, ideally such that $x$'s which are ``good'' for a given $X$ receive high scores.
Given $X$ and a privacy parameter $\e > 0$, the exponential mechanism will sample a random $x$ with probabilities $\Pr(x) \propto \exp( \eps \cdot s(X,x))$.
The output is $\e$-differentially private so long as the score function satisfies a \emph{bounded sensitivity} property: for any pair of neighboring datasets $X,X'$\footnote{We say $X,X'$ are neighboring if they differ on the presence/absence of just one individual.} and for all $x \in \cC$, one has $|s(X,x) - s(X',x)| \leq 1$.

Beyond privacy, one also wants the resulting $x$ to be useful.
While this can also mean different things depending on the context, we will take ``utility'' to mean that $x$ is close to some good $x^*(X)$.
To prove that this happens for the exponential mechanism defined above, one shows:
\begin{enumerate}
    \item High-scoring $x$'s are good: if $s(X,x) \geq 0$, then $\|x - x^*\| \leq \alpha$, for a small $\alpha > 0$.
        (The choice of $0$ is without loss of generality; the mechanism is invariant under additive shifts of $s$.)
    \item Not too few high-scoring $x$'s: $\tfrac{|\{ x \in \cC \, : \, s(X,x) \geq t\}|}{|\cC|} \geq \tfrac 1 {r}$. (If $\cC$ is an infinite set one can replace $|\cdot|$ with some measure of volume.)
\end{enumerate}
Then a simple argument shows that $x$ such that $\|x - x^*\| \leq \alpha$ is selected with probability at least $1 - r \exp(-\e t)$.
In other words, the mechanism selects a good $x$ so long as $t \gg \tfrac{\log r}{\e}$.

The exponential mechanism can lead to computationally inefficient algorithms for two basic reasons.
First, in high-dimensional statistics it is often natural to use score functions which seem hard to compute -- the Tukey depth, for just one example \cite{LiuKKO21, BrownGSUZ21}.
Second, as we will face when we turn to mean estimation, even with score functions that are easy to compute, sampling $x$ with $\Pr(x) \propto \exp( \e \cdot s(X,x))$ may not be computationally tractable.

\subparagraph{Convex programs as score functions}
The SoS exponential mechanism can address both of these sources of intractability for instances of the exponential mechanism where the \emph{proofs} of bounded sensitivity are captured in a powerful formal proof system, the \emph{Sum of Squares proof system} (SoS).
To de-mystify this a little without yet delving into the details of SoS proofs, we can see a high-level picture of the algorithms which use the SoS exponential mechanism.
Ultimately, these algorithms fit into the exponential mechanism framework, but with special score functions.

To wit, suppose that we can arrange for the score function $s(X,x)$ to take the form of a linear optimization problem over a convex set $\cK(X)$ with an additional linear constraint involving $x$:
\[
    s(X,x) = \max_{y \in \cK(X)} \iprod{c,y} \text{ such that } Ay = x
\]
for some matrix $A$ and vector $c$.
As long as $\cK(X)$ admits a computationally efficient separation oracle, $s(X,x)$ will be polynomial-time computable -- this already addresses the first source of potential intractability in the exponential mechanism.
A simple argument shows even more: for each $X$, $s(X,x)$ is actually concave in $x$.
Thus, the distribution $\Pr(x) \propto \exp( \e \cdot s(X,x))$ is log-concave, so we can (usually) sample from it in polynomial time!
This sampling step itself has to be done privately: luckily for us, log-concave sampling is so well understood that private methods for polynomial-time log-concave sampling are known \cite{BassilyST14}.

The restriction that the score function take the form of an optimization problem is not a major one: many useful score functions in high-dimensional statistics, such as the Tukey depth, are naturally expressible in this way.
But how to find a score function which is simultaneously a \emph{convex} optimization problem and maintains both bounded sensitivity and utility?
The SoS method gives us a way to construct such a score function automatically, starting from any score function which is expressible as a (potentially non-convex) optimization problem, and for which the \emph{proofs} of bounded sensitivity and utility are expressible in the SoS proof system.

\subsection{Private Mean Estimation}
We proceed with a high-level description of our algorithm for private mean estimation and its analysis.
One could obtain our algorithm as an instance of proofs-to-private-algorithms, by:
\begin{itemize}
    \item Constructing a certain simple-to-analyze exponential mechanism-based mean estimator.
    \item Writing the simple analysis as a series of SoS proofs and applying (meta)-Theorem~\ref{thm:metatheorem}.
\end{itemize}

However, so that our algorithm can be understood without tackling the SoS exponential mechanism in full generality, we will now give a more concrete description of the algorithm and its analysis.
We give a high level version of this description here, and in the main body of our paper we actually provide a full analysis of the algorithm without appealing to Theorem~\ref{thm:metatheorem}.

Let us recall the setup for our problem.
There is a random variable $X$ on $\R^d$ with $\Cov(X) \preceq I$ and $\|\E X \| \leq R$.
The goal is to estimate $\E X$ from i.i.d.\ copies $X_1,\ldots,X_n$, subject to $\e$-DP.
The promise $\| \E X \| \leq R$ is information-theoretically necessary for pure differential privacy \cite{KarwaV18, BunKSW19}.

Our algorithm has strong guarantees in the presence of adversarially-corrupted samples, and obtains sub-Gaussian confidence intervals given heavy-tailed samples.
In fact, this is a side-effect of the fact that our algorithm uses well-studied SDPs from the robust statistics setting, viewing them through the SoS exponential mechanism lens to obtain privacy.
In particular, folklore adaptations of the analyses in \cite{CherapanamjeriFB19,kothari2017outlier} directly show that our algorithm is robust, so to avoid a proliferation of notation we give the proof in the non-robust setting -- see Remarks~\ref{rem:fine-robustness} and~\ref{rem:coarse-robustness}.

Like prior algorithms for private mean estimation, our algorithm has two phases.
\begin{enumerate}
    \item Reduce $R$ to $\poly(d)$, roughly by finding a large ball containing a large number of samples.
    \item Estimate $\E X$ under the assumption $\| \E X \| \leq \poly(d)$.
\end{enumerate}

We start with the second step, which captures much of our conceptual contribution -- even under the assumption $R \leq \sqrt d$, prior algorithms could not achieve pure differential privacy.
Then we discuss the first step, where prior algorithms require $\Omega(d (\log R + \log(1/\beta)))$ samples.
We give an algorithm with sample complexity $\log R(d + \log(1/\beta))$.
When $\log R \ll d$, our algorithm improves over prior work.
(Information-theoretically, $d \log R + \log(1/\beta)$ is possible.)

\subsubsection{Private mean estimation when $\| \E X \| \leq \poly(d)$}
Let us write $\mu = \E X$.
For simplicity, for now we focus on the case that our goal is to find $\hat{\mu}$ such that $\|\hat{\mu} - \mu\| \leq O(1)$ using $O(d)$ samples.
(We can reduce from the $\alpha$-error case to this one by placing the samples in buckets containing around $1/\alpha^2$ samples and taking sample means within each bucket.)
We will also think of $\e$ as a small constant.
A merit of our approach is that it allows powerful techniques from robust statistics to be used for private algorithm design: in this case, our algorithm draws heavily on a robust mean estimation algorithm due to \cite{CherapanamjeriFB19}, using SoS exponential mechanism to privatize its use of convex programming.

\paragraph{Iterative refinement/gradient descent}
Our algorithm will iteratively refine an initial estimate of the mean (without loss of generality, the origin), producing a series of estimates $\hat{\mu}_0 =0, \hat{\mu}_1,\ldots,\hat{\mu}_T$.
In each step $t$, we will privately find a unit vector $v$ such that $\iprod{v, \mu - \hat{\mu}_t} \geq \Omega(1) \cdot \|\mu - \hat{\mu}_t\|$.
Then we can replace $\hat{\mu}_t$ with $\hat{\mu}_{t+1} = \hat{\mu}_t + r \cdot v$ for some appropriate step size $r > 0$.
By standard reasoning this means that $O(\log d)$ steps suffice to obtain $\|\hat{\mu}_T - \mu\| \leq O(1)$.

This gradient-descent approach introduces only logarithmic overheads into our sample complexity and running time, so now we turn to the heart of our algorithm: privately finding $v$.

\paragraph{Finding private gradients}
Given samples $X_1,\ldots,X_n \in \R^d$ and $\hat{\mu}_t$, we would like to privately select a unit vector $v$ such that $\iprod{v, \hat{\mu}_t - \mu} \geq \Omega(1) \cdot \|\hat{\mu}_t - \hat{\mu}\|$.
We will use the exponential mechanism, for which we need to define a score function.
For this, we are inspired by recent work in robust and heavy-tailed statistics, where the following has become a standard fact \cite{LugosiM19a}:
\begin{fact}[Directions with many outliers are good, informal]
    \label{fact:outliers}
    Suppose that $\|\hat{\mu}_t - \mu\| \gg 1$.
    Then, with high probability over $X_1,\ldots,X_n$, so long as $n \gg d$, there are at least $0.9n$ samples $X_i$ such that $\iprod{X_i - \hat{\mu}_t, \tfrac{\mu - \hat{\mu}_t}{\|\mu - \hat{\mu}_t\|}} \geq \Omega(1) \cdot \|\hat{\mu}_t - \mu\|$.
    Furthermore, for every unit vector $v$ such that
    \[
        \left | \{ i \, : \, \iprod{X_i - \hat{\mu}_t, v} \geq \Omega(1) \cdot \|\hat{\mu}_t - \mu\| \} \right | \geq 0.8n\mcom
    \]
    we have $\iprod{v, \mu - \hat{\mu}_t} \geq \Omega(1) \cdot \|\mu  - \hat{\mu}_t\|$.
\end{fact}

Fact~\ref{fact:outliers} makes a good choice of score function clear: we should use 
\[
    s(X,v) = | \{ i \, : \, \iprod{X_i - \hat{\mu}_t, v} \geq \Omega(1) \cdot \|\hat{\mu}_t - \mu\| \} |\mper
\]
Utility of this score function is captured by Fact~\ref{fact:outliers}, and bounded sensitivity is clear by construction.
(In fact, it is exactly this bounded sensitivity property which has already made Fact~\ref{fact:outliers} so important in robust and heavy-tailed statistics.)
We do not necessarily know the value of $\|\hat{\mu}_t - \mu\|$, but getting a private estimate of this quantity is not too difficult -- we privatize a procedure due to \cite{CherapanamjeriFB19}.

A straightforward analysis shows that, since the volume of the $d$-dimensional unit ball is roughly $\exp(d)$, exponential mechanism with this score function will $\e$-privately select a good $v$ so long as $n \gg d/\e$, which is exactly the sample complexity we expect for estimating $\mu$ to error $O(1)$.

\paragraph{Finding private gradients in polynomial time}
Of course, the key problem is that sampling with $\Pr(v) \propto \exp(\e \cdot s(X,v))$ may not be possible in polynomial time.
Indeed, a seemingly-easier problem, finding the highest-scoring $v$, seems closely related to computing Tukey depth, which is NP-hard.

However, some hope comes again from robust/heavy-tailed statistics, where approximation algorithms (often based on semidefinite programming) for the problem of finding the highest scoring $v$ have become an invaluable tool.
Our starting point is the by-now standard construction of such a semidefinite relaxation, which we briefly review.
First, we write this optimization problem as a degree-$2$ polynomial optimization problem in variables $v_1,\ldots,v_d$ and $b_1,\ldots,b_n$, where the latter are constrained so as to be $0/1$ indicators of $\iprod{X_i - \hat{\mu}_t, v} \geq r$ for some threshold value $r$:
\[
    \max \sum_{b=1}^n b_i \text{ such that } \|v\|^2 \leq 1, b_i^2 = b_i, \text{ and } b_i \iprod{X_i - \hat{\mu}_t, v} \geq b_i \cdot r \text{ for all $i$}\mper
\]

The degree-$2$ SoS relaxation of this optimization problem optimizes over (degree 2) \emph{pseudoexpectations}, which are linear functionals $\pE \, : \, \R[b,v]_{\leq 2} \rightarrow \R$ defined on degree at most $2$ polynomials in $b,v$ which are normalized and positive: $\pE 1 = 1$ and $\pE p(b,v)^2 \geq 0$ for all linear $p$.
Concretely:
\begin{align}
    \label{eq:intro-sos}
    \max_{\pE} \pE \sum_{i \leq n} b_i \text{ s.t. } \pE \|v\|^2 \leq 1, \pE b_i^2 = \pE b_i, \text{ and } \pE b_i \iprod{X_i - \hat{\mu}_t, v} \geq r \cdot \pE b_i \text{ for all $i$.}
\end{align}
This optimization problem can be solved via semidefinite programming -- the resulting (equivalent) SDP optimizes over $(1 + n + d) \times (1 + n + d)$ block matrices
\[
    \max \Tr B \text{ such that } \left ( \begin{matrix} 1 & \tilde{b}^\top & \tilde{v}^\top \\ \tilde{b} & B & W^\top \\ \tilde{v} & W & V \end{matrix} \right ) \succeq 0, \Tr V \leq 1, B_{ii} = \tilde{b}_i,\text{and } \iprod{X_i - \hat{\mu}_t, W_i} \geq r \cdot B_{ii}
\]
Here, $V$ is a proxy for the rank-one matrix $vv^\top$ and similarly for $B$ and $bb^\top$.
This SDP is known to be a good approximation to the problem of finding the maximum-score $v$.

We prove the following fact.
(As an aside, this is where SoS proofs enter the picture: establishing facts such as the below for SoS SDP relaxations in general requires constructing SoS proofs.)

\begin{fact}[Bounded Sensitivity and Utility for \eqref{eq:intro-sos}, informal]
    \label{fact:sensitivity-utility-intro}
    \textbf{Bounded Sensitivity:} Changing a single sample in optimization problem \eqref{eq:intro-sos} can change the objective value by at most $1$.
    \textbf{Utility:} With high probability over $X_1,\ldots,X_n$, if $\|\hat{\mu}_t - \mu\| \gg 1$, then any feasible $\pE$ in \eqref{eq:intro-sos} with objective value at least $0.8n$ satisfies $\iprod{\pE v, \mu - \hat{\mu}_t} \geq \Omega(1) \cdot \|\mu - \hat{\mu}_t\|$.
\end{fact}

While this proof largely adapts similar arguments in the robust statistics literature, there is a key technical innovation.
Prior algorithms employing the SDP described above use a nontrivial rounding step to extract a good vector $v$ from the $d \times d$ PSD matrix $V$.
However, for reasons we discuss below, to use the SoS exponential mechanism, it is important that $v$ can be read directly off of the SDP (more precisely, that the rounding algorithm used is linear).
This means we need a stronger rounding procedure than that used in prior works -- we are able to show that $\tilde{v} = \pE v$, a simple linear function of $\pE$, is a good choice.

\subparagraph{Sampling with convex programs}
The ``utility'' part Fact~\ref{fact:sensitivity-utility-intro} solves the problem of finding a high-scoring $v$ in polynomial time, via semidefinite programming.
But we want to find such a high-scoring $v$ privately.

The key idea is to use the SDP to construct a score function: \emph{convexity of the set of feasible solutions and linearity of the objective function now imply log-concavity of the resulting sampling problem}!
This observation, while simple, is remarkably powerful: it allows us to employ a well-studied SDP from robust statistics nearly out-of-the-box to obtain strong privacy guarantees.

Ultimately, we employ the score function:
\[
    s(X,v_0) = \max_{\pE} \pE \sum_{i \leq n} b_i \text{ s.t. } \pE \|v\|^2 \leq 1, \pE b_i^2 = \pE b_i, \text{ and } \pE b_i \iprod{X_i - \hat{\mu}_t, v} \geq r \cdot \pE b_i \text{ for all $i$} \text{ and } \pE v = v_0\mcom
\]
together with the private sampling algorithm of \cite{BassilyST14}, to sample from $\Pr(v_0) \propto \exp(\e \cdot s(X,v_0))$.

We make a few remarks on the precise way that $s(X,v_0)$ depends on $v_0$; that is, via the constraint $\pE v = v_0$, since this choice is not accidental.
First of all, it is important that the constraints of the optimization problem defining $s(X,v_0)$ depend \emph{linearly} on $v_0$; otherwise we might not retain log-concavity of the resulting sampling problem.
We can do this only because the rounding algorithm described in Fact~\ref{fact:sensitivity-utility-intro}, which proves that $\pE v$ itself is useful, is a simple linear function of $\pE$.

We also note an important interplay between the ``lifted'' nature of the SDP and the utility analysis of the exponential mechanism.
On the one hand, the power of the SDP comes from lifting from the $d + n$ variables $v,b$ to $(d+n+1)^2$ variables, and solving a convex problem in the lifted space.
On the other hand, for the exponential mechanism to satisfy utility with just $O(d)$ samples, it is important that we use it to sample in $d$ dimensions (where, in particular, there is a $2^{O(d)}$-sized cover) rather than, say, $d^2$ dimensions.
This tension is resolved by hiding the additional variables inside the optimization problem which defines $s$, so that exponential mechanism still samples from a $d$-dimensional distribution, but we can still use the power of SoS and semidefinite programming.

\subparagraph{Lipschitzness}
The sampling algorithm of \cite{BassilyST14} (like other algorithms for sampling from log-concave probability distributions) runs in polynomial time in the ambient dimension, \emph{so long as the distribution has Lipschitz log-probabilities}.
Natural approaches to force $s(X,v_0)$ to be Lipschitz, which generally take the form of randomized smoothing, risk violating pure DP, because an algorithm computing a randomized smoothing of $s$ will have some small probability of quietly failing, at which point privacy is at risk.

To ensure that the resulting algorithm runs in polynomial time, therefore, we have to show that $s(X,v_0)$ is Lipschitz with respect to $v_0$.
To establish Lipschitzness, we make a two-step argument:
\begin{enumerate}
    \item We show that dual solutions to the SDP defining $s(X,v_0)$ have optimal dual certificates which are not too large in norm. (At most, say, $\poly(d,n)$.)
    \item We show that any such dual solution to the SDP for $s(X,v_0)$ which certifies an upper bound of $c$ can be adapted to a dual solution for $s(X,v_0 + \Delta)$ which certifies an upper bound of $c + \|\Delta\| \poly(d,n)$, for any small perturbation vector $\Delta$.
\end{enumerate}
Together, these imply that $s(X,v_0)$ is $\poly(d,n)$-Lipschitz with respect to $v_0$.
Since the Lipschitz constant only arises in our algorithm's running time, this suffices for our purposes.
This concludes our overview of the second phase of our private mean estimation algorithm.

\subsubsection{Coarse estimation: from $R$ to $\poly(d)$}

We now describe our algorithm to privately localize $\mu = \E X$ to a ball of radius $\poly(d)$, beginning only with the promise that $\| \E X \| \leq R$ for some large number $R$.
The goal is to do so using as few samples as possible.
Existing efficient algorithms \cite{KamathLSU19} can perform this task with probability $1-\beta$ using $O(d(\log R + \log(1/\beta))/\e)$ samples (we present this analysis in our paper to capture the dependence on $\log(1/\beta)$).
If we used this algorithm, we would obtain sub-Gaussian confidence intervals only when $n \gg \tfrac{d \log R + d \log(1/\beta)}{\e}$.
However, if we do not worry about running time, the same task can be accomplished using the exponential mechanism using only $O(d \log R / \e + \log(1/\beta) / \e)$ samples -- much fewer for $\beta \ll 1$, which is important for constructing confidence intervals.

While we do not quite obtain optimal complexity, we are able to improve on existing algorithms in the regime $\log R \ll d$; our algorithm requires $\tilde{O}(d \log R / \e + \log R \log(1/\beta) / \e)$ samples.

Our algorithm again follows the proofs-to-private-algorithms approach.
We start with the following basic instantiation of the exponential mechanism.
We are given samples $X_1,\ldots,X_n$.
With probability at least $1-\beta$ over the choice of $X_1,\ldots,X_n$, as $n \gg \log(1/\beta)$, any ball of radius $\poly(d)$ containing $0.9n$ of the samples will have center which has distance at most $\poly(d)$ to $\mu$.
So, we would like to use the exponential mechanism with score function $s(X,x) = |\{ i \, : \, \|X_i - x\| \leq \poly(d) \} |$ to select a point from the ball of radius $R$ centered at the origin.

As before, it is not clear how to perform the sampling task that this would require in polynomial time.
So, we replace the score function with a convex relaxation, in this case built from the ``degree-$4$'' SoS relaxation of of the following polynomial optimization problem:
\[
    \max_{b,x} \sum_{i=1}^n b_i \text{ s.t. } \|x\|^2 \leq R^2, b_i^2 = b_i, \text{ and } b_i \|X_i - x\|^2 \leq \poly(d) \cdot b_i \text{ for all $i$}\mper
\]
(The degree-$4$ SoS relaxation shows up here because the optimization problem involves polynomials of degree $3$, and SoS relaxations are defined only for even degrees.)
That is, we use the score function
\[
    s(X,x_0) = \max_{\pE} \pE \sum_{i=1}^n b_i \text{ s.t. } \pE \text{ satisfies } \|x\|^2 \leq R^2, b_i^2 = b_i, b_i \|X_i - x\|^2 \leq b_i (R/10)^2, \text{ and } \pE x = x_0
\]
(for the definition of ``satisfies'', see Section~\ref{sec:prelims}).

Once we have decided to use this particular SoS relaxation, the outlines of the algorithm and its analysis are largely similar to the second phase of the algorithm, but with different SoS proofs plugged in.
In particular, we prove an analogue of Fact~\ref{fact:sensitivity-utility-intro} for this setting, and establish Lipschitzness of the SDP, so that we can use the same strategy as in second phase.

The key step is to show that a high-scoring $x_0$ has distance at most $R/2$ to any $x'$ which has distance $\poly(d)$ to $0.8n$ samples.
This statement turns out to have a relatively simple SoS proof.
This shows that SoS exponential mechanism manages to reduce $R$ to $R/2$.
Iterating this $\log R$ times completes the algorithm.

%% file: table.tex
\begin{table}[]
\centering
\begin{tabular}{@{}llllll@{}}
\toprule
  Algorithm & Sample Complexity & Privacy & Sub-Gaussian Rates & Robust & Poly Time\\ \midrule
  Sample Mean & $d$ & \red{None} & \red{No} & \red{No} & Yes \\
  \cite{KamathSU20}-A & $d$ & Pure & \red{-} & Yes & \red{No} \\
  \cite{KamathSU20}-B & $d$ & \red{Concentrated} & \red{No} & \red{No} & Yes \\
  \cite{KamathSU20}-C & \red{$d^{1.5}$} & Pure & \red{No} & \red{No} & Yes \\
  \cite{LiuKKO21}-A & $d$ & \red{Approximate} & \red{-} & Yes & \red{No} \\
  \cite{LiuKKO21}-B & \red{$d^{1.5}$} & \red{Concentrated} & \red{No} & Yes & Yes \\
    Theorem~\ref{thm:main-intro} & $d$ & Pure & $n \gg \tfrac{d \log R + \min(d,\log R) \log(1/\beta)}{\e}$ & Yes & Yes \\
\end{tabular}
    \caption{Comparing private mean estimation algorithms for distributions with bounded covariance. (Some papers contain several algorithms.) Sample complexity column ignores logarithmic factors in $d$. ``-'' indicates that we are not aware of any analysis in the literature.}
\label{tab:comparison}
\end{table}

%% file: prelims.tex
\section{Preliminaries}
\label{sec:prelims}

\subsection{SoS Proofs and Pseudoexpectations}
We give a brief overview of SoS; for details see \cite{barak2016proofs}.
\paragraph{SoS Proofs}
We first informally review the SoS proof system.
Let $p(x),p_1(x),\ldots,p_m(x)$ be multivariate polynomials in indeterminates $x_1,\ldots,x_n$.
The SoS proof system can prove statements of the form:
\[
    \text{For all $x \in \R^n$, if $p_1(x) \geq 0,\ldots,p_m(x) \geq 0$, then $p(x) \geq 0$.}
\]
Polynomials are highly expressive, so a wide range of mathematical statements can be encoded in the above form.
An SoS proof of such a statement is a family of polynomials $\{ q_S(x) \, : \, S \subseteq [m] \}$ such that each $q_S(x) = \sum_{i=1}^r (q_S^{(i)}(x))^2$ is a sum of squares, and $p(x) = \sum_{S \subseteq [m]} q_S(x) \cdot \prod_{i \in S} p_i(x)$. The proof has degree $D \in \N$ if each term in this sum is a polynomial of degree at most $D$.
We write:
\[
    p_1 \geq 0 ,\ldots,p_m \geq 0 \proves_D p \geq 0\mper
\]
Additionally, we will need one non-standard definition: for any given $j \in [m]$, we say that an SoS proof is degree $D_j$ \emph{with respect to} $p_j$ if every term $q_S(x) \cdot \prod_{i \in S} p_i(x)$ in the proof such that $j \in S$ has degree at most $D'$.
We will sometimes write $p_1 \geq 0,\ldots p_m \geq_0 \proves_{D,\deg_{p_j} = D_j} p \geq 0$.

SoS proofs are dual solutions to semidefinite programs arising from the Sum of Squares method, where pseudoexpectations are primal solutions.
As in many applications of convex programming, to prove facts about primal solutions, the main technique is to construct duals.
In particular, to show that the SoS SDPs we use for exponential mechanism score functions satisfy bounded sensitivity and privacy, it suffices to construct SoS proofs witnessing these facts.

\paragraph{Pseudoexpectations}
For even $d \in \N$, a degree-$d$ pseudoexpectation in indeterminates $x = x_1,\ldots,x_n$ is a linear operator $\pE \, : \, \R[x_1,\ldots,x_n]_{\leq d} \rightarrow \R$ which satisfies $\pE 1 = 1$ an $\pE p^2 \geq 0$ for every degree-$d/2$ polynomial $p$. We say that $\pE$ \emph{satisfies} an inequality $p(x) \geq 0$ if for every $q$ such that $\deg(p\cdot q^2 ) \leq d$ we have $\pE p q^2 \geq 0$.

\paragraph{Archimedean systems and duality}
We say that a system of polynomial inequalities $p_1(x) \geq 0,\ldots, p_m \geq 0$ is Archimedean if for some real $M > 0$ it contains the inequality $\|x\|^2 \leq M$.
SoS proofs and pseudoexpectations satisfy a natural duality for Archimedean systems, which we use often.
Namely: for every Archimedian system $p_1,\ldots,p_m$ and every polynomial $f$ and every degree $d$, exactly one of the following holds.
\begin{enumerate}
\item For every $\e > 0$ there is an SoS proof $p_1 \geq 0 ,\ldots,p_m \geq 0 \proves_d f \geq -\e$
\item There is a degree-$d$ pseudoexpectation satisfying $p_1 \geq 0 ,\ldots, p_m \geq 0$ but $\pE f < 0$.
\end{enumerate}

\subsection{Privacy}

We have already seen the definition of (pure) differential privacy.
Our approach relies heavily upon the exponential mechanism of~\cite{McSherryT07}, we employ a volume-based version which appears in~\cite{KapralovT13}.
The proof is standard but we include it for completeness.

\begin{theorem}[volume-based exponential mechanism~\cite{McSherryT07, KapralovT13}]
	\label{thm:volume-based-exponential-mechanism}
	The exponential mechanism $M_E$ on inputs $X, \cH \subset \R^d$, $s$, selects and outputs some object $h \in \cH$, where the probability a particular $h$ is selected is proportional to $\exp\paren{\frac{\eps s\paren{X, h}}{2 \Delta}}$.
	Let $\cH^* \subseteq \cH$ be a set such that, $ \OPT\paren{X} \le \inf_{h \in \cH^*}s\paren{X, h}$ be a lower bound for the score attained by the objects in $\cH^*$ with respect to the dataset $X$. Moreover, let $\vol\paren{S}$ denote the Lebesgue measure of $S$ in $\R^d$.
	Then
	\begin{equation*}
		\Pr\Brac{s\paren{M_E\paren{X}} \le \OPT\paren{X} - \frac{2\Delta}{\epsilon} \Paren{\ln \Paren{\frac{\vol\paren{\cH}}{\vol\paren{\cH^*}} + t}}} \le \exp\paren{-t}.
	\end{equation*}
\end{theorem}
\begin{proof}
	We follow the same argument as the standard exponential mechanism analysis.
	\begin{align*}
		\Pr\brac{s\paren{M_E\paren{X}} \le c} &= \frac{\int_{h: s\paren{X, h} \le c, h \in \cH} \exp\Paren{\frac{\eps s\paren{X, h}}{2 \Delta}} \, \mathrm{d} h}{\int_{h': h' \in \cH} \exp\Paren{\frac{\eps s\paren{X, h'}}{2 \Delta}} \, \mathrm{d} h'} \\
		&\le
		\frac{\vol\paren{\cH} \exp\Paren{\frac{\eps c}{2 \Delta}}}{\vol\paren{\cH^*} \exp\Paren{\frac{\epsilon \OPT\Paren{X}}{2\Delta}}} \\
		&=
		\frac{\vol\paren{\cH}}{\vol\paren{\cH^*}} \exp\Paren{\frac{\epsilon \paren{c - \OPT\paren{X}}}{2 \Delta}}.
	\end{align*}
	From this inequality, the theorem statement can be obtained by substituting in the prescribed value for $c$.
	It remains to explain the first inequality. The numerator can be upper bounded since we are taking the integral at most over $\cH$, and the value of the integral at each point is less than $\exp\paren{\epsilon c/ 2\Delta}$.
	Similarly, the denominator can be lower bounded by considering only the points in $\cH^*$, all of which have a lower bound of $\exp\paren{\epsilon \OPT\paren{X} / 2\Delta}$.
\end{proof}

We will also extensively use the following result of \cite{BassilyST14}.

\begin{theorem}[Lemma 6.5 of \cite{BassilyST14}]
\label{thm:efficient-sampling}
    For every $\e$ and $d \in \N$, there is an $\e$-DP algorithm $A$ with the following guarantees.
    Given access to an evaluation oracle for a concave, $L$-Lipschitz function $f \, : \, \R^d \rightarrow \R$ and to membership and projection oracles for a convex set $\cC \subseteq \R^d$, the algorithm produces a sample from a distribution $D$ such that for every (measurable) $S \subseteq \cC$,
    \[
        e^{-\e} \Pr(A \in S) \leq \frac{\int_S \exp(f)}{\int_{\cC} \exp(f)} \leq e^{\e} \Pr(A \in S)\mper
    \]
    The algorithm runs in time $\poly(d,L\diam(\cC),1/\e, \log \diam \paren{C})$, making at most that many queries to the oracles.
\end{theorem}
\cite{BassilyST14} actually provides the running time bound $\poly(d,L,\diam(\cC),1/\e)$, but we can deduce the more precise bound in the theorem statement by a simple scaling argument.
We note that recent work of~\cite{MangoubiV21} provides an algorithm which improves the running time's dependence on $1/\e$ from polynomial to polylogarithmic, as well as reducing the polynomial dependence on $d$ for convex sets $\cC$ that contain a ball of radius $r$, at an additional cost of $\poly \log \paren{1/r}$ (Remark 2.5 of \cite{MangoubiV21}). 
As our main focus is on providing polynomial time algorithms and not on designing the fastest algorithms, we do not further employ their improved methods.

%% file: metatheorem.tex
\section{Meta-Theorem on SoS Exponential Mechanism}
\label{sec:metatheorem}
To describe the SoS exponential mechanism more completely, we prove a meta-theorem on its performance.
This meta-theorem could be extended in various ways, but the version we give here captures the mechanism as it is used in our paper.
The reader interested solely in our result on mean estimation may comfortably skip this section as we do not rely on it elsewhere, but the exposition and abstraction here may make it easier to follow.

\paragraph{Meta-theorem Setup}
Consider the following template to capture the exponential mechanism in the language of polynomials.
Let $\cX$ be a universe of possible datasets and $\cC \subseteq \R^n$ be a set of candidates.
Consider a score function of the following form.
For each dataset $X$, suppose there is a family of polynomials $p^X,p_1^X,\ldots,p_m^X$ in variables $x_1,\ldots,x_n,y_1,\ldots,y_N$, and the score function $s(X,x)$ is given by
\[
    s(X,x) = \max_{y} p^X(x,y) \text{ such that } p_1^X(x,y) \geq 0,\ldots,p_m^X(x,y) \geq 0\mper
\]

\begin{example}[Tukey depth with margin]
    \label{ex:tukey}
  To make things a bit less abstract, let us see that a score function closely related to the Tukey depth is expressible in this form.
    Suppose $X = X_1,\ldots,X_N \in \R^n$ is a dataset.
    Recall that the Tukey depth of a point $z \in \R^n$ is given by the maximum number of $X_i$'s which lie on one side of a hyperplane through $z$.

    Later on, for a fixed $z \in \R^n$ with $\|z\| = 1$, we will want to run the exponential mechanism to select a direction $v$ such that many $X_i$'s lie above a hyperplane through $z$ in direction $v$, at distance at least distance $r > 0$ from that hyperplane, so our score function for $v$ is the number of such $X_i$'s.
We can express this as follows:
    \begin{align}
        \label{eq:tukey}
        s(X,v) = \max_{\substack{b_1,\ldots,b_N}} \sum_{i=1}^N b_i \text{ s.t.\ for all $i$, } b_i^2 = b_i \text{ and } b_i \iprod{X_i - z,v} \geq b_i \cdot r \mper
    \end{align}

    Note that for fixed $z$ and $v$ it is easy to compute the number of $X_i$'s lying above the resulting hyperplane, but finding the maximizing $v$ already appears to be a hard problem.
    (Of course, we want to accomplish only a related task: \emph{sampling} from a distribution on high-scoring $v$'s.)
    Indeed, approximation algorithms for finding such $v$ play a key role in recent algorithmic advances in robust and heavy-tailed statistics statistics.
    \end{example}

\subparagraph{Utility and Bounded Sensitivity in the Language of Polynomials}
Utility in this framework takes exactly the same form as in the usual exponential mechanism -- to demonstrate utility, one would show that high-scoring $x$'s are good, and there are not too many low-scoring $x$'s.
The first of these statements is naturally expressible in the SoS proof system, and it turns out that the latter will not need to be expressed as an SoS proof.
There is just one technical subtlety: the requirement that $x \in \cC$, if it is used in the proof of utility, must be captured by the polynomials $p_1(x,y) \geq 0, \ldots, p_m(x,y) \geq 0$, as the SoS proof system has no other way of natively using the hypothesis that $x \in \cC$.
(We illustrate this in the example below.)

To formulate bounded sensitivity within SoS will take a little additional work -- SoS will require for there to be a certain kind of witness to bounded sensitivity.
This is not a major restriction, as natural proofs of bounded sensitivity for optimization-based score functions typically yield such witnesses anyway.

Making this concrete, a natural way to show that a score function like the above satisfies bounded sensitivity is to relate feasible solutions to the optimization problem for $X$ to those for a neighboring dataset $X'$, without losing too much in the objective value.
Let us suppose that for each neighboring pair $X,X'$ there is a transformation $y'(y)$ such that for all $x$, if $y$ is feasible for $X$ (i.e. $p_1^X(x,y) \geq 0,\ldots,p_m^X(x,y) \geq 0$) then $y'(y)$ is feasible for $X'$ (i.e. $p_1^{X'}(x,y'(y)) \geq 0,\ldots, p_m^{X'}(x,y'(y)) \geq 0$).
If additionally $p^X(x,y) - p^{X'}(x,y'(y)) \leq 1$, this transformation (together with the corresponding one mapping $y'$s to $y$s) witnesses bounded sensitivity for $s$.
For technical reasons, our meta-theorem imposes the restriction that $y'$ is a linear function of $y$, but this still suffices for our algorithms.

\begin{example}[Continuation of Example~\ref{ex:tukey}]
    \label{ex:tukey-2}
    First addressing utility: it turns out that the proof of utility for $v$ having high score according to \eqref{eq:tukey} will rely on $\|v\| \leq 1$.
    So we will have to strengthen our system of polynomials to include this constraint.
    As a technicality, we will also shift our objective function so that having positive score is good enough for utility.
    \begin{align}
        \label{eq:tukey-2}
         \max_{b,v} \sum_{i=1}^N b_i - 0.9N \text{ such that } b_i^2 = b_i, b_i \iprod{X_i - z,v} \geq b_i \cdot r, \|v\|^2 \leq 1\mper
    \end{align}

    Now, to establish bounded sensitivity, consider two neighboring datasets $X = X_1,\ldots,X_N$ and $X' = X_1',X_2,\ldots,X_N$, where $X$ and $X'$ differ on the first vector.
    If $(b,v)$ is a feasible solution to \eqref{eq:tukey-2} for $X$ with objective value $t$, then we can replace it with $(0,b_2,\ldots,b_N,v)$ to get a feasible to solution for $X'$ with objective value at least $t-1$.
    The meta-theorem requires there to be an SoS proof of this fact; this (easy) SoS proof was first established in \cite{Hopkins20}.
\end{example}

\subparagraph{Robustly Satisfiable Polynomials}
Before stating our meta-theorem, we need one more technical definition, capturing a certain well-conditioned-ness property of a polynomial optimization problem.
Ultimately, this condition will imply a Lipschitz property of semidefinite relaxations of that optimization problem.
This Lipschitz property will be used, in turn, to bound the running time of MCMC-based samplers for probability densities using those semidefinite relaxations as log-probabilities.

The details of the following definition (Definition~\ref{def:cond-poly}) may appear opaque, but they are not too important -- a good intuitive interpretation is that the polynomial optimization problem
\[
    \max_{x,y} p(x,y) \text{ s.t. } p_1(x,y) \geq 0, \ldots, p_m(x,y) \geq 0
\]
has a robust space of feasible solutions.
Roughly, this means that for any $x$ there is a small ball $B$ around $x$ such that for any $x' \in B$ there is a feasible solution $(x',y)$ whose objective value isn't too large.
(The actual condition we use is slightly weaker than this.)
This type of condition is common in meta-theorems involving the SoS proof system, to rule out the use of pathological $p,p_1,\ldots,p_m$, and it is typically not too difficult to establish -- see, e.g. \cite{hopkins2017power,weitz2017polynomial}.

\begin{definition}[Robustly satisfiable polynomial systems]
  \label{def:cond-poly}
    Let $\cC \subseteq \R^n$ and let $p,p_1,\ldots,p_m$ be polynomials in $x = x_1,\ldots,x_n$ and $y = y_1,\ldots,y_N$.
    Let $\eta > 0$.
    Consider a family of optimization problems, one for each $x \in \cC$, given by
    \[
        \max_{y} p(x,y) \text{ s.t. } p_1(x,y) \geq 0,\ldots,p_m(x,y) \geq 0\mper
    \]
    We say this family is $\eta$-robustly satisfiable if, for each $x \in \cC$, each $x'$ in the ball of radius $\eta$ around $x$ can satisfy the constraints.
    That is, for each $x'$ such that $\|x' - x\| \leq \eta$, there exists $y$ such that $p_1(x',y) \geq 0, \ldots, p_m(x',y) \geq 0$.
\end{definition}

Note that in our algorithms, $\eta$ will factor only into running times and not sample complexity or accuracy guarantees, so rather coarse bounds on $\eta$, perhaps loose by polynomial factors, suffice for our purposes.

\begin{example}[Continuation of Examples~\ref{ex:tukey},~\ref{ex:tukey-2}]
    Let us imagine now that $\cC$ is a ball of radius $0.9$ centered at the origin, to see a proof sketch of $\eta$-robust satisfiability for \eqref{eq:tukey-2}.
    For each $v \in \cC$ and every vector $\Delta$ with $\|\Delta\| \leq 0.1$, the constraints of \eqref{eq:tukey-2} are satisfied by $(v+\Delta,0)$.
    So, \eqref{eq:tukey-2} is $\eta$-well-conditioned for $\eta =0.1$, with respect to $\cC$.
\end{example}

With this setup in hand, we can state our meta-theorem.

\begin{theorem}[Meta-Theorem on SoS Exponential Mechanism]
    \label{thm:metatheorem}
  Let $\cC \subseteq \R^n$ be a compact, convex set and $\cX$ a universe of possible datasets, equipped with a ``neighbors'' relation.
    Suppose that for every dataset $X$ there exists an Archimedean and $\eta$-robustly satisfiable system of polynomial inequalities $\cP^X(x,y) = \{p_1^X(x,y) \geq 0,\ldots,p_N^X(x,y) \geq 0\}$ and a polynomial $p^X(x,y)$, all of degree at most $D$, in indeterminates $x_1,\ldots,x_n,y_1,\ldots,y_N$ such that for every neighboring dataset $X'$ there is a linear function $y'(y)$ such that bounded sensitivity has an SoS proof:
\[
    \forall j, \, \cP^X(x,y) \proves_{\deg(p_j^{X'})} p_j^{X'}(x,y'(y)) \ge 0 \text{ and } \cP^X(x,y) \proves_{D} p^X(x,y) - p^{X'}(x,y') \leq 1\mper
\]
Suppose also that for every $X$, there are SoS proofs $\cP^X(x,y) \proves_D p(x,y) \leq 1/\eta$ and $\cP^X(x,y) \proves_D -p(x,y) \leq 1/\eta$.
Furthermore, suppose that the polynomials $\cP^X$ and $p^X$, and the polynomials used in the above SoS proofs, all have coefficients expressible in at most $B$ bits.

Then for every $\e > 0$ and $D \in \N$ there exists an $\e$-differentially private algorithm which takes as input the polynomials $p^X,p_1^X,\ldots,p_m^X$ and $B,\eta > 0$, with the following guarantees:

    \textbf{Utility:} For every $X$, if there is an SoS proof of utility for $X$ which is degree-$\deg(p)$ with respect to $p$, i.e.,
\[
    \cP^X(x,y) \cup \{ p(x,y) \geq 0 \} \proves_{D,\deg_{p} = \deg(p)} \|x - x^*(X)\|^2 \leq \alpha^2
\]
for some vector $x^*(X) \in \R^n$ and $\alpha > 0$, where the coefficients of all polynomials involved in the proof are expressible with $B$ bits, and if
\[
  \frac{ \vol(\cC)} {\vol(\{ x \in \cC \, : \, \exists y \text{ s.t. } \cP^X(x,y) \text{ and } p^X(x,y) \geq t \} )} \leq r \mcom
\]
    then the algorithm outputs $x$ such that $\|x - x^*(X)\| \leq \alpha + 2^{-B}$ with probability at least $1 - r \exp(-\Omega(\e t))$.

    \textbf{Running time:} The algorithm runs in time
    \[
        \poly \Paren{ n^{D}, N^{D}, m^{D}, \frac 1 {\e}, \frac 1 {\eta}, \diam(\cC), B }\mcom
    \]
    making at most this many calls to membership and projection oracles for $\cC$.
\end{theorem}

\subsection{Proof of Theorem~\ref{thm:metatheorem}}

We describe the algorithm we use to prove Theorem~\ref{thm:metatheorem}; then we assemble the lemmas we need for the analysis.

\paragraph{SoSExponentialMechanism} Input: polynomials $p^X,p_1^X,\ldots,p_m^X$, $D \in \N$, $\eta > 0$, $B \in \N$.
\begin{enumerate}
    \item For $x_0 \in \cC$, let
    \[
        s(X,x_0) = \max_{\pE} \pE p^X(x,y) \text{ such that } \deg \pE = D, \pE \text{ satisfies } \cP^X, \text{ and } \pE x = x_0
    \]
where the optimization is over $\pE$ in indeterminates $x,y$.
    \item Let $B' = B + T(\diam{\cC}, 1/\eta, 1/\e, d)$, where $T$ is a sufficiently-large polynomial in the running time of the log-concave private sampler of \cite{BassilyST14}, when run with finite-precision arithmetic.
    \item Run the log-concave private sampling algorithm of \cite{BassilyST14} (Lemma 6.5) with score function $s$, Lipschitz parameter $\poly(1/\eta)$, and privacy parameter $\e/4$.
        Whenever the sampling algorithm makes a call to $s(X,x_0)$, solve the underlying SDP to $\poly(B')$ bits of precision.
\end{enumerate}

Theorem~\ref{thm:metatheorem} is immediate from the following lemmas, all of which we establish in the next section.

\begin{lemma}[High-scoring $x_0$ is found in polynomial time]
    \label{lem:meta-score}
    Given the setup of Theorem~\ref{thm:metatheorem}, for all $X$, if
    \[
      \frac{ \vol(\cC)} {\vol(\{ x \in \cC \, : \, \exists y \text{ s.t. } \cP^X(x,y) \text{ and } p^X(x,y) \geq t \} )} \leq r \mcom
    \]
    then with probability at least $1-r \exp(- \e (t/2-1) )$, the $x_0$ output in step (2) of SoSExponentialMechanism has $s(X,x_0) \geq 0$.
\end{lemma}

\begin{lemma}[High-scoring $x_0$ is useful]
    \label{lem:meta-util}
    Under the assumptions of Theorem~\ref{thm:metatheorem}, for all $X$, if there is an SoS proof of utility for $X$ as described in Theorem~\ref{thm:metatheorem}, then for all $x_0$ such that $s(X,x_0) \geq 0$, $\|x^*(X) - x_0\| \leq \alpha + 2^{-B}$.
\end{lemma}

\begin{lemma}[Privacy]
    \label{lem:meta-priv}
    Under the assumptions of Theorem~\ref{thm:metatheorem}, SoSExponentialMechanism satisfies $\e$-DP.
\end{lemma}

\subsection{Proofs of Lemmas}

We will prove Lemmas~\ref{lem:meta-score}, \ref{lem:meta-util}, and~\ref{lem:meta-priv}.

\textbf{Remark on numerical issues:} Because of the choice of $B'$, the guarantees of the log-concave sampling algorithm will apply equally well if it receives $s(X,x_0) \pm 2^{-B'}$ as if it receives $s(X,x_0)$ when making oracle calls to $s$.
(Given its running time, it cannot even read enough bits to tell the difference.)
So, we will henceforth ignore the difference and presume that the log-concave sampler observes the values $s(X,x_0)$ exactly.

The first step is to establish that the target probability distribution is actually log-concave and Lipschitz, so that we can use the guarantees of \cite{BassilyST14}.

\begin{lemma}
    \label{lem:meta-concave}
    For all $X$, the function $s(X,x_0)$ is concave in $x_0$.
\end{lemma}
\begin{proof}
    Consider $x_0$ and $x_0'$ and let $\pE$ be the optimal solution to the optimization problem defining $s(X,x_0)$ and similarly for $\pE'$ and $s(X,x_0')$.
    Then $\tfrac 12 \pE + \tfrac 12 \pE'$ is feasible for $s(X,\tfrac 12 x_0 + \tfrac 12 x_0')$.
    So $s(X,\tfrac 12 x_0 + \tfrac 12 x_0') \geq \tfrac 12 s(X, x_0) + \tfrac 12 s(X, x_0')$.
\end{proof}

\begin{lemma}[Robustly satisfiable systems yield Lipschitz SDPs]
    \label{lem:meta-lipschitz}
    Given the setup in Theorem~\ref{thm:metatheorem}, for all $x_0,x_0' \in \cC$ and all $X$, we have $|s(X,x_0) - s(X,x_0')| \leq \poly(1/\eta)\cdot \|x_0 - x_0'\|$.
\end{lemma}
\begin{proof}
    As shorthand, let us write $s = s(X,x_0)$.
    Since $\cP^X$ is Archimedean, we can apply standard pseudoexpectation/SoS proof duality to conclude that for every $\e > 0$ there is a polynomial identity in variables $x,y$:
    \[
        s + \e - p^X(x,y) = \sum_{S \subseteq [m]} q_S(x,y) \prod_{i \in S} p_i^X(x,y) + \iprod{\lambda, x - x_0}\mcom
    \]
    where $q_S$ are SoS polynomials, all the terms above have degree at most $D$, and $\lambda \in \R^n$ is a vector.
    
    Our first goal is to bound $\|\lambda\|$.
    By robust satisfiability, if we let $x' = x_0 + \eta \cdot \tfrac{\lambda}{\|\lambda\|}$, there exists $y'$ such that $p_i^X(x',y') \geq 0$ for all $i$.
    Hence,
    \[
        s + \e - p^X(x',y') \geq \eta \|\lambda\|\mper
    \]
    Since there are SoS proofs that $p^X(x,y) \leq 1/\eta$ and $-p^X(x,y) \leq 1/\eta$, and we can take $|\e| \leq 1/\eta$, the left-hand side is $O(1/\eta)$, so we find $\|\lambda\| \leq O(1/\eta^2)$.

    We claim that $s(X,x_0') \leq s + O(1/\eta^2) \cdot \|x_0 - x_0'\|$.
    To see this, note that for each $\e > 0$, we can write
    \[
        s + \e + \iprod{\lambda, x_0 - x_0'} - p^X(x,y) = \sum_{S \subseteq [m]} q_S(x,y) \prod_{i \in S} p_i^X(x,y) + \iprod{\lambda, x - x_0'}\mcom
    \]
    which, after Cauchy-Schwarz, certifies the upper bound $s + \e + \|\lambda\| \|x_0 - x_0\|$ on $s(X,x_0')$.
    Since this works for all $\e > 0$, we find $s(X,x_0') \leq s + O(1/\eta^2) \|x_0 - x_0'\|$.
\end{proof}

As a corollary of Lemmas~\ref{lem:meta-concave} and~\ref{lem:meta-lipschitz}, combined with Lemma 6.5 of \cite{BassilyST14}, we obtain:

\begin{corollary}
    \label{cor:meta-sampling-correctness}
    Given the setup of Theorem~\ref{thm:metatheorem}, for every $X$, the output of step (3) of SoSExponentialMechanism is a sample from a distribution $D_X$ supported on $\cC$ such that for every event $A$
    \[
        e^{-\e/2} \Pr_{D_X^{\text{target}}}(A) \leq \Pr_{D_X}(A) \leq e^{\e/2} \Pr_{D_X^{\text{target}}}(A) 
    \]
    where $D_X^{\text{target}}$ is the distribution with density proportional to $\exp((\e/2) s(X,x_0))$
    Furthermore, the sampler runs in time at most $\poly(d, \tfrac 1 \e, \tfrac 1 \eta, \diam(\cC))$, making at most that many calls to a membership oracle for $\cC$ and to an evaluation oracle for $s(X,\cdot)$.
\end{corollary}

\subsubsection{Privacy: proof of Lemma~\ref{lem:meta-priv}}

\begin{lemma}
    \label{lem:meta-sensitivity}
    Given the conditions of Theorem~\ref{thm:metatheorem}, the score function $s$ has sensitivity at most $1$.
\end{lemma}
\begin{proof}
    Let $X,X'$ be neighboring datasets, and let $\pE$ be an optimal solution to the optimization problem defining $s(X,x_0)$.
    We will construct a feasible solution $\pE'$ for $s(X',x_0)$ whose objective value is at most $s(X,x_0) + 1$.
    Since we could swap $X$ and $X'$, this will prove that $|s(X,x_0) - s(X',x_0)| \leq 1$.

    For any degree $D$ polynomial $f$, we define $\pE' f(x,y) = \pE f(x,y'(y))$.
    Note that the degree of $\pE'$ as a pseudoexpectation is the same as that of $\pE$, because $y'(y)$ is linear.

    We claim that $\pE'$ is feasible for $s(X',x_0)$.
    Clearly $\pE' x = x_0$, so we just need to check that $\pE'$ satisfies $\cP^{X'}$.
    For each $j$, we check that $\pE'$ satisfies $p_j^{X'}(x,y) \geq 0$.
    Consider any square polynomial $q$ such that $\deg (q \cdot p^{X'}_j) \leq D$.
    We need to show $\pE' q \cdot p^{X'}_j \geq 0$.

    Using the SoS proof $\cP^X(x,y) \proves_{\deg(p_j^{X'})} p_j^{X'}(x,y'(y))$, we can write
    \[
        q(x,y') \cdot p^{X'}_j(x,y') = q(x,y') \cdot \sum_{S \subseteq [m]} q_S(x,y) \prod_{i \in S} p_i^{X}(x,y)
    \]
    where every term in the sum on the right-hand side has degree at most $\deg p_j^{X'}$.
    Therefore, for every $S$, we have $q(x,y'(y)) \cdot q_S(x,y) \prod_{i \in S} p_i^X(x,y)$ has degree at most $D$.
    So, applying $\pE'$ to both sides, we find that $\pE' q p^{X'}_j \geq 0$, using that $\pE$ satisfies $\cP^X$.
    
    Finally, we have to check that $\pE p^X(x,y) - \pE' p^{X'}(x,y) \leq 1$.
    We expand the definitions and use our SoS proof of bounded sensitivity.
    \begin{align*}
        \pE p^X(x,y) - \pE' p^{X'}(x,y) = \pE p^X(x,y) - \pE p^{X'}(x,y') \leq 1\mper
    \end{align*}
\qedhere
\end{proof}

\begin{proof}[Proof of Lemma~\ref{lem:meta-priv}]
    By Lemma~\ref{lem:meta-sensitivity} and the usual analysis of the exponential mechanism, an output from the target distribution $D_X^{\text{target}}$ from Corollary~\ref{cor:meta-sampling-correctness} would satisfy $\e/2$-DP.
    Since, according to Corollary~\ref{cor:meta-sampling-correctness}, the actual distribution output by step (3) of the algorithm differs only by multiplicative $e^{\e/2}$, the output of the log-concave sampler satisfies $\e$-DP.
\end{proof}

\subsubsection{Utility: proofs of Lemmas~\ref{lem:meta-score} and~\ref{lem:meta-util}}

\begin{proof}[Proof of Lemma~\ref{lem:meta-score}]
    By the standard analysis of the exponential mechanism, a sample from $D_X^{\text{target}}$ (as described in Corollary~\ref{cor:meta-sampling-correctness}) would output $x_0$ with $s(X,x_0)$ with probability at least $1-r \exp(-\e t/2)$.
    Since the actual output distribution of step (3) is $e^{\e/2}$ multiplicatively close to this, we are done.
\end{proof}

\begin{proof}[Proof of Lemma~\ref{lem:meta-util}]
    Since $x_0$ has $s(X,x_0) \geq 0$, there exists $\pE$ of degree $D$ satisfying $p_1^X(x,y) \geq 0,\ldots,p_m^X(x,y) \geq 0$ and having $\pE x = x_0$ and $\pE p^X(x,y) \geq 2^{-\poly(B)}$.
    Then applying $\pE$ to either side of the SoS proof of utility, we obtain $\pE \|x - x^*(X)\|^2 \leq \alpha^2 + 2^{-\poly(B)}$.
    By convexity, $\|\pE x - x^*(X) \| \leq \alpha + 2^{-\poly(B)}$, so we are done.
\end{proof}

%% file: coarse-estimation.tex
\newcommand{\Score}{\textsc{Score}}
\section{Coarse Estimation}
\label{sec:coarse}
Our overall algorithm for private mean estimation proceeds in two phases.
At start of the algorithm, we assume that we have a point $\mu_0$ such that $\| \mu - \mu_0\|_2 \leq R$.
Equivalently, by re-centering, we assume that $\|\mu\|_2 \leq R$.
Recall that packing lower bounds for differential privacy generally necessitate an assumption of this nature~\cite{HardtT10}. 

In the first phase of the algorithm we perform coarse mean estimation, improving our initial estimate of the mean from $R$ to $\cO(\sqrt{d})$ error in $\ell_2$-norm.
In the second phase, we perform fine mean estimation, further improving our estimate from $\cO(\sqrt{d})$ to $\alpha$ error.
The overall algorithm will follow by combining these two phases. 
In this section, we describe our algorithm for coarse mean estimation, quantified in Theorem~\ref{thm:coarse-estimation}.

\begin{theorem}[Combination of \Cref{thm:simple-high-dimensional-coarse-mean-estimation} and \Cref{thm:sos-coarse-estimate}]
\label{thm:coarse-estimation}
Suppose $\cD$ is a $d$-dimensional distribution with mean $\mu$ and covariance $\Sigma$, such that $\normt{\mu} \le R$, $\normt{\Sigma} \le 1$. Then there exists a polynomial-time $\epsilon$-DP algorithm that takes $n$ \iid\ samples from $\cD$ and outputs $y$ such that $\normt{y - \mu} \le \cO\paren{\sqrt{d}}$ with probability $1-\beta$, and sample complexity 
\begin{equation*}
n = \tilde{\cO}\Paren{\frac{d\log R + \min \Paren{ d,  \log R } \log \paren{1/\beta} }{\epsilon}} \mper
\end{equation*}
\end{theorem}
To demystify the sample complexity in the above theorem, note that in order to  prove this theorem, we present two algorithms with different sample complexities. The above sample complexity is the minimum of the two.

Our first algorithm uses a rather simple approach: privately estimate the mean in each coordinate up to constant error (using, e.g., the exponential mechanism, though one could also employ private histograms as in~\cite{KarwaV18}), and use composition over the $d$ coordinates. 
The sample complexity of this approach is $\cO(d\paren{\log R +  \log1/\beta + \log d}/\epsilon)$ (see \Cref{thm:simple-high-dimensional-coarse-mean-estimation}).
Similar methods have been previously employed for multivariate mean estimation (see, e.g.,~\cite{KamathLSU19}), but we provide an analysis for completeness, and include a particular focus on the error probability $\beta$.

We observe that this coordinate-wise approach incurs a cost of $d \log (1/\beta)$, whereas the primary goal for heavy-tailed mean estimation in the non-private setting is to decouple these two factors.
To this end, we give a second algorithm, which uses our SoS exponential mechanism paradigm.
This algorithm has sample complexity $ \cO(\log R \paren{d+  \log1/\beta + \log \log R}/\epsilon)$ (see \Cref{thm:sos-coarse-estimate}). 
In the settings where $\log R \ll d$, this algorithm yields a better sample complexity than the simpler approach.
An interesting open question is to design an efficient algorithm for this setting with sample complexity $\tilde \cO((d + \log R + \log (1/\beta))/\epsilon)$; an appropriate application of the exponential mechanism provides an inefficient algorithm with this sample complexity.

\subsection{First Algorithm: Coordinate-Wise}
\label{sec:coordinatewise}
\subsubsection{Main Steps}
As the name suggests, our coordinate-wise algorithm works by running a one-dimensional algorithm on each coordinate of the samples from a $d$-dimensional distribution. 
The one-dimensional algorithm exploits the fact that the true mean is within constant distance of $0.99 n$ of the samples with high probability. We use the exponential mechanism to find any point that satisfies the same property. 
By triangle inequality, this new point will be close to the true mean. 

In the following we describe the main theorem of this subsection and the main steps of proving it.
\begin{lemma}
\label{lem:1d-appropriate-radius}
Suppose $X_1,\dots, X_n$ are \iid\ samples from a one-dimensional distribution $\cD$ with mean $\mu$ and variance bounded above by $1$. 
  Furthermore, suppose $n \ge 4 \log\paren{1/\beta}/\alpha^2$, and $R^* \ge \sqrt{2/\alpha}$. Then with probability $1-\beta$ there exists a set $G \subseteq [n]$ such that 
$\abs{G} \ge \paren{1-\alpha} n$, and 
\begin{equation*}
	\forall i \in G : \abs{X_i - \mu} \le R^* \mper
\end{equation*}
\end{lemma}
\begin{proof}
This lemma is a standard application of Chebyshev's inequality and Hoeffding's bound.
Suppose the $X_i$'s are \iid\ samples from $\cD$ and suppose $\cD$ has variance $\sigma^2$. Then by Chebyshev's inequality we know that 
\begin{equation*}
\Pr\brac{\Abs{X_i - \mu} \ge R^*} \le \frac{\sigma^2}{{R^*}^2} \mcom
\end{equation*} 
Let $I_i = \mathbbm{1}\brac{\Abs{X_i  - \mu} \ge R^*}$, then $I_i$'s are independent random variables between $0$ and $1$. Therefore, we may use Hoeffiding's bound. Let $S_n = \sum_{i=1}^n I_i$, then
\begin{equation*}
\Pr\Brac{S_n \ge \frac{\sigma^2}{{R^*}^2} n + \sqrt{n \log 1/\beta}} \le \beta \mper
\end{equation*}
  Since $R^* \ge \sqrt{2/\alpha}$, $n \ge 4 \log\paren{1/\beta}/\alpha^2$, $\sigma^2 \le 1$, each of these terms are bounded by $\alpha n/2$, implying that $|G| \geq (1-\alpha)n$ as desired.
\end{proof}

We will prove the following theorem using a combination of the arguments from \Cref{subsubsec:simple-coarse-algorithm} to \Cref{subsubsec:simple-coarse-runtime}.
\begin{theorem}[\textsc{coarse-1d-estimate}]
	\label{thm:coarse-estimation-1d-simple}
	Let $X_1, \dots, X_n \in \R$, $R > R^* \ge 1$, $0 \le  \alpha < 1/3$.
	Suppose there exists a point $x_0 \in \Brac{-R, R}$ such that $\Abs{\Set{i \suchthat \Abs{X_i - x_0} \le R^*}} \ge \paren{1-\alpha} n$.
	Then there exists an $\epsilon$-DP algorithm and a universal constant $C_0$ such that if $n \ge C_0\Paren{\log R + \log\paren{1/\beta}}/\epsilon\alpha$, the algorithm
	returns $y$ such that $\Abs{y - x_0} \le \cO\Paren{R^*}$,  with probability $\paren{1-\beta}$, and runtime $\cO\Paren{n^2 \log R}$.
\end{theorem}

By combining \Cref{thm:coarse-estimation-1d-simple} and \Cref{lem:1d-appropriate-radius}, we arrive at the following theorem for univariate coarse mean estimation.
\begin{theorem}[one-dimensional coarse mean estimation]
	\label{thm:one-dimensional-coarse-mean-estimation}
	Suppose $\cD$ is a one-dimensional distribution with mean $\mu$, and covariance $\sigma^2$, such that $\abs{\mu} \le R$ and $\sigma^2 \le 1$. Then there exists an $\epsilon$-DP algorithm that takes $n$ \iid\ samples from $\cD$ and outputs $y$ such that $\Abs{y - \mu} \le \cO\paren{1}$ with probability $1-\beta$, sample complexity 
	\begin{equation*}
		n = \cO\Paren{\frac{\log R  + \log \paren{1/\beta}}{\epsilon}}\mcom
	\end{equation*}
	and runtime $\poly\paren{n, \log R, 1/\eps}$.
\end{theorem}

To extend this to the multivariate setting, we apply the algorithm from \Cref{thm:one-dimensional-coarse-mean-estimation} with privacy budget $\epsilon/d$ and failure probability $\beta/d$ to each coordinate ($d$ times in total), and using composition of differential privacy.
\begin{theorem}[simple high dimensional coarse mean estimation]
	\label{thm:simple-high-dimensional-coarse-mean-estimation}
	Suppose $\cD$ is a $d$-dimensional distribution with mean $\mu$ and covariance $\Sigma$, such that $\normt{\mu} \le R$ and $\normt{\Sigma} \le 1$. Then there exists an $\epsilon$-DP algorithm that takes $n$ \iid\ samples from $\cD$ and outputs $y$ such that $\normt{y - \mu} \le \cO\paren{\sqrt{d}}$ with probability $1-\beta$, sample complexity 
	\begin{equation*}
		n = \cO \Paren{\frac{d\log R  + d \log \paren{1/\beta} + d\log d}{\epsilon}} \mcom
	\end{equation*}
	and runtime $\poly\Paren{n, d, \log R, 1/\eps}$.
\end{theorem}

In the rest of this subsection we prove \Cref{thm:coarse-estimation-1d-simple}. To do so, we propose an algorithm and provide privacy, accuracy and runtime analysis for it.
\subsubsection{Algorithm}
\label{subsubsec:simple-coarse-algorithm}
Our algorithm (\Cref{alg:1d-coarse-estimation}) is based on the exponential mechanism. We construct a cover over $\brac{-R, R}$, and run the exponential mechanism over it to capture a large number of $X_i$'s in a ball of radius $\cO\paren{R^*}$. 
This quantity is formalized using the score function as defined in \Cref{def:one-d-score}.
\begin{definition}[Score function]
	\label{def:one-d-score}
	Let $X_1, \dots, X_n \in \R, R^* \ge 1$. For any point $y \in \R$ we define
	\begin{equation*}
		\Score\paren{y} =  \Abs{\Set{i \suchthat \Abs{X_i - y} \le 2R^*}} \mper
	\end{equation*}
\end{definition}
  \begin{algorithm}
  	\caption{\textsc{coarse-1d-estimate}}
	\label{alg:1d-coarse-estimation}
	\begin{mdframed}
		\mbox{}
		\begin{description}
			\item[Input:]
			Parameters $R > R^* \ge 1$,  dataset $X_1, \dots, X_n \in \R$, 
			\item[Operation:]
			\break
			\begin{enumerate}
				\item Let $I = \Set{R^* k \suchthat \floor{-R/R^*}\le k \le \ceil{R/R^*}, k \in \Z}$. 
				\item Run the exponential mechanism over the points in $I$ with score function $\Score$ (\Cref{def:one-d-score}), privacy budget $\epsilon$, and sensitivity $1$, to obtain $\wh{y}$. 
			\end{enumerate}
			\item[Output:]  \wh{y}
		\end{description}
	\end{mdframed}
\end{algorithm}
\subsubsection{Privacy analysis}
Since $\Score$ has sensitivity $1$, the algorithm is $\epsilon$-DP by privacy of the exponential mechanism.
\subsubsection{Accuracy analysis}
\begin{lemma}
	Suppose $0 \le \alpha < 1/3$, $n \ge 2 \Paren{\log \Paren{2\ceil{R/R^*} + 1} + \log\paren{1/\beta}}/\epsilon \alpha$, and that there exists some point $x_0 \in \brac{-R , R}$ such that $\Abs{\Set{i \suchthat \Abs{X_i - x_0} \le R^*}} \ge \paren{1-\alpha} n$.
	Then the exponential mechanism in \Cref{alg:1d-coarse-estimation} will return $\wh{y}$ such that 
	\begin{equation*}
		\Abs{\wh{y} - x_0} \le 3R^* \mcom
	\end{equation*}
	with probability $1-\beta$.
\end{lemma}
\begin{proof}
	We show that there exists a point in $I = \Set{R^* k \suchthat \floor{-R/R^*}\le k \le \ceil{R/R^*}, k \in \Z}$ that has $\Score$ at least $\paren{1 - \alpha} n$. 
  Let $G$ be the set of points which are close to the point $x_0$, that is, $G = \Set{i \suchthat \Abs{X_i - x_0} \le R^*}$ and $\Abs{G} \ge \paren{1 - \alpha} n$. 
  Let $y^*$ be the point in $I$ that is closest to $x_0$.
  Since $x_0 \in \brac{-R,  R}$, $\Abs{x_0 - y^*} \le R^*$. Then 
	\begin{equation*}
		\forall i \in G \; : \; \Abs{y^* - X_i} \le \Abs{y^* - x_0} + \Abs{x_0 - X_i} \le R^* +R^* \le 2R^* \mper
	\end{equation*}
	Therefore $y^*$ has $\Score$ at least $\Paren{1 - \alpha} n$. 
  By the guarantees of the standard exponential mechanism, the exponential mechanism when instantiated with the $\Score$ function over $I$ will return \wh{y} such that
	\begin{equation*}
		\Score\paren{\wh{y}} \ge \paren{1-\alpha} n - \frac{2}{\epsilon} \Paren{\log \Paren{2\ceil{R/R^*} + 1} + \log \paren{1/\beta}} \mcom
	\end{equation*}
	with probability $1-\beta$.
  Since $\alpha n \ge \frac{2}{\epsilon} \Paren{\log \Paren{2\ceil{R/R^*} + 1} + \log\paren{1/\beta}}$, the exponential mechanism will return $\wh{y}$ such that
	\begin{equation*}
		\Score\paren{\wh{y}} \ge \paren{1 - 2\alpha} n \mcom
	\end{equation*}
	with probability $1-\beta$.
	Let $T = \Set{i \suchthat \Abs{X_i - \wh{y}} \le 2R^*}$, then $\Abs{T} \ge \paren{1-2\alpha} n$, with probability $1 - \beta$. Moreover $G = \Set{i \suchthat \Abs{X_i - x_0} \le 2R^*}$ has at least $\paren{1 - \alpha} n$ members. Therefore since $\alpha < 1/3$, we have that $T \cap G \neq \emptyset$, and therefore there exists some $j \in T \cap G$. Hence $\abs{X_j - x_0}\le R^*$ and $\abs{X_j - \wh{y}} \le 2R^*$. Therefore 
	$\Abs{y - x_0} \le 3R^*$ with probability $1-\beta$, as desired.
\end{proof}

\subsubsection{Runtime analysis: lazy exponential mechanism}
\label{subsubsec:simple-coarse-runtime}
If the exponential mechanism in \Cref{alg:1d-coarse-estimation} is \naively implemented it will take $\Theta\paren{Rn}$ time.
This is a pseudo-polynomial time algorithm, but we would prefer a truly polynomial time algorithm, with running time polynomial in $\log R$.
In the following theorem, we describe how this application of the exponential mechanism can be implemented in $\cO\paren{n^2 \log R}$ time.
This style of ``lazy exponential mechanism'' may be folklore, but we include the argument for completeness.

\begin{theorem}[lazy exponential mechanism]
  \label{thm:lazy-exp}
  Suppose $C \in \N$, $I = \Set{b_1, \dots, b_m}$ is a set of bins, and $n \in \N$ objects $O = \Set{o_1, \dots, o_n}$ are given. Moreover suppose that each item can belong to at most $C$ bins, and that for each object $o$, the set of the bins it belongs to can be computed in time $\cO\paren{C \log m}$. For a bin $b_i$, let 
	\begin{equation*}
		\Score\paren{b_i} = \Abs{ \Set{j \suchthat \text{$o_j$ belongs to $b_i$}}} \mper
	\end{equation*}
	Then the exponential mechanism over $I$ with score function $\Score$ can be implemented in time $\cO\paren{Cn^2 \log m}$.
\end{theorem}
\begin{proof}
	Let $J = \Set{b_i \suchthat \Score\paren{b_i} > 0}$. We can construct $J$ and compute the score function for all of its members by iterating over $O$ and finding which bins each of them belongs to. This will take time $\cO\paren{Cn \log m}$. Now since all of the other bins are empty, the probability that the exponential mechanism with sensitivity $\Delta$ and privacy budget $\epsilon$ selects $b_j \in J$ has to be 
	\begin{equation*}
		\Pr\brac{b_j} := \frac{\exp\paren{\frac{\epsilon \Score\paren{b_j}}{2 \Delta}}}{\sum_{k, b_k \in J} \exp\paren{\frac{\epsilon \Score\paren{b_k}}{2 \Delta}} + m - \Abs{J}} \mper
	\end{equation*}
	Therefore the probability of choosing each member of $J$ can be computed in time $\cO\paren{n \log m}$.
	Moreover, let $\bot$ be a special symbol representing the event that the exponential mechanism selects a $b_j$ such that $b_j \notin J$.
  It is easy to see that
	\begin{equation*}
		\Pr\brac{\bot} := \frac{m - \Abs{J}}{ \sum_{k, b_k \in J} \exp\paren{\frac{\epsilon \Score\paren{b_k}}{2 \Delta}} + m - \Abs{J} } \mper
	\end{equation*}

  We now draw an element from the set $\Set{b_j \suchthat j \in J} \cup \Set{\bot}$ where the probabilities are as established above.
  If the sampling algorithm picks a $b_j \in J$ we can return that, and we are done. Otherwise, if $\bot$ is picked, we can sample an index from $i \in \Set{1, \dots, m - \Abs{J}}$ uniformly at random in time $\cO \paren{\log m}$. Now we aim to construct a uniformly random member from $I - J$ from this random index $i$. To do we just need to consider $K = \set{b_i, \dots, b_{i+\abs{J}}}$. By the intermediate value theorem, there exists a member $b_{i+k} \in K$ such that $k$ equals the number of the members of $J$ which have index smaller than or equal to $i+k$. For the smallest $k$ that has such property, we return $b_{j+k}$, it is easy to see that this scheme returns a uniformly random member from $I - J$. Moreover, the overall algorithm takes time at most $\cO\paren{Cn^2\log m}$. Note that in this proof, the $\log m$ factors come from the bit complexity of indices of the members of $I$.
\end{proof}

\begin{corollary}
	The exponential mechanism in \Cref{alg:1d-coarse-estimation} can be implemented in time $\cO\paren{n^2 \log R}$.
\end{corollary}
\begin{proof}
  Each $X_i$ can be counted in the $\Score$ of at most $5$ members of $I$. Furthermore, for each $X_j$, to find the members of $I$ for which it is counted in, we just need to find the points in $\brac{X_i - 2R^*, X_i + 2R^*}$ which are of the form $kR^*$, for some $k \in \Z$. This can be done in time $\cO\paren{\log R}$. The statement thus holds by an application of \Cref{thm:lazy-exp}.
\end{proof}

\subsection{Second Algorithm: SoS Exponential Mechanism}

\subsubsection{Main Steps}
We proceed to describe our second coarse-estimation algorithm, based on our SoS Exponential Mechanism framework. 
Our coverage is for a concrete problem (as opposed to the abstract setting in our meta-theorem in~\Cref{sec:metatheorem}) and conceptually simpler than our fine-estimation algorithm (\Cref{sec:fine}), so it may be a good starting point for a reader to understand the framework.

Consider the following simple but inefficient method for performing coarse estimation.
Let the score function for a point $v$ be the number of datapoints within a ball of radius $\poly(d)$ centered at $v$. 
One could run the exponential mechanism with this score function, with the set of candidate points being the ball of radius $R$.
Compared to the coordinate-wise algorithm described in~\Cref{sec:coordinatewise}, this approach works directly in the high-dimensional space. 

Of course, this algorithm is not polynomial time, as the naive method of implementing it would involve computing the score function for all points in a cover of size exponential in $d$.
Instead, appealing to our SoS Exponential mechanism framework, we first replace the score function with a similar SoS optimization problem. 
The resulting algorithm and optimization problem are both described in~\Cref{subsubsec:sos-coarse-algorithm}.
Note that for technical reasons, we actually run the exponential mechanism $\cO(\log R)$ times, making a constant factor improvement in the radius each time.
This iterative process is what gives rise to a sub-optimal $\cO(\log R \log (1/\beta)/\epsilon)$ term in the sample complexity, 
To prove accuracy of the overall procedure, we first show that the solution to the SoS optimization problem will give a point with high utility (\Cref{lem:accuracy-coarse-sdp} in~\Cref{sec:accuracy-coarse-sdp}).
Then, assuming the instantiations of the exponential mechanism are all ``successful'' (that is, produce points which are of comparable utility to the one which optimizes the score function), it is not hard to show that the final point output will satisfy our desired accuracy guarantee (\Cref{cor:acc-assume-exp} in~\Cref{sec:acc-assume-exp}).

At this point, there are three components of the exponential mechanism left to argue: privacy, accuracy, and computational efficiency.
Privacy is fairly straightforward: we prove that the score function has bounded sensitivity (\Cref{lem:coarse-sdp-sensitivty} in~\Cref{sec:coarse-sdp-sensitivity}), which implies privacy (\Cref{cor:coarse-sdp-privacy}) by combining the exponential mechanism's privacy guarantee with basic composition.
Utility follows the standard recipe for the volume-based exponential mechanism (\Cref{thm:volume-based-exponential-mechanism}):  in addition to the aforementioned sensitivity bound, we require that points with high utility are sufficiently dense (\Cref{lem:coarse-sdp-volume} in~\Cref{sec:coarse-sdp-volume}).
Combining the two gives the desired accuracy guarantee in~\Cref{lem:coarse-sdp-accuracy} in~\Cref{sec:coarse-sdp-accuracy}.
Finally, to ensure that the algorithm is polynomial time, we employ the sampling-based exponential mechanism of~\cite{BassilyST14} (\Cref{thm:efficient-sampling}).
We again use the fact that the score function has bounded sensitivity, but also Lipschitzness (\Cref{lem:coarse-sdp-lipschitzness} in~\Cref{sec:coarse-sdp-lipschitzness}), and concavity (\Cref{cor:coarse-sdp-concavity} in~\Cref{sec:coarse-sdp-concavity}).
Combining these properties allow us to conclude the exponential mechanism can be implemented in polynomial time (\Cref{lem:sos-coarse-runtime} in~\Cref{subsubsec:sos-coarse-runtime}).

These allow us to conclude the following result for coarse mean estimation.
\begin{theorem}[high dimensional coarse mean estimation via SoS exponential mechanism]
\label{thm:sos-coarse-estimate}
Suppose $\cD$ is a $d$-dimensional distribution with mean $\mu$ and covariance $\Sigma$, where $\normt{\mu} \le R$, and $\normt{\Sigma} \le 1$.
Then there exists a polynomial time $\epsilon$-DP algorithm \textsc{coarse-estimate} that takes 
\begin{equation*}
	n = \cO \Paren{\frac{\log R \Paren{d + \log\paren{1/\beta} + \log\log R}}{\epsilon}}
\end{equation*}
\iid\ samples from $\cD$ and outputs $y$, such that $\normt{y - \mu} \le \cO\paren{\sqrt{d}}$ with probability $1-\beta$.
\end{theorem}

Before we begin the main analysis, we establish a deterministic condition on the dataset under which our algorithm succeeds.
This condition holds with high probability if the dataset is sampled from an appropriate distribution.
Specifically, the following is a multivariate analogue of \Cref{lem:1d-appropriate-radius}, which states that most samples from a distribution with bounded covariance will be within a bounded distance of the mean.
The proof is similar to before.
\begin{lemma}
\label{lem:appropriate-radius-around-the-mean}
Suppose $X_1, \dots, X_n$ are \iid\ samples from a $d$-dimensional distribution $\cD$, with mean $\mu$ and covariance $\Sigma$, where $\normt{\Sigma} \le 1$. 
Furthermore suppose $n \ge 4 \log\paren{1/\beta}/\alpha^2$ and $R^* \ge \sqrt{2 d/\alpha}$. Then with probability $1-\beta$, there exists a set 
$G$ such that 
$\Abs{G} \ge \paren{1 - \alpha} n$, and
\begin{equation*}
\forall i \in G: \normt{X_i - \mu} \le R^*\mper
\end{equation*}
\end{lemma}

With this in place, it suffices to prove the following theorem, which is our main technical result of this section. 
\begin{theorem}
\label{thm:coarse-dataset}
	For every $\epsilon > 0$ and $n,d \in \N$ there is an $\e$-differentially private polynomial-time algorithm \textsc{coarse-estimate}$(R,R^*,X_1,\ldots,X_n)$ with the following guarantees. 
  There exists a universal constant $C_0$ that
	given $X_1,\ldots,X_n \in \R^d$, $R, R^* > 0$,  \textsc{coarse-estimate} outputs $y \in \R^d$ such that if there exists $x \in \R^d$ with $\|x\|_2 \leq R$ and such that $|\{X_i \, : \, \|X_i - x\|_2 \leq R^* \}| \ge \paren{1-\alpha} n$, $R \ge 1000R^* > 1$, $n \ge C_0 \frac{\log R}{\alpha \epsilon}\Paren{d + \log\paren{1/\beta} + \log \log R}$, $0 < \alpha < 10^{-4}$, then $\|y - x\|_2 \leq O(R^*)$, with probability at least $1-\beta$. 
\end{theorem}

In the rest of this section (from~\Cref{subsubsec:sos-coarse-algorithm} to \Cref{subsubsec:sos-coarse-runtime}) we prove \Cref{thm:coarse-dataset}. To do so, we present an algorithm and provide privacy, accuracy, and runtime analysis for it. 

\subsubsection{Algorithm}
We describe the score function employed in the exponential mechanism.
Our algorithm repeatedly invokes the exponential mechanism with a geometrically decreasing radius. 
\label{subsubsec:sos-coarse-algorithm} 
\begin{definition}[\textsc{coarse-sdp}]
\label{def:coarse-sdp}
	Let $X_1,\ldots,X_n, y \in \R^d$, and $R > r > 1$.
	Let \textsc{coarse-sdp} be the following convex program, where the maximization is over degree-$4$ pseudo-distributions in indeterminates $v_1,\ldots,v_d,b_1,\ldots,b_n$.
	\begin{align*}
		\max_{\pE} & \, \, \pE \sum_{i=1}^n b_i \text{ s.t.}\\
		& \pE \text{ satisfies } b_i^2 = b_i \quad \forall i\\
		& \pE \text{ satisfies } \|v\|_2^2 \leq (2R + R/1000)^2\\
		& \pE \text{ satisfies } b_i \|X_i - v\|_2^2 \leq b_i \cdot r^2 \quad \forall i\\
		& \pE v = y
		\mper
	\end{align*}
\end{definition}

  \begin{algorithm}[H]
	\caption{\textsc{coarse-estimate}}
	\label{alg:sos-exponenital-coarse}
	\begin{mdframed}
		\mbox{}
		\begin{description}
			\item[Input:]
			Parameters $R \ge 1000 R^* \ge 1 $,  dataset $X_1, \dots, X_n \in \R^d$
			\item[Operation:]
			Let $m = \floor{\log_5 \paren{R / 1000 R^*}}$.
			For $0 \le t < m$ do:
			\begin{enumerate}
			\item Let $R_t =  R \paren{0.2}^t$.
			\item Run the exponential mechanism with score function
			\textsc{coarse-sdp} $\paren{R \gets R_t, r \gets R_t/100}$, privacy budget $\eps/m$, and sensitivity $1$ over the ball of radius $R_t + R_t/1000$ in $\R^d$ to obtain $y_t$.
			\item For every $j$ let $X_j  = X_j - y_t$.
			\end{enumerate}
			\item[Output:]  $\wh{y} = \sum_{i= 0}^{m-1} y_i$
		\end{description}
	\end{mdframed}
\end{algorithm}

\subsubsection{SoS Analysis}
\label{sec:accuracy-coarse-sdp}
We prove that a point with a high score function value will also have high utility by containing a large fraction of the dataset within a bounded radius.
\begin{lemma}
  \label{lem:accuracy-coarse-sdp}
  Let $X_1,\ldots,X_n \in \R^d$, $R > r > R^* > 1$, and $\alpha \in \brac{0, 1/3}$ and let $\pE$ be a feasible solution to \textsc{coarse-sdp} with $\pE \sum_{i \leq n} b_i \geq \paren{1 - 2\alpha} n$.
  For any $x_0 \in \R^d$ such that $| \{ i \, : \, \|X_i - x_0\|_2 \leq R^* \}| \geq \paren{1 - \alpha} n$ and $\|x_0\|_2 \leq R$,
  \[
  \| y - x_0 \|_2 < 10(r + \sqrt{\alpha} R)\mper
  \]
\end{lemma}
\begin{proof}
  Let $G = \{ i \, : \, \|X_i - x_0\| \leq R^* \}$ be the \emph{good} set of indices. Since $\sum_{i \in G} \pE\brac{b_i}
  +
  \sum_{i \notin G} \pE\brac{b_i}
  =
  \sum_{i=1}^n \pE\brac{b_i} \ge \paren{1-2\alpha} n$, and $\Abs{G} \ge \paren{1- \alpha} n$,
  \[\sum_{i \in G} \pE\brac{b_i} \ge \paren{1-2\alpha} n -  \paren{n - \abs{G}} \ge \paren{1 - 3\alpha} n \mper \]
  We have
  \[|G| \cdot \|x_0 - v\|^2 = \sum_{i \in G} b_i \cdot \|x_0 - v\|^2 + \sum_{i \in G} (1-b_i)  \cdot \|x_0 - v\|^2\mper
  \]
  We will bound the two terms separately.
  For the first, for any $i \in G$, by SoS triangle inequality (\Cref{fact:sos-triangle-inequality}) we have
  \[
  \pE b_i \|x_0 -v\|^2 \leq 2 \pE b_i (\|x_0 - X_i\|^2 + \|X_i - v\|^2) \leq 4 r^2 \cdot \pE b_i\mper
  \]
  (Here we used that $\pE$ satisfies $b_i \geq 0$.)
  For the second term, again using triangle inequality, we have
  \[
  \pE (1-b_i) \|x_0 - v\|^2 \leq 2 \pE\paren{1 - b_i}\paren{ \snormt{v} + \snormt{x_0}} 
  \leq 2 \pE\paren{1 - b_i}\paren{ R^2 + (2R+R/1000)^2} 
  \leq 11 R^2\pE (1-b_i)\mper
  \]
  (Here we used that $\pE$ satisfies $1 - b_i \ge 0$.)
  So overall,
  \[
  |G| \cdot \pE \|x_0 -v\|^2 \leq 4 r^2 \cdot \pE \sum_{i \in G} b_i + 11 R^2 \cdot \pE \sum_{i \in G} (1-b_i)\mper
  \]
  Now  $\sum_{i \in G} \pE b_i \geq \paren{1 - 3\alpha} n$, so $\pE \sum_{i \in G} (1-b_i) \leq 3 \alpha n$, and hence
  \[
  \pE \|x_0 - v|^2 \leq 11 \cdot \frac{nr^2 + 3\alpha nR^2}{|G|} \leq 66(r^2 + \alpha R^2)
  < 100\paren{r^2 + \alpha R^2}
  \mper
  \]
  Since $\|\pE v - x_0\| \leq \sqrt{\pE \|v - x_0\|^2}$, we are done.
\end{proof}

\subsubsection{Accuracy Analysis Assuming Exponential Mechanism's Success}
\label{sec:acc-assume-exp}
Here we show that if each iteration of~\Cref{alg:sos-exponenital-coarse} produces a high utility solution, the resulting estimates will improve in accuracy at a rate exponential in the number of iterations. 

\begin{proposition}
\label{prop:acc-assume-exp}
Suppose $\abs{\set{i  :  \normt{X_i - x_0} \le R^*}} \ge \paren{1 - \alpha} n$, $\alpha \le 10^{-4}$, $\|x_0\|_2 \leq R$, and assume the exponential mechanism in \Cref{alg:sos-exponenital-coarse} is successful in every step, i.e., $y_t$ has score larger than $\paren{1-2\alpha} n$ for every $t$. Then for each $t$,
\begin{equation*}
\normt{\sum_{i=0}^{t-1} y_i - x_0} \le R_t \mper
\end{equation*}
\end{proposition}
\begin{proof}
  We prove this proposition through induction. 
  For the base case $t = 0$, the sum $\sum_{i=0}^{t-1} y_i$ is empty and we are left with $\|x_0\| \leq R_0 = R$, which holds by assumption. 
  For the inductive case, we suppose the statement holds for $t -1$, and prove that it holds for $t$. 
  The inductive hypothesis is that 
  $\normt{\sum_{i= 0}^{t-2} y_i - x_0} \le R_{t - 1}$. 
  We may now use~\Cref{lem:accuracy-coarse-sdp}, where $x_0$ and $R$ in the lemma statement are replaced by $\sum_{i= 0}^{t-2} y_i - x_0$ and $R_{t-1}$, respectively.
  This gives us 
  \begin{equation*}
    \normt{\sum_{i=0}^{t-1} y_t - x_0} \le 10\paren{R_{t-1}/100 + R_{t-1}/100} \le R_{t-1} / 5 = R_t \mcom
  \end{equation*}
  as desired.
\end{proof}

As an immediate corollary, we can bound the resulting error after~\Cref{alg:sos-exponenital-coarse} is complete.
\begin{corollary}
\label{cor:acc-assume-exp}
  Under the same assumptions as~\Cref{prop:acc-assume-exp} and noting that that $m = \log_{5}\paren{R/1000R^*}$, \Cref{alg:sos-exponenital-coarse} outputs $y$, such that
  \begin{equation*}
    \normt{y - x_0} \le 1000 R^*\mper
  \end{equation*}
\end{corollary}

\subsubsection{Sensitivity}
\label{sec:coarse-sdp-sensitivity}
We now show that the score function has bounded sensitivity.
\begin{lemma}[\textsc{coarse-sdp} sensitivity]
\label{lem:coarse-sdp-sensitivty}
  \textsc{coarse-sdp} has sensitivity $1$ with respect to the $X_i$'s.
\end{lemma}
\begin{proof}
Without loss of generality suppose $X, X'$ are two datasets that differ in the first index.
Suppose for inputs $\paren{X_1, X_2, \dots X_n}$, $\pE$ is the degree-$4$ pseudo-expectation that maximizes the objective function of \textsc{coarse-sdp}. In order to prove that \textsc{coarse-sdp} has sensitivity $1$, we prove that there exists a feasible solution $\pE'$ to the \textsc{coarse-sdp} convex program, on inputs $\paren{X_1', X_2, \dots, X_n}$, such that 
\begin{equation*}
\pE' \brac{\sum b_i} \ge \pE \brac{\sum b_i} - 1 \mper
\end{equation*}
If we do so, sensitivity follows by symmetry. To construct  $\pE'$, we describe it by the pseudo-moments up to the fourth degree. For every degree-$4$ monomial $f$ in $b, v$ that does not contain $b_1$, let 
\begin{equation*}
\pE' \brac{f\paren{b , v}} = \pE\brac{f\paren{b, v}} \mcom
\end{equation*}
and for every degree-$4$ monomial $f$ in $b, v$, that contains $b_1$, 
let
\begin{equation*}
\pE' \brac{f\paren{b , v}} = 0 \mper
\end{equation*}
Then $\pE'$ satisfies the constraints in \textsc{coarse-sdp}. Furthermore, $\pE' \brac{\sum  b_i}  \ge \pE \brac{\sum b_i } - 1$. It remains to show that $\pE'$ is a pseudo-expectation. This is because $\pE'\brac{\paren{\paren{1, b, v}^{\tensor 2}} \paren{\paren{1, b, v}^{\tensor 2}}^T}$ is PSD, which is a direct result of \Cref{lem:nulling-psd-matrices}.
\end{proof}

Combining this with the privacy guarantees of the exponential mechanism and basic composition, we see that the overall algorithm is $\epsilon$-DP.
\begin{corollary}[privacy analysis]
\label{cor:coarse-sdp-privacy}
The exponential mechanism in \Cref{alg:sos-exponenital-coarse} is $\epsilon/m$-DP. Hence \Cref{alg:sos-exponenital-coarse} is $\epsilon$-DP by composition.
\end{corollary}

\subsubsection{Lipschitzness}
\label{sec:coarse-sdp-lipschitzness}
To prove that the score function is Lipschitz, we rely upon the following lemma, which is a direct result of the duality of pseudo-distributions and SoS proofs (see, e.g.,~\cite{barak2016proofs}).
\begin{lemma}
\label{lem:coarse-sdp-polynomial-identity}
Suppose $R, r, X_1, \dots, X_n$  are fixed. Let $f\paren{y}$ denote the maximum value \textsc{coarse-sdp} can obtain for $y$. Then for every $y \in \R^d, \omega > 0$, there exists $\alpha_i \in \R\brac{b, v}_{\le 2}$, a degree-$2$ SoS polynomial $\beta$, $\gamma_i \ge 0$, $\lambda_i \in \R$, and  degree-$4$ SoS polynomial $\sigma$ such that
\begin{align*}
f\paren{y} + \omega -\sum_{i=1}^n b_i &= \sum_{i=1}^n \alpha_i\paren{b, v}\paren{b_i^2 - b_i}
+ 
\beta\paren{b, v} \paren{\paren{2R +R/1000}^2 - \snormt{v}} \\
&+
\sum_{i=1}^n \gamma_i b_i \paren{r^2 - \snormt{X_i - v}}
+
\sum_{i=1}^d \lambda_i \paren{v_i - y_i}
+
\sigma\paren{b, v} \mper
\end{align*}
Note that the above equality is a polynomial identity in $b$ and $v$.
\end{lemma}
Note that the constraints on $b$ and $v$ in the optimization problem imply boundedness (\Cref{lem:boundedness-coarse-optimization-problem}) and therefore our problem satisfies the duality conditions.
We use this lemma to prove the following theorem.
\begin{lemma}[\textsc{coarse-sdp} lipschitzness]
\label{lem:coarse-sdp-lipschitzness}
In the convex optimization problem \textsc{coarse-sdp}, suppose $R, r, X_1, \dots, X_n$ are fixed. Let $f\paren{y}$ denote the maximum value \textsc{coarse-sdp} can obtain for $y$.
Then $f\paren{y}$ is $L$-Lipschitz over the ball of radius $R+R/1000$ in $\R^d$, where $L = n\sqrt{d}/R$.
\end{lemma}
\begin{proof}
Let
\begin{align*}
g\paren{\alpha, \beta, \gamma, \sigma; y} &= 
\sum_{i=1}^n b_i
+
\sum_{i=1}^n \alpha_i\paren{b, v}\paren{b_i^2 - b_i}
+ 
\beta\paren{b, v} \paren{\paren{2R +R/1000}^2 - \snormt{v}} \\
&+
\sum_{i=1}^n \gamma_i b_i \paren{r^2 - \snormt{X_i - v}}
+
\sum_{i=1}^d \lambda_i \paren{v_i - y_i}
+
\sigma\paren{b, v} \mper
\end{align*}
Note that $g$ is a polynomial in $b$ and $v$. 
Then by \Cref{lem:coarse-sdp-polynomial-identity} we know that for any $y \in \R^d, \omega > 0$, there exists $\alpha, \beta, \gamma, \lambda, \sigma$ such that
\begin{equation}
\label{eq:lipschitzness-pol-identity-coarse}
f\paren{y} + \omega = g\paren{\alpha, \beta, \gamma, \lambda, \sigma; y} \mcom
\end{equation}
is a polynomial identity in $b$ and $v$.
First, we prove that the $\lambda_i$'s are bounded.
To do so we plug different values for $b$ and $v$ in \Cref{eq:lipschitzness-pol-identity-coarse}. Let $b = 0$, and $v = y + te_j$, such that $\normt{y + te_j} = 2R + R/1000$, and $t > 0$. Such $t$ exists and is larger than $R$ due to $\normt{y} \le R + R/1000$, and the triangle inequality. Plugging this $b$ and $v$ into \Cref{eq:lipschitzness-pol-identity-coarse}, gives us
\begin{equation*}
f\paren{y} + \omega \ge \lambda_j t \mper
\end{equation*}
Since $t \ge R$, $f\paren{y} + \omega \ge R\lambda_j$. Similarly, by plugging $b = 0, v = y - te_j$, we obtain $f\paren{y} + \omega \ge -R\lambda_j$, and therefore, $\abs{\lambda_j} \le \paren{f\paren{y} + \omega}/R$. Since $f\paren{y} \le n$, we obtain $\abs{\lambda_i} \le \paren{n + \omega}/R$, and hence 
\begin{equation*}
\normt{\lambda} \le \sqrt{d} \paren{n + \omega}/R \mper
\end{equation*}

Now that we have bounded $\lambda$, we can prove that $f$ is Lipschitz. To do so, suppose $y' = y + \delta u$, where $\normt{u} = 1$.
If for $y'$ we can construct a dual solution that obtains value $f\paren{y} + L\delta + h\paren{\omega}$, where $\lim_{\omega \to 0} h\paren{\omega} = 0$, we have proven that $f\paren{y'} \le f\paren{y} + L\delta$, and Lipschitzness will follow by symmetry.
To construct this dual solution, let $\paren{\alpha', \beta',\gamma',\lambda', \sigma'} = \paren{\alpha, \beta, \gamma, \lambda, \sigma}$. Then
\begin{align*}
g\paren{\alpha', \beta',\gamma', \lambda' \sigma'; y'} &= 
g\paren{\alpha, \beta, \gamma, \lambda, \sigma; y}
+
\delta \iprod{\lambda, u} \\
&\le
f\paren{y} + \omega + \delta\normt{\lambda} \\
&\le
f\paren{y} + \omega + \delta \sqrt{d} \paren{n + \omega} / R \\
&\le
f\paren{y} + \delta \sqrt{d}n/R + \omega \paren{1 + \delta \sqrt{d}/R} \mcom
\end{align*}
as desired. Therefore $f$ is $L = n\sqrt{d}/R$-Lipschitz.
\end{proof}

\subsubsection{Volume}
\label{sec:coarse-sdp-volume}
We show that the set of points with high score function is sufficiently dense. 
This essentially follows from the fact that every point which is close to the optimal solution also has high utility, and by an argument about the volume of $\ell_2$-balls.
\begin{lemma}[\textsc{coarse-sdp} volume]
\label{lem:coarse-sdp-volume}
Let $R>1000R^*>1$.
Suppose there exists a point $x_0 \in \R^d$ such that $| \{ i \, : \, \|X_i - x_0\|_2 \leq R^* \}| \geq \paren{1 - \alpha} n$ and $\|x_0\|_2 \leq R$. Let $\bbB$ denote the ball of radius $R +R/1000$ in $\R^d$. Then there exists a set $\cH^* \subseteq \bbB$, such that for every point $y$ in $\cH^*$, the value of \textsc{coarse-sdp} with $\paren{R \gets R, r \gets R/100}$ is larger than $\paren{1-\alpha} n$ and 
\begin{equation*}
\log \frac{\vol \paren{\bbB}}{\vol \paren{\cH^*}} \le d \log\paren{2000} \mper
\end{equation*}
\end{lemma}
\begin{proof}
Let $\cH^*$ be the ball of radius $R/1000$ centered at $x_0$. 
Then $\cH^* \subseteq \bbB$, and
\begin{equation*}
\log \frac{\vol{\bbB}}{\vol\paren{\cH^*}} = \log \Paren{\frac{R +R/1000}{R/1000}}^d \le d\log\paren{2000} \mper
\end{equation*}
It remains to prove that every point in $\cH^*$ has high \textsc{coarse-sdp} value.
Suppose $y$ is a point in $\cH^*$.
We know that \textsc{coarse-sdp} is a relaxation of the polynomial optimization problem; therefore, if we propose a solution to the polynomial optimization problem with value $\paren{1-\alpha} n$,  we have proven that \textsc{coarse-sdp} has value at least $\paren{1-\alpha} n$.
Hence, it remains to show that there exists a set of size at least $\paren{1-\alpha}n$ from the $X_i$'s such that $\normt{X_i - y} \le R/100$. Consider the set $G = \Set{i \; : \; \normt{X_i - x_0} \le R^*}$. We know that $\abs{G} \ge \paren{1-\alpha}n$.
Moreover,
\begin{equation*}
\forall i \in G: \; \normt{X_i - y} \le \normt{X_i - x_0} + \normt{x_0  -y}  \le R^* + R/1000 \le 2R/1000 \le R/100 \mcom
\end{equation*}
as desired.
\end{proof}
\subsubsection{Concavity}
\label{sec:coarse-sdp-concavity}
It is not hard to see that the score function is concave.
\begin{corollary}[\textsc{coarse-sdp} concavity]
\label{cor:coarse-sdp-concavity}
	Let $f\paren{y}$ denote the value of \textsc{coarse-sdp} for $y$. Then $f$ is concave.
\end{corollary}
\begin{proof}
	Direct result of \Cref{lem:convacity-maximization}.
\end{proof}

\subsubsection{Accuracy Analysis of the Exponential Mechanism}
\label{sec:coarse-sdp-accuracy}
Given that we have established bounded sensitivity (\Cref{lem:coarse-sdp-sensitivty}) and that the set of good solutions is sufficiently dense (\Cref{lem:coarse-sdp-volume}), we can apply the volume-based exponential mechanism (\Cref{thm:volume-based-exponential-mechanism}).
\begin{lemma}[guarantee of the exponential mechanism for \textsc{coarse-sdp}]
\label{lem:coarse-sdp-accuracy}
Let $f\paren{y}$ denote the value of \textsc{coarse-sdp} for $y$.
The exponential mechanism in \Cref{alg:sos-exponenital-coarse} with score function $f$, and privacy budget $\epsilon/m$ returns $y$ such that
\begin{equation*}
f\paren{y} \ge 
\paren{1-\alpha} n - \frac{2m}{\epsilon} \Paren{ d \log\paren{2000} + \log\paren{m/\beta}} \mcom
\end{equation*}
with probability $1 - \beta/m$.
Specifically, if
\begin{equation*}
n \ge \frac{2m}{\alpha \epsilon}\Paren{d \log\Paren{2000} + \log\paren{m/\beta}} \mcom
\end{equation*}
then
\begin{equation*}
f\paren{y} \ge \paren{1-2\alpha} \mcom
\end{equation*}
with probability $1-\beta/m$.
We say the exponential mechanism has been successful if the above inequality holds. 
\end{lemma}

\subsubsection{Runtime Analysis of the Exponential Mechanism}
\label{subsubsec:sos-coarse-runtime} 
At this point, we have established bounded sensitivity (\Cref{lem:coarse-sdp-sensitivty}), Lipschitzness (\Cref{lem:coarse-sdp-lipschitzness}), and concavity (\Cref{cor:coarse-sdp-concavity}).
These properties suffice to apply the efficient sampling algorithm of~\Cref{thm:efficient-sampling}. 
\begin{lemma}[runtime of the exponential mechanism]
\label{lem:sos-coarse-runtime}
There exists a sampling scheme for the exponential mechanism in \Cref{alg:sos-exponenital-coarse}
that takes time $\poly\paren{n, d, \log R, 1/\epsilon}$.
\end{lemma}
Note that, since the function is $L = n\sqrt{d}/R$-Lipschitz and since we are running the algorithm over the ball of radius $R+R/1000$ (i.e., $\diam(\cC) \leq R + R/1000$), we have that $L \diam(\cC) =O(n\sqrt{d})$ thus making the running time polynomial in $n$ and $d$, but not in $R$.

\subsection{Robustness}
To conclude, we note that both of our coarse estimation algorithms are robust to contamination.
\begin{remark}[Robustness of Coarse Estimation]
	\label{rem:coarse-robustness}
	For both coarse estimation algorithms (i.e., the coordinate-wise approach, \Cref{alg:1d-coarse-estimation}, and the SoS exponential mechanism approach, \Cref{alg:sos-exponenital-coarse}), the $\Score$ function has sensitivity $1$. Moreover, there exists a constant $\eta_0$ (say $\eta_0 \le 0.01$), such that if $\Score$ is perturbed by $\eta n$, $\eta \le \eta_0$, all of our arguments still hold. Therefore, there exists a universal constant $\eta_0$ for which our algorithm has the same guarantees if $\eta$-fraction of the samples are corrupted, where $\eta \le \eta_0$.
\end{remark}

%% file: fine-estimation.tex
\newcommand{\tmu}{\widetilde{\mu}}
\newcommand{\SDPVAL}{\textsc{SDP-VAL}}
\newcommand{\QUADVAL}{\textsc{QUAD-VAL}}
\newcommand{\QUAD}{\textsc{QUAD}}
\newcommand\numberthis{\addtocounter{equation}{1}\tag{\theequation}}
\section{Fine Estimation}
\label{sec:fine}
From \Cref{sec:coarse}, we know how to find an estimate of the mean up to distance $\cO\paren{\sqrt{d}}$.
In this section we prove the following theorem. Putting the following theorem and our theorem for coarse estimation (\Cref{thm:coarse-estimation}) together, gives us an efficient algorithm for privately estimating the mean of a $d$-variate distribution with bounded mean and covariance up to error $\alpha$.

\begin{theorem}[high dimensional fine mean estimation via SoS exponential mechanism]
	\label{thm:high-dimensional-fine-mean-estimation}	
	Suppose $\cD$ is a $d$-dimensional distribution with mean $\mu$ and covariance $\Sigma$, where $\normt{\Sigma} \le 1$. Moreover, suppose that an initial estimate $\tmu_0$ is given such that $\normt{\mu -\tmu_0} \le \cO\paren{\sqrt{d}}$. Then there exists a polynomial time $\epsilon$-DP algorithm \textsc{fine-estimation} that takes $n$ samples from $\cD$ and outputs $\tmu^*$ such that $\normt{\tmu^* - \mu} \le \alpha$, with probability $1-\beta$, and sample complexity
	\begin{equation*}
		n = \tilde{\cO} \Paren{\frac{d + \log\paren{1/\beta}}{\alpha^2\epsilon}} \mper
	\end{equation*}
	More specifically, $\tilde{\cO}$ is hiding the following lower order logarithmic factors: $\log d$, and $\log\paren{1/\alpha}$
\end{theorem}
\begin{proof}
	Applying \Cref{thm:fine-estimation} with $m = 1/\alpha^2$, and noting that \Cref{ass:when-at-mu-sdp-str} holds with probability $1-\beta$, due to \Cref{lem:well-conditonedness-sdp}.
\end{proof}

Suppose $\cD$ is a $d$-dimensional distribution with mean $\mu$, and covariance $\Sigma$ such that $\normt{\Sigma} \le 1$. Our goal is to privately estimate the mean of this distribution up to accuracy $\alpha$ with probability $1-\beta$.
Moreover, suppose that we have access to an initial estimate of the mean such that 
$\normt{\mu - \tmu_0} \le \cO\paren{\sqrt{d}}$.
We build our algorithm on top of the recent work in efficient high-dimensional mean estimation where the high-dimensional median of means technique is used \cite{CherapanamjeriFB19, Hopkins20, LugosiM19a}.
As opposed to non-biased estimators that have $\cO\paren{\frac{d}{\beta\alpha^2}}$ sample complexity, these estimators achieve the $\cO\paren{\frac{d + \log\paren{1/\beta}}{\alpha^2}}$, or the \emph{sub-Gaussian rate} sample complexity.

In this technique, we take $n = mk$ samples, and divide them into $k$ bins $\Set{\cB}_{i=1}^k$ each having $m$ elements. Then we compute the means of each of these bins and name them $Z_i$'s.
A key observation then is that if $k \ge \max\set{d, \Omega\paren{\log\paren{1/\beta}}}$, then with probability $1-\beta$ over the randomness of the samples, the true mean, $\mu$ is $\paren{0.99, \cO\paren{\sqrt{1/m}}}$-central. 
\begin{definition}[$\paren{\rho, r}$-centrality \cite{LugosiM19a, Hopkins20}]
	Suppose $Z_1, \dots, Z_k \in \R^d$. A point $x \in \R^d$ is $\Paren{\rho, r}$-central, if and only if for all unit vectors $v$ we have 
	\begin{equation*}
		\Abs{\Set{i \suchthat \iprod{Z_i - x, v} \le r} } \ge \rho k
	\end{equation*}
\end{definition}
In other words, with high probability, for any direction $v$, if we center $Z_i$'s at the point $\mu$, and project them onto the direction $v$, a large fraction of the projections will fall in a radius of $\cO\paren{\sqrt{1/m}}$ of the mean. This can be viewed as $\mu$ being a median in every direction. 

Now note that if we can find some other point $\wh{\mu}$ such that $\normt{\wh{\mu} - \mu} \le \cO\paren{\sqrt{1/m}}$, by setting $m = 1/\alpha^2$ we may obtain the sub-Gaussian rate as desired.
Suppose there exists another point $\tmu$ which is also $\paren{0.1, \cO\paren{\sqrt{1/m}}}$-central, then by choosing $v$ along the direction of $\mu - \tmu$, and applying the pigeonhole principle and the triangle inequality over the projections of $Z_i$'s in this direction, we may conclude that $\normt{\tmu - \mu} \le \cO\paren{\sqrt{1/m}}$. 
Another useful observation is that if $\normt{\tmu - \mu} \le \cO\Paren{\sqrt{\frac{1}{m}}}$, then $\tmu$ is a central point. Therefore, we arrive at the following conclusion: assuming $\mu$ is central; then a point $\tmu$ is central, if and only if it is close to $\mu$. This is the key observation that helps us characterize the points that are close to $\mu$. One can think of this centrality property as the ``median" part of the ``median of means" approach.

Therefore, answering the following question, gives us an algorithm for estimating the mean with sub-Gaussian rates.

\begin{quote}
\textbf{Main Problem}
Suppose $Z_1, \dots, Z_k \in \R^d$ are given, and there exists a point $\mu \in \R^d$ which is a central point for $Z_1, \dots, Z_k$.
How can we find a central point, both privately and efficiently?
\end{quote}

Suppose we did not care about either privacy or efficiency. Similar to \cite{CherapanamjeriFB19}, our goal is to start off from an estimate $\tmu$ and move towards a center $\mu$. If we were close to $\mu$, from the discussion above, $\tmu$ would have been a center. Therefore, one way to check whether or not we have reached our destination is by checking if $\tmu$ is central.
What if $\tmu$ were far from the center $\mu$ --- which direction would move us towards $\mu$?
If $\tmu$ is far from the center, then $\tmu$ is not central and moreover, there exists some direction $v$ such that
\begin{equation*}
	\Abs{\Set{i \suchthat\iprod{Z_i - \tmu, v} \ge 0.99 \normt{\mu - \tmu}}} \ge 0.9 k \mper
\end{equation*}
This direction $v$ certifies that $\tmu$ is not central. Now a key observation is that if $v$ certifies non-centrality, then $v$ points towards $\mu$, i.e.
\begin{equation*}
	\Iprod{v, \frac{\mu - \tmu}{\normt{\mu - \tmu}}} \ge \Omega\Paren{1} \mper
\end{equation*}
Now, if we move in the direction of $v$ (which certifies non-centrality) with step-size proportional to the distance we are from $\mu$ (i.e., $\normt{\tmu - \mu}$), we can move towards $\mu$ at a linear rate. Doing this for a number of steps which is logarithmic in the initial distance gives us a point which is central with respect to $Z_i$'s and has distance $\cO\paren{\sqrt{1/m}}$ to $\mu$.

In order to employ this approach both efficiently and privately we need to execute the following tasks efficiently and privately: 
\begin{enumerate}
  \item Check whether or not the current estimate $\tmu$ is central.
  \item Find an estimate of the distance to the central point or $\normt{\tmu - \mu}$.
  \item Find a direction $v$ that certifies the non-centrality of $\tmu$.
\end{enumerate}

Even setting the concern of privacy aside for the moment, it is not clear how to efficiently check whether or not the current estimate $\tmu$ is central. 
One approach involves constructing a net over the unit sphere $\cS^{d-1}$ and counting the number of $Z_i$'s that fall far for each direction in the unit sphere. 
However, doing so would take exponential time as the number of points on the net grows exponentially in $d$.
In order to perform this efficiently, \cite{CherapanamjeriFB19} uses the standard SDP relaxation of the following quadratic program, which finds the direction that has most of $Z_i$'s at distance at least $r$ from the current estimate $\tmu$.
Note that a point $x$ is $\paren{\rho, r}$-central if and only if the following quadratic program has a low value.

\begin{definition}[quadratic optimization problem \cite{Hopkins20, CherapanamjeriFB19}]
	Let $Z_1, \dots, Z_k, \tmu \in \R^d$, $r > 0$. Let $\QUAD\paren{\tmu, r, Z}$ be the following quadratic program.
	\label{def:quadratic-optimization-problem}
	\begin{align*}
		\QUAD\paren{\tmu, r, Z} :=
		\max_{v, b} \quad &\sum_{i=1}^k b_i^2 \\
		\text{s.t.} \quad  
		&b_i^2 = b_i \quad \forall i \\
		&\sum_{i=1}^d v_i^2 = 1 \\
		&b_i^2 r \le \iprod{Z_i - \tmu, b_i v} \quad \forall i
	\end{align*}
\end{definition}
The following is the standard SDP relaxation of the above quadratic program.
\begin{definition}[SDP Relaxation \cite{Hopkins20, CherapanamjeriFB19}]
	Let $Z_1, \dots, Z_k, \tmu \in \R^d$, $r > 0$. Let $\SDP\paren{\tmu, r, Z}$ be the following semi-definite program.
	\label{def:sdp-relax}
	\begin{align*}
		\SDP\paren{\tmu, r, Z} : = 
		\max_{v, b, B, U, W} \quad & \Tr \Paren{B} 
		\\
		\text{s.t.} \quad &\begin{bmatrix}
			1 & \transpose{b}& \transpose{v} \\
			b   &B & W \\
			v &\transpose{W} & V
		\end{bmatrix}
		\succcurlyeq 0 \\
		& B_{ii} = b_i \quad \forall i \\
		&\Tr(V) = 1 \quad \forall i\\
		&B_{ii} \cdot r \le \iprod{Z_i - \tmu, W_i} \quad \forall i
	\end{align*}
\end{definition}
Now based on the above SDP relaxation, we may define a new notion of centrality.
\begin{definition}[$\paren{\rho, r}$-SDP-centrality \cite{Hopkins20}]
A point is $\paren{\rho, r}$-SDP-central if and only if the SDP relaxation's value for that point is less than $1- \rho$.
\end{definition}
The main advantage SDP-centrality has over the previous notion of centrality is that it can be computed efficiently.
It can be proven that SDP-centrality has the same good properties that the previous definition of centrality provides: 
\begin{enumerate}
  \item The true mean $\mu$ is ($0.999$, $\cO\paren{\sqrt{1/m}}$)-SDP-central with high probability
(\Cref{lem:well-conditonedness-sdp}).
\item  Under the assumption that $\mu$ is $\paren{$0.999$, \cO\paren{\sqrt{1/m}}}$-SDP-central, we have that $\tmu$ is $\paren{$0.1$, \cO\paren{\sqrt{1/m}}}$-SDP-central if and only if $\normt{\mu - \tmu} \le \cO\paren{\sqrt{1/m}}$ (a corollary of the proof of \Cref{thm:halt-estimation}, Lemma~4 of \cite{CherapanamjeriFB19})
\item  The solution to the SDP relaxation can certify the non-SDP-centrality of points which are far from a SDP-center. Moreover, there exists a rounding scheme that obtains a direction $v$ from this solution such that
$\iprod{v, \frac{\mu - \tmu}{\norm{\mu - \tmu}}} \ge \Omega\paren{1}$.\footnote{\cite{CherapanamjeriFB19} shows that taking the top eigenvector of $V$ in \Cref{def:sdp-relax} is one such rounding scheme. In \Cref{thm:accuracy-ng}, we show that the simple rounding scheme of taking $v$ itself provides the same guarantees.}
\end{enumerate}
These properties of the SDP relaxation, together with its computational efficiency, yield a computationally efficient algorithm for estimating the mean up to $\cO\paren{\sqrt{1/m}}$ error.

Now that we have reviewed \cite{CherapanamjeriFB19}'s method, we describe how our \emph{private} algorithm works.
One naive approach to ensuring privacy would be to add Laplace noise to the direction after the rounding step.
However, it is not clear how to prove sensitivity bounds in this case to guarantee pure differential privacy.
In order to privately find the direction, we run the exponential mechanism over the unit ball.\footnote{To be precise, in our final algorithm this is not exactly the unit ball. We run the exponential mechanism over a ball of radius $1-\Theta\paren{1}$. This helps us ensure \Cref{def:cond-poly} or Lipschitzness.}
If the score function of the exponential mechanism for a given vector $y$ is 
\begin{equation*}
\Score\paren{y} = 
\sum_{i=1}^k \mathbbm{1} \brac{\iprod{Z_i - \tmu, y} \ge r} \
\end{equation*}
and we assume that the exponential mechanism always returns the $y$ with the highest score, the exponential mechanism would return a direction $y$ such that $\iprod{y, \mu -\tmu}\ge \Omega\paren{1}$. Moreover, the sensitivity of the score function is only $1$. 
However, na\"ive implementations of the exponential mechanism would not be computationally efficient, running in time exponential in $d$.
To alleviate this issue, we recall the efficient sampling algorithm of~\cite{BassilyST14}.
Their results say that if a score function is concave and Lipschitz, and is defined over a bounded and convex domain, then appropriate techniques from the literature on sampling log-concave distributions allow us to efficiently implement the exponential mechanism.
Meanwhile, this score function may not even be continuous.

To define a score function $\Score\paren{y}$ that has the mentioned properties, we use the SDP relaxation mentioned in \Cref{def:sdp-relax}, and intersect it with a new condition $v = y$ to define $\SDPVAL$ (see \Cref{def:SDPVAL}). In order to show that sampling from this $\Score$ function can be done efficiently, we show that it is concave (\Cref{cor:sdp-val-concavity}) and Lipschitz (\Cref{lem:sdpval-lipschitzness}). Moreover, in order to show that with high probability the exponential mechanism returns a high scoring vector, we show that the score function has sensitivity $1$ and that the volume of the high scoring points is large (\Cref{cor:sdp-sensitiviy}, \Cref{subsubsec:sdp-val-volume}).
It remains to show that given a feasible solution $X(1, b, v, B, V)$ to the SDP program defined in \Cref{def:sdp-relax} with high value, $\iprod{v, \normt{\mu- \tmu}}\ge \Omega\paren{1}$.
This guarantees that if the exponential mechanism returns a high-scoring vector, then that vector points towards $\mu$.
Note that previous works use a slightly different rounding scheme. In \cite{CherapanamjeriFB19}, they show that the top eigenvector of the sub-matrix $V$ of $X(1, b, v, B, V)$ has the said property; however we show that a simpler rounding scheme (alluded to above) also gives our desired guarantees (\Cref{thm:accuracy-ng}).
For further details on gradient estimation, see \Cref{subsec:grad-estimation}.

It remains to address the other two components of the algorithm: determining when to stop, and estimating the distance. 
For the former, the problem would be trivial without privacy: one could simply check the value of the SDP with $r = \Theta\paren{\sqrt{1/m}}$. 
If the value is low, we would be close, otherwise, we would have been far. 
Our private approach is similar, but we also add noise to the value of the SDP with Laplace mechanism to ensure privacy. 
We prove the SDP value has sensitivity $1$, in order to bound the magnitude of the Laplace noise.
Further discussion on how to determine when to stop appears in \Cref{subsec:halt-estimation}.

Finally, for distance estimation sans privacy, a binary search suffices to select an $r$: find the maximum $r$ such that for this choice of $r$, the SDP has value larger than $0.9k$. It can be proven that choosing $r$ this way gives an approximation of the distance to the center \cite{CherapanamjeriFB19}. 
However, it is not clear how to directly make this approach private. 
Therefore, we instead find some $r$ for which the SDP has value between $0.9k$ and $0.95k$ via a private binary search. Note that if we are far enough from the center, and if we are given a good coarse estimate, both of these values should be feasible, which is ensured by the guarantee of the stopping step and having access to a coarse estimate. Since the value of SDP is a decreasing function in $r$ and we are sufficiently far from the center, it can be proven that if the SDP has value between $0.9k$ and $0.95k$ for $r$, then $r$ is an approximation of the distance to the center. 
Full discussion on the private binary search procedure and private distance estimation appears in \Cref{subsec:distance-estimation}.

\paragraph{Structure.} First we define a new quadratic program and a new SDP that we use to define a score function for the exponential mechanism that helps us find the gradient in our algorithm. Then we review some useful lemmas and assumptions from \cite{CherapanamjeriFB19} in \Cref{subsubsec:useful-lemmas}. After that we propose the main algorithm and the state the guarantees it provides in \Cref{subsec:fine-algorithm}.  As discussed above, our algorithm has three parts: estimating whether it should stop or not, estimating the distance, and estimating the direction to the true mean. We present our algorithms, guarantees and proofs, for these three tasks in \Cref{subsec:halt-estimation}, \Cref{subsec:distance-estimation}, and \Cref{subsec:grad-estimation} respectively. In \Cref{subsec:gradient-descent} we show that after sufficient number of rounds of the gradient descent our algorithm obtains a point close to the true mean. Finally, in \Cref{subsec:fine-robust}, we explain why our algorithm is robust under constant fraction corruption in the samples.
\subsection{Preliminaries}
\subsubsection{Definitions}
\label{subsubsec:definitions}
In order to run the exponential mechanism we need to define a score function over the set of directions. In order to do so, we define the following programs based on the previous definitions (\Cref{def:quadratic-optimization-problem}, \Cref{def:sdp-relax}). These programs are the intersection of the above programs with the new condition, $v = y$. These new programs can be interpreted as: for a given direction $v$, how many of the $Z_i$'s are at distance at least $r$ from the current estimate $\tmu$, in the direction $v$? 
\begin{definition}[$\QUADVAL$] Let $Z_1, \dots, Z_k, \tmu, y \in \R^d$, $r > 0$. Let $\QUADVAL\paren{y; \tmu, r, Z}$ be the following quadratic program.
	\label{def:QUADVAL}
	\begin{align*}
		\QUADVAL\paren{y; \tmu, r, Z} :=
		\max_{v, b} \quad &\sum_{i=1}^k b_i^2 \\
		\text{s.t.} \quad  
		&b_i^2 = b_i \quad \forall i\\
		&v_i = y_i \quad \forall i\\
		&b_i^2 r \le \iprod{Z_i - \tmu, b_i v} \quad \forall i
	\end{align*}
\end{definition}
Unlike the previous SDP (\Cref{def:sdp-relax}), the following SDP is \textit{not} a standard SDP relaxation of the above quadratic program. This is because, in order to be able to run the exponential mechanism efficiently, we need to define our score function over the unit ball as opposed to the unit sphere; hence, $\normt{y}$ may not be equal to $1$. However, we show that $\SDPVAL\paren{y}$ is larger than $\QUADVAL\paren{y}$, in \Cref{lem:quad-val-sdp-val}.
\begin{definition}[$\SDPVAL$]
Let $Z_1, \dots, Z_k, \tmu, y \in \R^d$, $r > 0$. Let $\SDPVAL\paren{y; \tmu, r, Z}$ be the following semi-definite program.
\label{def:SDPVAL}
\begin{align*}
	\SDPVAL\paren{y; \tmu, r, Z}:= \max_{v, b, B, W, V} \quad & \Tr\Paren{B} \\
	\text{s.t.} \quad &\begin{bmatrix}
		1 & \transpose{b}& \transpose{v} \\
		b   &B & W \\
		v &\transpose{W} & V
	\end{bmatrix}
	\succcurlyeq 0 \\
	&v_i = y_i \quad \forall i\\
	& B_{ii} = b_i \quad \forall i \\
	&\Tr(V) = 1 \\
	&B_{ii} \cdot r \le \iprod{Z_i - \tmu, W_i} \quad \forall i
\end{align*}
Note that the above semi-definite program can also be viewed as the following convex program, where the maximization is over degree-$2$ pseudo-distributions in indeterminates $v_1, \dots, v_d, b_1, \dots b_k$. In our proofs, we use the two definitions interchangeably.
 \begin{align*}
 	\SDPVAL(y; \tmu, r, Z) := \max_{\pE} \quad & \pE \sum_{i=1}^k b_i  \\
 	\text{s.t.} \quad &
 	\pE v_i = y_i \quad \forall i \\
 	&\pE \text{ satisfies } b_i = b_i^2 \quad \forall i \\
 	&\pE \text{ satisfies } \snorm{v} = 1 \\
 	&\pE \text{ satisfies } b_i^2 r \le \iprod{Z_i - \tmu, b_i v} \quad \forall i
 \end{align*}
\end{definition}
\subsubsection{Useful Lemmas}
\label{subsubsec:useful-lemmas}
In this subsection we state some useful lemmas regarding the solutions of the quadratic program and the SDP defined in \Cref{def:quadratic-optimization-problem}, and \Cref{def:sdp-relax}. These lemmas are due to \cite{CherapanamjeriFB19}. We state a more general version of them due technical reasons, however the proofs remain similar.

The following lemma states that if we are at the true mean, then with high probability, the semidefinite program defined in \Cref{def:sdp-relax}, cannot have a high score. In the quadratic program's language, this corresponds to a large fraction of $Z_i$'s being close to the true mean when projected onto any given direction.
\begin{lemma}[Lemma 7 in \cite{CherapanamjeriFB19}]
\label{lem:well-conditonedness-sdp}
	Let $M \ge 20$, and $\cY = \paren{Y_1, \dots, Y_k} \in \R^{k\times d}$ be $k$ \iid random vectors with mean $\mu$ and covariance $\Lambda$ and let $\cS$ denote the set of feasible solutions of the $\SDP\paren{\mu, r, \cY}$ in \Cref{def:sdp-relax}. Then, we have for $r \ge 5M \paren{\sqrt{\Tr \Lambda/k} + \sqrt{\norm{\Lambda}}}$ and $k \ge 8 M^2 \log 1/\beta$
	\begin{equation*}
		\max_{\paren{v, b, B, U, W} \in \cS} \Tr{B} \le \frac{k}{M} \mcom
	\end{equation*}
	with probability at least $1- \beta$. More specifically, for $r\ge 10M\sqrt{\normt{\Lambda}}$ if $k \ge 8M^2 \log 1/\beta$, and $k \ge d$,
	\begin{equation*}
		\max_{\paren{v, b, B, U, W} \in \cS} \Tr{B} \le \frac{k}{M} \mcom
	\end{equation*}
	with probability at least $1-\beta$.
\end{lemma}
Therefore, we may assume that with high probability, the following assumption holds. In \cite{CherapanamjeriFB19}, they state the above lemma and the following assumption with $M = 20$. In our proofs we require a larger value for $M$. For the sake of tidiness, we do not set a value for $M$ in the following assumptions and we state them for general $M$. However, in our proofs we only require the following assumptions to hold for $M = 1000$. 
\begin{assumption}[Assumption 2 in \cite{CherapanamjeriFB19}]
	\label{ass:when-at-mu-sdp-str}
	Suppose
	$r^* =2\sqrt{k/n} = 2\sqrt{1/m}$.
	For the bucket means $Z = \Paren{Z_1, \dots, Z_k}$ let $\cS_r$ denote the set of feasible solutions for $\SDP\paren{\mu, r, Z}$. Then, for all $r \ge 5M r^*$,
	\begin{equation*}
		\max_{v, b, B, W, V \in \cS_r} \sum_{i=1}^k \Tr{B} \le \frac{k}{M}
	\end{equation*}
\end{assumption}

The above assumption is a strengthening of the following assumption.

\begin{assumption}[Assumption 1 in \cite{CherapanamjeriFB19}]
	\label{ass:when-at-mu-simple-str}
	For the bucket means, $Z = \Paren{Z_1, \dots, Z_k}$, we have:
	\begin{equation*}
		\forall v \in \R^d, \normt{v} = 1 \implies \Abs{\set{i : \iprod{Z_i - \mu, v} \ge 5M r^*}} \le \frac{k}{M}
	\end{equation*}   
\end{assumption} 
The following lemma is about the feasible solutions to the $\SDP$ problem that have high values.
\begin{lemma}[Lemma 3 in \cite {CherapanamjeriFB19}]
	\label{lem:RT-lemma-str}
	Let us assume \Cref{ass:when-at-mu-sdp-str}. Let 
	\begin{equation*}
		X = 
		\begin{bmatrix}
			1 & \transpose{b}& \transpose{v} \\
			b   &B & W \\
			v &\transpose{W} & V
		\end{bmatrix}
		\in 
		\R^{\paren{k + d + 1} \times \paren{k + d + 1}}
	\end{equation*}
	be a PSD matrix. Moreover, suppose that
	\begin{equation*}
		B_{i,i} = b_i,\quad 
		\Tr\paren{V} = 1,\quad 
		\Tr\paren{B} \ge \psi k \mper
	\end{equation*}
	Then, there is a set of at least $\paren{\psi - 1/ M}k$ indices $\cT$ such that for all $i \in \cT$:
	\begin{equation*}
		\iprod{Z_i - \mu, W_i} < 5Mr^*B_{i, i} \mcom
	\end{equation*}
	and a set of at least $\paren{\psi - \psi'}/\paren{1 - \psi'}k$ indices $\cR$ such that for all $j \in \cR$, we have $B_{j,j} \ge \psi'$, where $0 \le \psi' \le \psi$.
\end{lemma}
\subsection{Main Algorithm and Theorem}
In this section we state our main algorithm and theorem. As discussed before, our main algorithm is a gradient descent style algorithm and consists of three parts: \textsc{halt-estimation}, \textsc{distance-estimation}, \textsc{gradient-estimation}. We discuss these algorithms and their guarantees in \Cref{subsec:halt-estimation}, \Cref{subsec:distance-estimation}, and \Cref{subsec:grad-estimation} respectively.
\label{subsec:fine-algorithm}
\begin{algorithm}[H]
	\caption{\textsc{fine-estimation}}
	\label{alg:fine-estimation}
	\begin{mdframed}
		\mbox{}
		\begin{description}
			\item[Input:]
			Parameters: Data points $X_1, \dots, X_n \in \R^d$, number of bins $k$, number of iterations $T$, initial estimate $\tmu_0$.
			\item[Operation:]
			Step-size $\eta = 0.075$\\
			Split the data in $k$ bins $\Set{\cB_i}_{i=1}^k$, each having $m$ members. For all $1 \le i \le k$, let $Z_i = \textsc{mean}\paren{\cB_i}$. \\
			For $1 \le t \le T$ do:
			\begin{enumerate}
				\item Let $h_t = \textsc{halt-estimation} \Paren{Z, \tmu_{t-1}}$
				\item If $h_t$, then return $\tmu_{t-1}$. 
				\item Let $d_t = \textsc{distance-estimation} \Paren{Z, \tmu_{t-1}}$
				\item Let $g_t = \textsc{gradient-estimation} \Paren{Z, \tmu_{t-1}, d_t}$
				\item Let $\tmu_{t} = \tmu_{t-1}+ \eta d_tg_t$
			\end{enumerate}
			\item[Output:]  $\tmu_T$
		\end{description}
	\end{mdframed}
\end{algorithm}
The following theorem is the guarantee that our algorithm provides, under the assumption that \Cref{ass:when-at-mu-sdp-str} holds, and an initial estimate is given. This theorem together with \Cref{lem:well-conditonedness-sdp}, gives us \Cref{thm:high-dimensional-fine-mean-estimation}.
\begin{theorem}[\textsc{fine-estimation}]
	\label{thm:fine-estimation}
	Let $Z_1, \dots, Z_k, \tmu_0, \mu \in \R^d$ $r^* = 2\sqrt{k/n} = 2\sqrt{1/m}$.
	Suppose $\normt{\tmu_0 - \mu_0} \le C_0 \sqrt{d}$ where $C_0$ is a universal constant, $M = 1000$, and \Cref{ass:when-at-mu-sdp-str} holds. Then there exists an efficient $\epsilon$-DP algorithm \textsc{fine-estimation} that returns $\tmu^*$ such that if $k \ge \tilde{\cO}\paren{\paren{d + \log\paren{1/\beta}}/\eps}$,
	then
	$\normt{\tmu^* - \mu} \le \cO\paren{r^*}$, with probability $1-\beta$. Note that $\tilde{O}$ is hiding the following lower order logarithmic factors: $\log d$, and $\log m$.
\end{theorem}
\begin{proof}
	Putting together \Cref{thm:halt-estimation} (\textsc{halt-estimation}), \Cref{thm:distance-estimation} (\textsc{distance-estimation}), \Cref{thm:gradient-estimation} (\textsc{gradient-estimation}), and \Cref{thm:gradient-descent} (gradient descent), and setting $T = \log \paren{dm} \ge  \log\paren{d/r^*}$.
\end{proof}
\subsection{Halt Estimation}
If we did not care about privacy, in order to know whether we are close to $\mu$ or not, and we had a current estimate $\tmu$, we could just check if $\SDP$ has score smaller than $0.9k$ for a small enough distance choice or not. If that was the case, $\tmu$ would have been close enough to $\mu$. In this section we privatize this condition by adding Laplace noise to the value of $\SDP$, and argue that this algorithm provides privacy guarantees and high accuracy with high probability. 
\label{subsec:halt-estimation}
\begin{algorithm}[H]
	\caption{\textsc{halt-estimation}}
	\label{alg:halt-estimation}
	\begin{mdframed}
		\mbox{}
		\begin{description}
			\item[Input:]
			Parameters: Data points $Z_1, \dots, Z_k \in \R^d$, current estimate $\tmu$.
			\item[Operation:]
			Let $M = 1000, N = 4M\paren{5M+1} +10 M$\\
			Let $H = \mathbbm{1} \brac{\SDP\paren{\tmu, Nr^*, Z} + \cL\paren{T/\eps}\le 0.91k}$, where $\cL$ is the Laplace distribution. 
			\item[Output:]  $H$
		\end{description}
	\end{mdframed}
\end{algorithm}
\begin{theorem}[\textsc{halt-estimation}]
\label{thm:halt-estimation}
Let $Z_1, \dots Z_k, \tmu, \in \R^d$, $r^* = 2\sqrt{k/n}$, $M\ge 1000$, $N = 4M\paren{5M+1} + 10M$. Suppose $k \ge 100 T\log \paren{T/\beta}/\eps$, and that $\Cref{ass:when-at-mu-sdp-str}$ holds. There exists an $\epsilon/T$-DP algorithm $\textsc{halt-estimation}$ that \textit{succeeds} with probability $1-\beta/T$. Moreover, in the regime where this algorithms succeeds we have
\begin{equation*}
H = 0 \implies \normt{\mu - \tmu} \ge M\paren{5M+1}r^* \mcom
\end{equation*}
and
\begin{equation*}
H = 1 \implies \normt{\mu - \tmu} \le Nr^* + 5Mr^*  \mper
\end{equation*}
\end{theorem}
\begin{proof}
Since $\SDP$ has sensitivity $1$ (\Cref{cor:sdp-sensitiviy}),
from the guarantees of the Laplace mechanism we know that 
\Cref{alg:halt-estimation} is $\epsilon/T$-DP and with probability $1-\beta/T$, it succeeds i.e. $\SDP\paren{\tmu, N r^*, Z} + \cL\paren{T/\epsilon} \in \brac{\SDP\paren{\tmu, Nr^*, Z} - T\log\paren{T/\beta}/\eps, \SDP\paren{\tmu, Nr^*, Z} + T\log\paren{T/\beta}/\eps}$.
Since $k \ge 100 T\log\paren{T/\beta}$, with probability $1-\beta/T$
\begin{equation*}
\SDP\paren{\tmu, Nr^*, Z} + \cL\paren{T/\epsilon} \in \brac{\SDP\paren{\tmu, Nr^*, Z} - 0.1 k, \SDP\paren{\tmu, Nr^*, Z} + 0.1 k} \mper
\end{equation*}
Therefore if we are in the success regime, $H=0$, means that 
\begin{equation*}
\SDP\paren{\tmu, Nr^*, Z} \ge 0.90 k \mcom
\end{equation*}
and $H = 1$, means that 
\begin{equation*}
\SDP\paren{\tmu, Nr^*, Z} \le 0.92k \mper
\end{equation*}
Now by \Cref{cor:sdp-far}, if we are in the success regime
\begin{equation*}
H = 0 \implies \normt{\mu - \tmu} \ge \frac{0.7 Nr^* - 5Mr^*}{2} \ge M \paren{5M+1} r^* \mcom
\end{equation*}
where the last inequality comes from $ N = 4M\paren{5M + 1} + 10M$. Now suppose that we are in the success regime and $H=1$.
Then $\SDP\paren{\tmu, Nr^*, Z} \le 0.92k$, which implies that $\QUAD\paren{\tmu, Nr^*, Z} \le 0.92k$, since $\SDP$ is a relaxation of $\QUAD$.
Therefore for any $v \in \R^d$, where $\normt{v} = 1$ we have that 
\begin{equation*}
\Abs{\Set{i \suchthat \iprod{Z_i - \tmu, v} \le Nr^*}} \ge 0.08 k \mper
\end{equation*}
On the other hand, from \Cref{ass:when-at-mu-simple-str} (which is a result of \Cref{ass:when-at-mu-sdp-str}), we know that for any $v' \in \R^d$, where $\normt{v'} = 1$ we have that
\begin{equation*}
\Abs{\Set{i \suchthat \iprod{Z_i - \mu, v'} \le 5Mr^*}} \ge \paren{1 - \frac{1}{M}} k \mper
\end{equation*}
Let $v = \paren{\mu - \tmu}/\normt{\mu - \tmu}$, $v' = -v$. Then $M \ge 1000$ and $0.08 + \paren{1 - 1/M} > 1$, there exists some index $j$ which is in the intersection of the above sets. Therefore for this $j$ we have that
\begin{equation*}
\iprod{Z_j - \mu + \mu -\tmu, v} \le Nr^* \mcom
\end{equation*} 
therefore
\begin{equation*}
\normt{\mu - \tmu}\le Nr^* + \iprod{Z_i - \mu, -v} \le Nr^* + 5Mr^* \mcom
\end{equation*}
as desired.
\end{proof}

\subsection{Distance Estimation}
\label{subsec:distance-estimation}
Suppose we have a current estimate $\tmu$, and $\mu$ is a central point. 
Furthermore suppose that $\tmu$ and $\mu$ are sufficiently far from each other.
In this section our goal is to find an estimate of the distance $\normt{\tmu - \mu}$.
If we did not care about privacy, we could look for an approximation of the distance to a central point by finding the maximum $r$ such that $\SDP$ obtains value at least $0.9k$ via a binary search up to arbitrary accuracy. However, it's not clear how to privatize this approach. Instead, we look for an $r$ for which $\SDP$ obtains a value between $0.91k$ and $0.94k$, via a private binary search (\Cref{thm:private-binary-search}, \Cref{lem:sdpval-score-search}). If $r$ has this property, then we can use the monotonicity of $\SDP$ in $r$ to argue that  $r$ is an approximation of the distance the central point (\Cref{lem:lower-bound-dist-est}, \Cref{lem:upper-bound-distance-estimation}).
\subsubsection{Algorithm}
\begin{algorithm}[H]
	\caption{\textsc{distance-estimation}}
	\label{alg:distance-estimation}
	\begin{mdframed}
		\mbox{}
		\begin{description}
			\item[Input:]
			Parameters: Data points $Z_1, \dots, Z_k \in \R^d$, current estimate $\tmu$.
			\item[Operation:]
			Let $r^* = 2\sqrt{k/n}$, $S = \log\paren{1.14C_0\sqrt{d}/r^*}$, $s = 0.92k$, $e = 0.93k$. \\
			Do the following binary search for $S$ rounds over $\brac{0, 1.14C_0\sqrt{d}}$, with $r_0$ being the middle point
			\begin{enumerate}
				\item Let $u = \SDP\paren{\tmu, r_t, Z} + N_i$, where $N_i \distras{} \cL \paren{TS/\eps}$.
				\item Let $u$ is larger than $ek$ move right, if it is less than $sk$ move left, if it is neither, move in an arbitrary direction.
			\end{enumerate}
			\item[Output:]  $r_T$.
		\end{description}
	\end{mdframed}
\end{algorithm}
\subsubsection{Main Theorem}
\begin{theorem}[\textsc{distance-estimation}]
\label{thm:distance-estimation}
Let $Z_1, \dots, Z_k, \tmu, \mu \in \R^d$.
Suppose  $r^*= 2\sqrt{k/n}$, $m = n / k$, $\normt{\mu - \tmu} \ge M\paren{5M + 1} r^*$, $k \ge 100 ST\log\paren{ST/\beta}/\eps$, $M \ge 100$, and that \Cref{ass:when-at-mu-sdp-str} holds. Moreover, suppose that $\normt{\mu - \tmu} \le C_0 \sqrt{d}$, where $C_0$ is a universal constant. Then there exists a universal constant $C_2$ and a $\epsilon/T$-dp algorithm \textsc{distance-estimation} using $S$ rounds of binary search that returns $\wh{d}$ such that 
\begin{equation*}
0.99 \normt{\mu - \tmu} \le \wh{d} \le 1.15 \normt{\mu - \tmu} \mcom
\end{equation*}
with probability $1-\beta/T$, where $S = C_2 \Paren{\log m + \log d}$.
We say that $\textsc{distance-estimation}$ has been successful if the above guarantee holds.
\end{theorem}
\begin{proof}
See 
\Cref{subsubsec:proof-dist-est}.
\end{proof}

\subsubsection{Private Binary Search}
The following theorem allows us to perform private binary search over a decreasing function, with low sensitivity.

\begin{theorem}[Private Binary Search]
	\label{thm:private-binary-search}
	Suppose $f \from [0, D] \times \mathcal{X} \to \brac{0, k}$ is a decreasing function and has sensitivity 1 with respect to its second component. Moreover two values $0 \le s < e \le 1$ are given. 
	Let $r_s = \argmax_{f(r) \ge sk - \Delta} r$, $r_e = \argmin_{f(r) \le ek + \Delta} r$,
	then there exists an $\epsilon$-DP algorithm using $S$ rounds of binary search that returns $x$ such that $x \in [r_e - a, r_s + a]$ with probability $1-\beta$, where $\Delta = S \log \paren{S/\beta}/\epsilon$, and $S = \log\paren{D/a}$.
\end{theorem}
\begin{proof}
	
	Suppose $x_i$ is the binary search parameter at step $i$, do the following binary search up to $S$ steps over the domain $\brac{0, D}$ and in the end return $x_S$. 
	\begin{enumerate}
		\item Let $f_i = f(x_i)$, and $f'_i = f_i + N_i$, where $N_i \distras{} \cL\paren{S/\epsilon}$. 
		\item If $f'_i$ is larger than $ek$ move right, if it is less than $sk$ move left, if it is neither, move in an arbitrary direction.
	\end{enumerate}
	Now we need to prove privacy and accuracy guarantees for this algorithm.
	\begin{itemize}
		\item[Privacy] The algorithm is $\epsilon$-DP by basic composition. Note that each step is $\epsilon/S$-DP.
		\item[Accuracy] 
		Note that for every $t \ge 0$
		\begin{equation*}
			\Pr\Brac{\abs{N_i} \ge tS/\epsilon} = \exp\paren{-t}.
		\end{equation*}
		and by union bound
		\begin{equation*}
			\Pr\Brac{\max \; \abs{N_i} \ge tS/\epsilon} \le S\exp\paren{-t}.
		\end{equation*}
		Let $t = -\log(\beta/T)$, since $\Delta = S\log\paren{S/\beta}/\eps$,
		\begin{equation*}
			\Pr\Brac{\max \; \abs{N_i} \ge \Delta} \le \beta.
		\end{equation*}
	\end{itemize}
	Now with probability $1-\beta$, $\forall i, N_i \le \Delta$. Suppose we are in this setting where all of the $N_i$'s are smaller than $\Delta$. We have three cases.
	\begin{enumerate}
		\item If $x_i \in (r_s, D]$, then $f_i \in [0, sk - \Delta)$, and $f'_i \in [-\Delta, sk)$, therefore, the binary search algorithm will choose left.
		\item If $x_i \in \brac{r_e, r_s}$, then $f_i \in \brac{sk - \Delta, ek + \Delta}$, and $f'_i \in \brac{sk - 2\Delta, ek+ \Delta}$, therefore the binary search algorithm might choose right or left (arbitrary direction).
		\item If $x_i \in [0, r_e)$, then $f_i \in (ek + \Delta, k]$, and $f_i' \in (ek, k+\Delta]$, the binary search algorithm will choose right.
	\end{enumerate}
	Therefore this algorithm is performing Binary Interval Search for the interval $\brac{r_e,r_s}$ with accuracy parameter $a$. Therefore, by \Cref{lem:non-private-binary-interval-search}, with probability $1-\beta$, the algorithm returns a point $x$, such that $x \in \brac{r_e - a, r_s + a}$.
\end{proof}
\subsubsection{Proof of \Cref{thm:distance-estimation}}
\label{subsubsec:proof-dist-est}
In this section we prove \Cref{thm:distance-estimation}. This theorem is proven by putting together \Cref{lem:sdpval-score-search}, \Cref{lem:lower-bound-dist-est}, and \Cref{lem:upper-bound-distance-estimation}.

\begin{lemma}[$\SDP$ Binary Search]
	\label{lem:sdpval-score-search}
	Let $Z_1, \dots, Z_k, \tmu, \mu \in \R^d$.
	Suppose  $r^*= 2\sqrt{k/n}$, $m = n / k$, $\normt{\mu - \tmu} \le C_0 \sqrt{d}$, $\normt{\mu - \tmu} \ge M\paren{5M + 1} r^*$, $k \ge 100 ST\log\paren{ST/\beta}/\eps$, $M \ge 100$, $C_1 = 1.14 C_0$, and that \Cref{ass:when-at-mu-sdp-str} holds.
	Let $f\paren{r}:= \SDP\paren{\tmu, r, Z}$, where $f:\brac{0, C_1\sqrt{d}} \to \R$.
	Let $\rho_s = \argmax_{f(r) \ge 0.91k} r,\rho_e = \argmin_{f(r) \le 0.94k} r,$. Then there exists a universal constant $C_2$ and a $\epsilon/T$-dp algorithm using $S$ rounds of binary search that returns $\wh{r}$ such that $\wh{r} \in \brac{\rho_e - r^*, \rho_s + r^*}$ with probability $1-\beta/T$, where $S = C_2 \Paren{\log m + \log d}$.
\end{lemma}
\begin{proof}
	We want to apply the algorithm from \Cref{thm:private-binary-search}.
	In order to do so, we must check that $\SDP$ is decreasing in $r$ (\Cref{lem:sdp-decreasing-r}) and has sensitivity $1$ (\Cref{cor:sdp-sensitiviy}). Furthermore, note that $\set{r \,: \, f\paren{r} \ge 0.91 k, 0 \le r\le C_1 \sqrt{d}}$, and $\set{r \, : \, f\paren{r} \le 0.94k, 0 \le r\le C_1 \sqrt{d}}$ are non-empty sets, which is true because of \Cref{cor:existence-of-r-with-sdp-less-than}, and \Cref{cor:existence-of-good-sol-for-sdp}.
	Now we may apply \Cref{thm:private-binary-search} with 
	\begin{equation*}
	s \gets 0.92k, \quad e \gets 0.93k,\quad D \gets C_1\sqrt{d}, \quad a \gets r^* = 2\sqrt{k/n} \mper
	\end{equation*}
	Therefore, $S = \log\paren{D/a} =  \log \sqrt{\frac{C_1dn}{4k}}$. Therefore, there exists a universal constant $C_2$ such that
	\begin{equation*}
		S \le C_2 \paren{\log d + \log m} \mper
	\end{equation*}
	It remains to bound $\Delta$ by $0.01k$, which is true by the assumption.
	\begin{align*}
		\Delta
		&= 
		ST \log \paren{ST /\beta}/\epsilon  \\
		&\le 0.01 k
	\end{align*}
	Therefore, using the algorithm from \Cref{thm:private-binary-search} we constructed an algorithm that returns $\wh{r} \in \brac{\rho_e + r^*, \rho_s - r^*}$ with probability $1-\beta/T$, using $\cO\paren{\log m + \log d }$ steps of the binary search, as desired.
\end{proof}

\begin{lemma}[lower bound for distance estimate]
	\label{lem:lower-bound-dist-est}
	Under the assumptions of \Cref{lem:sdpval-score-search}, $\wh{r} \ge \rho_e - r^*$ implies $\wh{r} \ge (1-M) \normt{\tmu - \mu}$.
\end{lemma}
\begin{proof}
	Let $\Delta$ be the unit direction $\tmu - \mu$, then by \Cref{ass:when-at-mu-simple-str} (which is implied by \Cref{ass:when-at-mu-sdp-str}),
	\begin{equation*}
		\Abs{\Set{i: \iprod{Z_i - \mu, \Delta} \ge 5M r^*}} \le k/M \mper
	\end{equation*}
	Therefore if we take the direction $-\Delta$, we have for $\paren{1 - 1/M} k$ of the points that
	\begin{equation*}
		\paren{Z_i - \tmu, -\Delta}
		= 
		\iprod{\tmu - \mu + \mu - Z_i, \Delta} 
		\ge
		\normt{\tmu - \mu} - 5 M r^*
	\end{equation*}
	Therefore $\SDP(\tmu, \normt{\tmu - \mu} - 5M r^*, Z) \ge (1 - 1/M) k$, since $\SDP$ is the standard relaxation of $\QUAD$. Now note that by \Cref{lem:sdpval-score-search}, $\SDPVAL(\tmu, \wh{r}, Z) \le 0.94 k$; therefore,
	\begin{equation*}
		\SDP(\tmu, \normt{\tmu - \mu} - 5M r^*, Z) \ge (1 - 1/M) k \ge 0.95 k > 
		0.94 k  \ge \SDP(\tmu, \rho_e, Z) \mper
	\end{equation*}
	Now by monotonicity of $\SDP$ in the second parameter, we conclude that $\rho_e \ge \normt{\tmu - \mu} - 5M r^*$. Therefore 
	\begin{equation*}
		\wh{r} \ge \rho_e - r^* \ge \normt{\tmu - \mu} - (5M +1) r^* \ge (1-\frac{1}{M}) \normt{\tmu - \mu} \mcom
	\end{equation*}
	as desired,
	where the final inequality comes from $\normt{\tmu - \mu} \ge M(5M +1)r^*$.
\end{proof}

\begin{lemma}[upper bound for distance estimation]
\label{lem:upper-bound-distance-estimation}
	Under the same assumptions as in \Cref{lem:sdpval-score-search},
	$\wh{r} \le \rho_s + r^*$ implies $\wh{r} \le 1.15 \normt{\tmu - \mu}$.
\end{lemma}

\begin{proof}
	The proof is similar to the proof in Lemma~5 from \cite{CherapanamjeriFB19}. We show that $\SDP(\tmu, 1.14 \normt{\tmu - \mu}, Z) \le 0.9k$. Suppose the contrary, then there exists some solution $(b,v,B,W,V)$ that achieves value larger than $0.9k$. Let $\cR, \cT$ be the sets from \Cref{lem:RT-lemma-str} where $\psi = 0.9 $, $\psi' = (1+1/M)/1.14$. 
	Now notice that
	\begin{align*}
		\frac{\abs{\cT} + \abs{\cR}}{k} = 
		\frac{0.9 - \frac{1  +1 /M}{1.14}}{1 - \frac{1+1/M}{1.14}} + \paren{0.9 - 1/M} > 1 \mcom
	\end{align*}
	for $M \ge 86$,
	therefore there exists some $j \in \cR \cap \cT$. For this j,
	\begin{align*}
		\paren{1 + 1/M}/1.14 \cdot \paren{1.14\normt{\tmu - \mu}} 
		&\le
		B_{j, j} 1.14 \normt{\tmu - \mu} \\
		&\le
		\iprod{Z_j - x, W_j} \\
		&= 
		\iprod{Z_i - \mu, W_j} 
		+
		\iprod{\mu - \tmu, W_j} \\
		& <
		B_{j, j} 5M r^* + \normt{\tmu - \mu} \mper
	\end{align*}
	Where the last step comes from $\normt{W_j} \le 1$. This is easy to see in the pseudo-expectation language,
	\begin{align*}
	\normt{\pE\brac{b_j v}}^2 &= \sum_\ell \pE\brac{b_j v_\ell}^2  \\
	&\le \sum_\ell \pE\brac{b_j^2}\pE\brac{v_\ell^2} \tag{Cauchy-Schwarz} \\
	&\le \pE\brac{\normt{v}^2} \tag{$\pE\brac{b_j^2} \le 1$} \\
	&= 1 \mper
	\end{align*}
	By rearranging the above inequality we obtain
	\begin{equation*}
		B_{j, j} > 1/M \cdot \normt{\tmu - \mu} \paren{5Mr^*}^{-1} > 1 \mcom
	\end{equation*}
	which is a contradiction. Therefore,
	\begin{equation*}
		\SDP\paren{\tmu, 1.14\normt{\tmu - \mu}, Z} \le 0.9k < 0.91k \le \SDP\paren{\tmu, \rho_s, Z} \mper
	\end{equation*}
	Now by monotonicity of $\SDP$ in the second parameter, we conclude that $\rho_s \le 1.14 \normt{\tmu - \mu}$.
	Therefore
	\begin{equation*}
		\wh{r} \le \rho_s + r^* \le 1.14 \normt{\tmu - \mu} + r^* \le 1.15 \normt{\tmu -\mu} \mcom
	\end{equation*}
	as desired, where the final inequality comes from $\normt{\tmu - \mu} \ge 100(5 \times 100 + 1) r^* \ge 100r^*$.
\end{proof}
\begin{corollary}[existence of $r$ with $\SDP$ less than $0.9k$]
\label{cor:existence-of-r-with-sdp-less-than}
In the proof of \Cref{lem:upper-bound-distance-estimation}, we have proven that under \Cref{ass:when-at-mu-sdp-str}, and the assumption that $M\ge 86$,
\begin{equation*}
\SDP\paren{\tmu, 1.14\normt{\mu - \tmu}, Z} \le 0.9k \mper
\end{equation*}
\end{corollary}
\subsection{Gradient Estimation}
\label{subsec:grad-estimation}
Suppose we have a current estimate $\tmu$ and $\mu$ is a central point Moreover, $\tmu$ and $\mu$ are sufficiently far from each other. Our goal is to privately find a direction $y$ such that $\iprod{y, \paren{\mu - \tmu}/\normt{\mu - \tmu}} \ge \Omega\paren{1}$.
In order to do this we run the exponential mechanism over a ball of radius $1 - \Theta(1)$ in $\R^d$ with score function $\SDPVAL$. We prove that $\SDPVAL$ has low sensitivity (\Cref{subsubsec:sdp-val-sensitivity}), is Lipschiz (\Cref{subsubsec:sdp-val-lipschitzness}), and is concave (\Cref{subsubsec:sdp-val-concavity}). These properties ensure that we can sample from the corresponding distribution efficiently and that our algorithm is private. In other to prove the accuracy of the exponential mechanism (\Cref{subsubsec:sdp-val-exp-acc}), we prove that there exists a large volume of points that have high scores (\Cref{subsubsec:sdp-val-volume}). If the exponential mechanism is accurate, it outputs a direction $y$ from the ball that has high a high score. In \Cref{subsubsec:sdpval-sos-analysis}, we show that if $y$ has a high score it satisfies $\Abs{\iprod{y, \paren{\mu - \tmu}/\normt{\mu - \tmu}}} \ge \Omega\Paren{1}$. We wrap up the guarantees of the exponential mechanism with $\SDPVAL$ in \Cref{thm:sdpval-exponential-mechanism}.

Now in order to decide between $y$ and $-y$, we run the exponential mechanism over these two candidates with score function being the number $Z_i - \tmu$'s that have positive inner product with each candidate. We show that at least one of the two candidates must have a high score and therefore with high probability it is outputted. Moreover the higher scoring candidate satisfies $\iprod{y, \paren{\mu - \tmu}/\normt{\mu - \tmu}} \ge \Omega\paren{1}$. See   \Cref{thm:deciding-between-2} for the theorem statement, and \Cref{subsubsec:proof-deciding-between-2} for its proof.

Finally, we put together the guarantees of
\Cref{thm:sdpval-exponential-mechanism},
and
 \Cref{thm:deciding-between-2} to obtain \Cref{thm:gradient-estimation}.
\subsubsection{Algorithm}
\begin{algorithm}[H]
	\caption{\textsc{gradient-estimation}}
	\label{alg:gradient-estimation}
	\begin{mdframed}
		\mbox{}
		\begin{description}
			\item[Input:]
			Parameters: Data points $Z_1, \dots, Z_k \in \R^d$, current estimate $\tmu \in \R^d$, estimate of the distance $d$
			\item[Operation:]
			\begin{enumerate}
				\item Let $M=1000$, $r = d/1.2$, $\zeta = 1/\paren{2M}$.
				\item Run the exponential mechanism over the ball of radius $1-\zeta$ with score function $\Score\paren{y} = \SDPVAL\paren{y; \tmu, r, Z}$ to obtain $y_0$.
				\item Run the exponential mechanism with score function $\Score\paren{y} = \sum_{i=1}^k \mathbbm{1} \brac{\iprod{Z_i - \tmu, y} > 0}$ over $\Set{y_0, -y_0}$ to obtain $\wh{y}$.
			\end{enumerate}
			\item[Output:]  $\wh{y}$
		\end{description}
	\end{mdframed}
\end{algorithm}
\subsubsection{Main Theorems}
The following theorem guarantees that the exponential mechanism in the second line of \Cref{alg:gradient-estimation}, is private and can be run efficiently and outputs a direction $y_0$ such that $\Abs{\iprod{y_0, \Delta}} \ge \Omega{1}$ with high probability.
\begin{theorem}[$\SDPVAL$ exponential mechanism]
\label{thm:sdpval-exponential-mechanism}
Let $Z_1 \dots Z_k, \tmu, \mu \in \R^d$, $r^* = 2\sqrt{k/n}$, $M \ge 1000$ , $d \ge 0$, $r = d/1.2$,
$\Delta = \paren{\mu - \tmu}/\normt{\mu - \tmu}$
Suppose \Cref{ass:when-at-mu-sdp-str} holds, $1.15 \normt{\mu - \tmu} \ge d \ge 0.9 \normt{\mu - \tmu}$,
$\normt{\mu - \tmu} \ge M \paren{5 M r^*}$, $\zeta = 1/\paren{2M}$. 
Then there exists a constant $C'_M$ depending only on $M$ such that the exponential mechanism in \Cref{alg:gradient-estimation} with score function $f\paren{y} = \SDPVAL\paren{y; \tmu, r, Z}$ and privacy budget $\epsilon/T$, over the ball of radius $1-\zeta$ in $\R^d$, returns $y_0$ such that if 
$k \ge C'_M T \paren{ d + \log\paren{T/\beta}}/\eps$,
then
\begin{equation*}
\Abs{\iprod{y_0, \Delta}} \ge 0.1 \mcom
\end{equation*} 
with probability $1- T/\beta$.
Moreover there exists an algorithm taking time $\poly\paren{n, d, 1/\epsilon}$ for sampling from the exponential mechanism.
\end{theorem}
\begin{proof}
Combination of \Cref{lem:sdp-val-sensitivity} (sensitivity, privacy), \Cref{thm:accuracy-exponential-sdpval} (accuracy of the exponential mechanism),  \Cref{thm:accuracy-ng} (accuracy guarantee of $\SDPVAL$), and \Cref{lem:runtime-exponential-sdpval} (runtime of the exponential mechanism). See \Cref{subsubsec:sdpval-sos-analysis} to \Cref{subsubsec:sdpval-runtime-analysis} for the full proof.
\end{proof}
From the previous theorem we know that with high probability $\Abs{\iprod{y_0, \tmu - \mu}} \Omega\paren{1}$
The following theorem guarantees that the third line of the algorithm is private and that it has accuracy. That is, conditioned on the success of the previous step, with high probability it outputs a vector $\wh{y}$ from $\set{y_0, -y_0}$ such that $\iprod{\wh{y}, \paren{\mu - \tmu}/\normt{\mu - \tmu}} \ge \Omega\paren{1}$.

\begin{theorem}[exponential mechanism for the two directions]
\label{thm:deciding-between-2}
Let $Z_1, \dots Z_k, \tmu, \mu, y_0 \in \R^d$, 
$r^* = 2\sqrt{k/n}$,
$\normt{\mu - \tmu} \ge M\paren{5 M r^*}$,
$\Delta = \paren{\mu - \tmu}/\normt{\mu - \tmu}$
and
$M \ge 1000$.
Suppose $\Abs{\iprod{y_0, \Delta}} \ge 0.1$, and that \Cref{ass:when-at-mu-sdp-str} holds.
Then
the exponential mechanism with score function $\Score\paren{y} = \sum_{i=1}^k \mathbbm{1}\brac{\iprod{Z_i - \tmu, y} > 0}$, over $\Set{y_0, -y_0}$, and privacy budget $\epsilon/T$, as in \Cref{alg:gradient-estimation}, returns $\wh{y}$ such that if $k \ge 2MT\paren{\log 2  + \log T/\beta}/\eps$, then
\begin{equation*}
\iprod{\wh{y}, \Delta} \ge 0.1 \mcom
\end{equation*}
with probability $1-\beta/T$.
\end{theorem}
\begin{proof}
See \Cref{subsubsec:proof-deciding-between-2}.
\end{proof}
Putting together \Cref{thm:sdpval-exponential-mechanism}, and  \Cref{thm:deciding-between-2} gives us the following algorithm, which guarantees that \Cref{alg:gradient-estimation} is private, has high utility with high probability, and that it can be run efficiently.
\begin{theorem}[\textsc{gradient-estimation}]
\label{thm:gradient-estimation}
Let $Z_1, \dots Z_k, \tmu, \mu \in \R^d$, $d \ge 0$, $r^* = 2\sqrt{k/n}$, $M\ge 1000$, $d\ge 0$, $\Delta = \paren{\mu - \tmu}/\normt{\mu - \tmu}$. Suppose  \Cref{ass:when-at-mu-sdp-str} holds. Then there exists a constant $C''_M$ depending only on $M$ and  $\epsilon/T$-dp algorithm \textsc{gradient-estimation} that if $0.99 \normt{\mu - \tmu} \le d \le 1.15 \normt{\mu - \tmu}$,  $\normt{\mu - \tmu} \ge M\paren{5Mr^*}$, and $k \ge C''_M T\paren{d + \log T/\beta}/\eps$, returns $\wh{y}$ such that 
\begin{equation*}
\iprod{\wh{y}, \Delta} > 0.1 \mcom
\end{equation*}
with probability $1 - \beta/T$. Moreover this algorithm has runtime $\poly \paren{n, d, 1/\eps}$. We say that $\textsc{gradient-estimation}$ has been successful if the above guarantee holds.
\end{theorem}

In the rest of this section we prove \Cref{thm:sdpval-exponential-mechanism}, and  \Cref{thm:deciding-between-2}.
\subsubsection{SoS Analysis}
\label{subsubsec:sdpval-sos-analysis}

\begin{theorem}[accuracy guarantee of $\SDPVAL$]
	\label{thm:accuracy-ng}
	Let $Z_1 \dots Z_k, \tmu, \mu \in \R^d$, $r^* = 2\sqrt{k/n}$, and $\Delta = \paren{\mu - \tmu}/\normt{\mu - \tmu}$.
	Suppose \Cref{ass:when-at-mu-sdp-str} holds and $M \ge 1000$, $d \ge 0.99 \normt{\mu - \tmu}$, and $\normt{\mu - \tmu} \ge M\paren{5Mr^*}$.
	Moreover, suppose $\SDPVAL\paren{y; \tmu, d/1.2 , Z} =\psi k$, where $\psi \ge 0.9$.
	Then
	\begin{equation*}
		\abs{\iprod{y, \Delta}} \ge 0.1 \mper
	\end{equation*}
\end{theorem}
\begin{proof}
	Let $r = d/1.2$, note that since $d \ge 0.99 \normt{\mu - \tmu}$, $r \ge 0.8 \normt{\mu - \tmu}$.
	Suppose $\pE$ is the solution to $\SDPVAL\paren{y; \tmu, r, Z}$ as in \Cref{def:SDPVAL}. Therefore, 
	\begin{align*}
		&y = \pE\Brac{v} \\
		&\pE\Brac{b_i^2} = \pE\Brac{b_i}
		\numberthis
		\label{eq:expected-b-squared-ng}
		\\
		&\sum_i \pE\Brac{v_i^2} \le 1 
		\numberthis
		\label{eq:expected-v-squared-ng}
		\\
		&\pE\Brac{b_i^2} \cdot r \le \iprod{Z_i - \tmu,\pE \Brac{b_i v}} \mper
	\end{align*}
	Moreover, $y$ has score $\psi k$, i.e.
	\begin{equation}
		\sum_i \pE\Brac{b_i^2} = \psi k.
		\label{eq:high-score-ng}
	\end{equation}
	Let $\cT$ be the set that its existence is guaranteed from \Cref{lem:RT-lemma-str}. Then $\Abs{\cT} \ge \paren{\psi - 1/M}k$, and
	\begin{equation}
		\forall i \in \cT:
		\pE\brac{\iprod{Z_i - \mu, b_i v}} < 5M r^* \pE\brac{b_i^2} \mper
		\label{eq:T-elements}
	\end{equation}
	Also note that 
	\begin{equation}
		\paren{2 \psi -1 - 1/M} k \le \sum_{i \in \cT} \pE\brac{b_i^2} \le \psi k \mper
		\label{eq:T-lb-ub}
	\end{equation}
	Therefore we may write
	\begin{align*}
		r\psi k &=
		r \cdot \sum_i \pE\brac{b_i^2} \\
		& = 
		r \cdot \Brac{ \sum_{i \in \cT} \pE\brac{b_i^2} +  \sum_{i \notin \cT} \pE\brac{b_i^2} } \\
		&\le
		\sum_{i \in \cT} \iprod{Z_i - \tmu, \pE\brac{b_i v}} + \paren{1 - \psi + 1/M} k r \\
		&=
		\sum_{i \in \cT} \pE\brac{b_i \iprod{Z_i - \tmu, v }} + \paren{1 - \psi + 1/M} k r \\
		&=
			\sum_{i \in \cT} \pE\brac{b_i \iprod{Z_i - \mu, v}}
		+
			\sum_{i \in \cT} \pE\brac{b_i \iprod{\mu - \tmu, v - \pE\brac{v}}}
		+
			\sum_{i \in \cT} \pE\brac{b_i} \iprod{\mu - \tmu,\pE\brac{v}}
		+
		\paren{1 - \psi + 1/M} k r \mper
	\end{align*}
	Now we need to bound each term separately.
	\begin{enumerate}
		\item 
			$\sum_{i \in \cT} \pE\brac{b_i \iprod{Z_i - \mu, v}}$.
		\begin{align*}
			\sum_{i \in \cT} \pE\brac{b_i \iprod{Z_i - \mu, v}}
			&\le 
			5M r^*  \sum_{i \in \cT} \pE\Brac{b_i^2}
			\tag{\Cref{eq:T-elements}}
			\\
			&\le
			5M k r^*
		\end{align*}
		\item
			$\sum_{i \in \cT}  \pE\brac{b_i \iprod{\mu - \tmu, v - \pE\brac{v}}}$.
		\begin{align*}
			\sum_{i \in \cT}  \pE\brac{b_i \iprod{u - \tmu, v-\pE\brac{v}}}
			&\le
			\iprod{\mu - \tmu, \pE\brac{ \paren{\sum_{i \in \cT} b_i} \paren{v - \pE\brac{v}}}} \\
			&\le
			\iprod{\mu - \tmu, \pE\brac{\sum_{i \in \cT} \paren{b_i - \pE\brac{b_i}}\paren{v - \pE\brac{v}}}}
			\\
			&\le
			\normt{\mu - \tmu} 
			\cdot
			\normt{\pE\brac{\sum_{i \in \cT} \paren{b_i - \pE\brac{b_i}}\paren{v - \pE\brac{v}}}} \\
			&\le
			\normt{\mu - \tmu} 
			\cdot
			\sqrt{
			\pE\brac{\paren{\sum_{i \in \cT} \paren{b_i - \pE\brac{b_i}}}^2}
			\cdot \pE \brac{ \snormt{v - \pE\brac{v}}}}
			\tag{\Cref{lem:sos-cauchy-schwarz}}
		\end{align*}
		Now we need to bound the two last terms.
		\begin{align*}
			\pE\Brac{\Paren{\sum_{i \in \cT} \paren{b_i - \pE\brac{b_i}}}^2} 
			&=
			\pE\Brac{\Paren{\sum_{i \in \cT} b_i - \pE\brac{\sum_{i \in \cT} b_i}}^2} \\
			&=
			\pE\brac{\paren{\sum_{i \in \cT} b_i}^2} - 
			\pE \brac{\sum_{i \in \cT} b_i}^2
			\tag{\Cref{lem:pseudo-covariance}}
			\\
			&\le
			\pE\brac{k \cdot \sum_{i \in \cT} b_i^2} - 
			\Paren{\sum_{i \in \cT} \pE\brac{b_i}}^2
			\tag{Cauchy-Schwarz}
			\\
			&=
			k \sum_{i \in \cT} \pE\brac{b_i^2} - \Paren{\sum_{i \in \cT} \pE\brac{b_i^2}}^2
			\tag{\Cref{eq:expected-b-squared-ng}}
			\\
			&\le
			\psi k^2 - \paren{2\psi - 1- /1M}^2 k^2
			\tag{\Cref{eq:T-lb-ub}} 
			\\
			&= 
			k^2 \paren{\psi - \paren{2\psi - 1-  1/M}^2}
		\end{align*}
		\begin{align*}
			\pE \brac{\snormt{v - \pE\brac{v}}}
			&=
			\pE\brac{\snormt{v}} - \snormt{\pE\brac{v}}
			\tag{\Cref{lem:pseudo-covariance}}
			\\
			&\le
			\pE\brac{\snormt{v}} \\
			&\le
			1 \tag{\Cref{eq:expected-v-squared-ng}}
		\end{align*}
		Therefore,
		\begin{align*}
			\sum_i \pE\brac{b_i \iprod{u - \tmu, v-\pE\brac{v}}}
			&\le
			k \sqrt{\psi - \paren{2\psi - 1 -1/M}^2} \cdot \normt{\mu - \tmu} \mper
		\end{align*}
		\item
			$\sum_{i \in \cT} \pE\brac{b_i} \iprod{\mu - \tmu,\pE\brac{v}}$.
		\begin{align*}
			\sum_{i \in \cT} \pE\brac{b_i}
			&=
			\sum_{i \in \cT} \pE\brac{b_i^2} 
			\tag{\Cref{eq:expected-b-squared-ng}} 
			\\
			&\le
			\psi k 
			\tag{\Cref{eq:T-lb-ub}}
		\end{align*}
		Therefore,
		\begin{equation*}
			\sum_{i \in \cT} \pE\brac{b_i} \iprod{\mu - \tmu,\pE\brac{v}}
			\le
			\psi k \abs{\iprod{\mu -\tmu, \pE\brac{v}}} \mper
		\end{equation*}
	\end{enumerate}
	Putting all parts together we obtain
	\begin{align*}
		\paren{2\psi-1 - 1/M} r k &\le
			\sum_i \pE\brac{b_i \iprod{Z_i - \mu, v}}
		+
			\sum_i \pE\brac{b_i \iprod{\mu - \tmu, v - \pE\brac{v}}}
		+ 
			\sum_i \pE\brac{b_i} \iprod{\mu - \tmu,\pE\brac{v}}
		\\
		&\le
			5M k r^*
		+
			k \sqrt{\psi - \paren{2\psi - 1 -1/M}^2} \cdot \normt{\mu - \tmu}
		+
			\psi k
			\abs{\iprod{\mu - \tmu, \pE\brac{v}}}  \mper
	\end{align*}
	Therefore if $\Delta = \paren{\mu - \tmu}/\normt{\mu - \tmu}$,
	$\tilde{d} = \normt{\mu - \tmu}$.
	\begin{equation*}
		\frac{1}{\psi} \cdot \brac{
			\paren{2\psi - 1 - 1/M} \cdot \frac{r}{\tilde{d}}
			- 
			\frac{5M r^*}{\tilde{d}}
			- 
			\sqrt{\psi - \paren{2\psi - 1 - 1/M}^2}
		}
		\le
		\abs{\iprod{\Delta, y}} \mper
	\end{equation*}
	Now since $0.9 \le \psi \le 1, M \ge 1000, r \ge 0.8 \tilde{d}, \tilde{d} \ge M\paren{5Mr^*}$, we get that
	\begin{align*}
		\abs{\iprod{\Delta, y}} 
		&\ge
		0.65 \paren{2\psi - 1 - 1/M} - \sqrt{\psi - \paren{2\psi - 1 - 1/M}^2} + 0.15 \paren{2\psi - 1 - 1/M} - \frac{5 M r^*}{\tilde{d}}    \\
		&\ge 0.15 \paren{2 \times 0.9 - 1 - 1/1000} - 1/1000  \\
		&> 0.1 \mper
	\end{align*}
\end{proof}
The following corollary is useful in \Cref{subsec:halt-estimation}.
\begin{corollary}[$\SDP$ far]
\label{cor:sdp-far}
Let $Z_1, \dots, Z_k, \tmu, \mu \in \R^d, r^* = 2\sqrt{k/n}$. Suppose \Cref{ass:when-at-mu-sdp-str} holds and $M \ge 1000$. Suppose $\SDP\paren{\tmu, r, Z} \ge 0.9k$. Then
\begin{equation*}
\normt{\mu - \tmu} \ge \frac{0.7 r - 5 Mr^*}{2} \mper
\end{equation*}
\end{corollary}
\begin{proof}
From the proof of \Cref{thm:accuracy-ng}, we have that
\begin{equation*}
2\normt{\mu - \tmu} + 5 k r^* \ge rk\paren{2 \times 0.9 - 1 - 1/M} \mper
\end{equation*}
Rearranging the terms and noting that $M\le 1000$, completes the proof.
\end{proof}

\subsubsection{Analysis of the $\SDPVAL$ Exponential Mechanism}
In order to analyze the exponential mechanism we need to analyze the following properties of $\SDPVAL$.
\begin{enumerate}
	\item Sensitivity: for privacy analysis of the exponential mechanism.
	\item Lipschitzness: for runtime analysis of the exponential mechanism.
	\item Concavity: for runtime analysis of the exponential mechanism.
	\item Volume of the set of \emph{optimal} points: for accuracy analysis of the exponential mechanism.
\end{enumerate}
In the following subsections we analyze these properties.

\subsubsection{Sensitivity}
\label{subsubsec:sdp-val-sensitivity}
\begin{lemma}[$\SDPVAL$ sensitivity]
	\label{lem:sdp-val-sensitivity}
	$\Delta_{\SDPVAL} \le 1$.
\end{lemma}
\begin{proof}
	We want to prove that for every two neighboring datasets $Z$ and $Z'$
	\begin{equation*}
		\max_{y, \tmu, r} \; \max_{Z, Z'} \; \abs{\SDPVAL\paren{y;\tmu, r Z} - \SDPVAL\paren{y;\tmu, r, Z'}} \le 1 \mper
	\end{equation*}
	If we prove that for every $y, \tmu, r$ and for every neighboring datasets $Z$, $Z'$
	\begin{equation*}
		\SDPVAL\paren{y; \tmu, r, Z} - 1 \le \SDPVAL\paren{y; \tmu, r, Z'},
	\end{equation*}
	by symmetry, our desired statement will be implied.
	Suppose that the optimal solution for $\paren{y;\tmu, r, Z}$ is $\paren{1, b , B, v, V, W}$. Without loss of generality, suppose that $Z$ and $Z'$ differ in the first index.
	Now consider $\paren{1, b', B', v, V, W'}$, where it is equal to $\paren{1, b, B, v, V, W}$ everywhere, except at 
	$b_1, B_{1,\cdot}, B_{\cdot, 1}, W_{1, \cdot}$, where it is equal to zero, or in other words, in the second column and row of the PSD matrix.
	Then $\paren{1, b', B', v, V, W'}$ will still satisfy the linear inequalities. Moreover the corresponding matrix for $\paren{1, b', B', v, V, W'}$, will be positive semidefinite, by \Cref{lem:nulling-psd-matrices}.
	Now note that $\Tr{B'} \ge \Tr{B} - 1$, and therefore,
	\begin{equation*}
		\SDPVAL\paren{y;\tmu, r, Z} - 1 \le \SDPVAL\paren{y;\tmu, r, Z'} \mcom
	\end{equation*}
	which completes the proof.
\end{proof}
The following corollary is useful in the analysis in \Cref{subsec:distance-estimation}.
\begin{corollary}[$\SDP$ sensitivity] $\Delta_{\SDP} \le 1$.
\label{cor:sdp-sensitiviy}
\end{corollary}
\begin{proof}
Same as \Cref{lem:sdp-val-sensitivity}.
\end{proof}

\subsubsection{Lipschitzness}
\label{subsubsec:sdp-val-lipschitzness}
\begin{lemma}
	\label{lem:pol-identity}
	Let $Z_1, \dots, Z_k, \tmu \in  \R^d, r \in \R, w > 0$. Then there exists $\alpha \in \R^d; \beta \in \R^k; \gamma \in \R; \lambda \in \R_{\ge 0}^k$, and a degree-$2$ SoS polynomial $\sigma \in \R\brac{b_1 \dots, b_k, v_1, \dots, v_k}_{\le 2}$ such that the following polynomial identity holds in variables $b_1, \dots, b_k, v_1, \dots, v_d$.
	\begin{align*}
		\SDPVAL\paren{y; \tmu, r, Z} + \omega - \sum_{i= 1}^k b_i^2 &= 
		\sum_{i=1}^d \alpha_i \paren{y_i - v_i} +
		\sum_{i=1}^k \beta_i \paren{b_i - b_i^2} +
		\gamma \paren{1 - \snormt{v}} \\& +
		\sum_{i=1}^k \lambda_i \paren{\iprod{Z_i - \tmu, bv} - b_i^2 r} +
		\sigma\paren{b, v}
	\end{align*}
\end{lemma}
This lemma is a direct result of duality of pseudo-distributions and sos proofs (e.g. see \cite{barak2016proofs}).
Note that the constraints on $b$ and $v$ in the optimization problem imply boundedness (\Cref{lem:boundedness-quadratic-optimization-problem}) and therefore our problem satisfies the duality conditions.

Now we use this lemma to prove the following theorem.
\begin{lemma}[$\SDPVAL$ Lipschitzness]
	\label{lem:sdpval-lipschitzness}
	$f\paren{y} := \SDPVAL\paren{y; \tmu, r, Z}$ is $L = \sqrt{d}k/\zeta$-Lipschitz over the ball of radius $\paren{1 - \zeta}$ in $\R^d$, where $0 < \zeta \le 1$.
\end{lemma}
\begin{proof}
	Let $\omega > 0$, then by \Cref{lem:pol-identity}, there exists $\alpha \in \R^d; \beta \in \R^k; \gamma \in \R; \lambda \in \R_{\ge 0}^k$ and a degree-$2$ SoS polynomial $\sigma$ such that 
	\begin{align*}
		\SDPVAL\paren{y; \tmu, r, Z} + \omega - \sum_{i= 1}^k b_i^2 &= 
		\sum_{i=1}^d \alpha_i \paren{y_i - v_i} +
		\sum_{i=1}^k \beta_i \paren{b_i - b_i^2} +
		\gamma \paren{1 - \snormt{v}} \\& +
		\sum_{i=1}^k \lambda_i \paren{\iprod{Z_i - \tmu, bv} - b_i^2 r} +
		\sigma\paren{b, v} \mper
	\end{align*}
	First we prove that $\alpha_i$s are bounded.
	Note that the above equality, is a polynomial identity, therefore we may pick different values for $b$ and $v$, and the equality should still hold.
	For $1\le j \le k$ let $b = 0, v = y + te_j$ where $t$ is chosen such that $\normt{y + te_j} = 1$. Such $t$ exists because $\normt{y} < 1 - \zeta$. Moreover $t > \zeta$ due to the triangle inequality.
	Now if we plug $b$ and $v$ into the above equality we obtain
	\begin{align*}
		\SDPVAL\paren{y; \tmu, r, Z} + \omega &= 
		-\alpha_j t +
		\sigma\paren{0, y + te_j}
		\tag{$\sigma \ge 0, \gamma_i \ge 0$}
		\\
		&\ge -\alpha_j t \mper
	\end{align*}
	Therefore, for every $1 \le j\le k$,
	\begin{equation*}
		\alpha_j \ge -\frac{k + \omega}{\zeta} \mper
	\end{equation*}
	Similarly by letting $b = 0, v = y - e_j t$, we may obtain
	$\alpha_j \le \paren{k+\omega}/\zeta$, therefore $\abs{\alpha_j} \le \paren{k+\omega}/\zeta$, and
	\begin{equation*}
		\normt{\alpha} \le \sqrt{d}\paren{k+\omega}/\zeta \mper
	\end{equation*}
	Now that we have bounded $\alpha$, we can move on to proving Lipschitzness. Let
	\begin{align*}
		g\paren{\alpha, \beta, \gamma, \lambda, \sigma; y; \tmu, r, Z} 
		&=
		\sum_{i= 1}^k b_i^2 +
		\sum_{i=1}^d \alpha_i \paren{y_i - v_i} +
		\sum_{i=1}^k \beta_i \paren{b_i - b_i^2} +
		\gamma \paren{1 - \snormt{v}} \\& +
		\sum_{i=1}^k \lambda_i \paren{\iprod{Z_i - \tmu, bv} - b_i^2 r} +
		\sigma\paren{b, v} \mper
	\end{align*}
	Note that the co-domain of $g$ is the set of degree-$2$ polynomials in $\paren{b, v}$.
	Suppose $y' = y + \delta u$, where $\normt{u} = 1$.
	In order to prove Lipschitzness, we prove that there exists
	$\alpha', \beta', \gamma', \lambda', \sigma'$
	such that 
	$g\paren{\alpha', \beta', \gamma', \lambda', \sigma'; y'; \tmu, r, Z}$
	is a degree-$0$ polynomial (i.e. it is a number in $\R$), and is at most 
	$\SDPVAL\paren{y; \tmu, r, Z} + L\delta + h\paren{\omega}$, where $\lim_{\omega \to 0} h\paren{\omega} = 0$.
	If we do so, by pseudo-expectation and SoS duality, we can conclude that for every $\omega > 0$, 
	$\SDPVAL\paren{y'; \tmu, r, Z} < \SDPVAL\paren{y; \tmu, r, Z}+ L\delta + h\paren{\omega}$.
	Taking $\lim_{\omega \to 0}$ from both sides gives us
	$\SDPVAL\paren{y'; \tmu, r, Z} \le \SDPVAL\paren{y; \tmu, r, Z}+ L \delta$.
	By symmetry, this implies that $\Abs{\SDPVAL\paren{y'; \tmu, r, Z} - \SDPVAL\paren{y; \tmu, r, Z}} \le L\delta$ and that  $f$ is $L$-Lipschitz in $y$.
	It remains to construct a dual solution that satisfies the described properties.
	To do so, take $\paren{\alpha', \beta', \gamma', \lambda', \sigma'} = \paren{\alpha,\beta, \gamma, \lambda, \sigma}$. Then
	\begin{align*}
		g\paren{\alpha', \beta', \gamma', \lambda', \sigma'; y'; \tmu, r, Z} 
		&=
		g\paren{\alpha, \beta, \gamma, \lambda, \sigma; y; \tmu, r, Z} + \delta \sum_{i=1}^d \alpha_i u_i \\
		&\le
		\SDPVAL\paren{y; \tmu, r, Z} + \omega + \delta \normt{\alpha} \\
		&\le
		\SDPVAL\paren{y; \tmu, r, Z} + \omega + \delta \Paren{ \sqrt{d}k/\zeta+ \sqrt{d}\omega/\zeta} \\
		&= 
		\SDPVAL\paren{y; \tmu, r, Z} + \delta \sqrt{d}k/\zeta + \omega\paren{\delta\sqrt{d}/\zeta +1 } \mper
	\end{align*}
	as desired. Therefore $f$ is $L = \sqrt{d}k/\zeta$-Lipschitz.
\end{proof}

\subsubsection{Concavity}
\label{subsubsec:sdp-val-concavity}
\begin{corollary}[$\SDPVAL$ concavity]
	\label{cor:sdp-val-concavity}
	Let $Z_1, \dots, Z_k, \tmu \in \R^d$, $r\ge 0$.
	Let $f\paren{y}:  =\SDPVAL\paren{y; \tmu, r, Z}$. Then $f$ is concave.
\end{corollary}
\begin{proof}
	Direct result of \Cref{lem:convacity-maximization}.
\end{proof}
\subsubsection{Volume}
\label{subsubsec:sdp-val-volume}
In this section we prove the following lemma.
\begin{lemma}[$\SDPVAL$ volume ratio]
\label{lem:volume-ratio-sdp-val}
Let $Z_1, \dots, Z_k, \tmu, \mu \in \R^d, r^* = 2\sqrt{k/n}, 1/\paren{2M} \ge \zeta \ge 0$. Suppose $\normt{\tmu - \mu} \ge M\paren{5M} r^*$, $M \ge 1000$, $d \le 1.15\normt{\mu - \tmu}$,and that \Cref{ass:when-at-mu-sdp-str} holds. Then there exists a set $\cH^*$ inside $\bbB_{1-\zeta}^d$, the ball of radius $1-\zeta$ in $\R^d$, and a constant $C_{M}$ depending only on $M$ such that 
\begin{equation*}
\log \frac{\vol \paren{\bbB^{d}_{1-\zeta}}}{\vol\paren{\cH^*}} \le C_{M} d \mcom
\end{equation*}
and
\begin{equation*}
\inf_{y \in \cH^*} \SDPVAL\paren{y; \tmu, d/1.2, Z} \ge \paren{1 - \frac{1}{M}} k \mper
\end{equation*}
\end{lemma}
\begin{proof}
Putting together \Cref{lem:goodness-set}, \Cref{lem:size-good-set}, and \Cref{cor:size-of-the-good-set}, and setting $\gamma = \theta = 1/M$, and noting that $d \le 1.15 \normt{\mu - \tmu}$ implies  
\begin{equation*}
\frac{d}{1.2} \le 0.96 \normt{\mu - \tmu} \le \paren{1 - 1/M}\paren{1- 2/M}\normt{\mu - \tmu} \mper
\end{equation*}
Also note that from \Cref{cor:size-of-the-good-set}, $C_M \ge \cO\Paren{\log M}$ should suffice.
\end{proof}

The following lemma tells us that there's a set of good points. Later in this section we show that this set has a large volume inside the ball of radius $1 - \zeta$.
\begin{lemma}[goodness of the set]
	\label{lem:goodness-set}
	Let $Z_1 \dots, Z_k , \tmu\in \R^d $, $r^* = 2\sqrt{k/n}$.
	Assume \Cref{ass:when-at-mu-sdp-str} holds. Let  $\tilde{d} = \normt{\mu - \tmu}$, $\Delta = \paren{\mu - \tmu}/\normt{\mu - \tmu}$. Then for any $0\le \gamma, \theta \le 1$, if $\iprod{u, \Delta} \ge 1 - \gamma$, $\normt{u} = 1$, $\tilde{d} \ge M \paren{5M}r^*$, $ \paren{1 - \theta} \paren{1- \gamma - 1/M} \tilde{d} \ge r$,
	\begin{equation*}
		\forall \alpha \in \brac{1- \theta, 1}: \SDPVAL\paren{\alpha u; \tmu,r, Z} \ge \paren{1- \frac{1}{M}} k \mper
	\end{equation*}
\end{lemma}
\begin{proof}
	By \Cref{ass:when-at-mu-simple-str} (which is a result of \Cref{ass:when-at-mu-sdp-str}) we know that for the direction $u$ there exists a set $S_{-u}$ such that $\Abs{S_{-u}} \ge k \paren{1 - 1/M}$ and 
	\begin{equation*}
		\forall i \in S_{-u}: \iprod{Z_i - \mu, -u} \le 5 M r^* \mper
	\end{equation*}
	Now for every $i \in S_{-u}$, we have that
	\begin{align*}
		\iprod{Z_i - \tmu, u} 
		&= \iprod{Z_i - \mu, u} + \iprod{\mu - \tmu, u} \\
		&\ge - 5 Mr^* + \normt{\mu - \tmu} \iprod{u, \Delta} \\
		&\ge \paren{1- \gamma}\tilde{d}  - 5Mr^* \\
		&\ge  \paren{1 - \gamma - 1/M} \tilde{d} \mper
	\end{align*}
	Therefore for at least $k\paren{1 - 1/M}$ indices
	\begin{equation*}
		\iprod{Z_i - \tmu, \alpha u} \ge \paren{1-\theta} \paren{1 - \gamma - 1/M} \tilde{d} \ge r \mcom
	\end{equation*}
	which completes the proof by \Cref{lem:quad-val-sdp-val}.
\end{proof}
The following corollary is useful in the analysis in \Cref{subsec:distance-estimation}.
\begin{corollary}
\label{cor:existence-of-good-sol-for-sdp}
Let $Z_1, \dots, Z_k, \tmu \in \R^d$, and $r^* = 2\sqrt{k/n}$.
Assume  \Cref{ass:when-at-mu-sdp-str}, holds, $\normt{\mu -\tmu} \ge M\paren{5M}r^*$.
Then
\begin{equation*}
\SDP\paren{\tmu,  \paren{1-\frac{1}{M}} \normt{\mu -\tmu}, Z} \ge \paren{1 - \frac{1}{M}}k \mper
\end{equation*}
\end{corollary}
The following lemma sates that the set defined in \Cref{lem:goodness-set}, has a high volume.
\begin{lemma}[size of the good set]
	\label{lem:size-good-set}
	Suppose $\Delta \in \R^d$ such that $\normt{\Delta} = 1$, $0 \le \gamma \le 1$ and $0 \le \zeta \le \theta \le 1$. Let
	\begin{equation*}
		S_{\theta, \zeta, \gamma, \Delta} = \Set{\alpha u \suchthat \iprod{u, \Delta} \ge 1- \gamma, 1- \theta \le \alpha \le 1 - \zeta, \normt{u} = 1, u \in \R^d} \mper
	\end{equation*}
	Let $\bbB_r^d$ denote the unit ball of radius $r$ in $\R^d$. Then
	\begin{equation*}
		\frac{\vol\paren{S_{\theta, \zeta, \gamma, \Delta}}}{\vol\paren{\bbB_{1 - \zeta}^d}} \ge \Paren{1 - \Paren{\frac{1 - \theta}{1-\zeta}}^d}
		\cdot
		\frac{1}{2}\Paren{\frac{\sqrt{1 - \paren{1- \gamma}^2}}{2}}^{d-1} \mper
	\end{equation*}
\end{lemma}
\begin{proof}
	In order to prove this lemma we lower bound $\vol\paren{S_{\theta, \zeta, \gamma, \Delta}},/\vol\paren{S_{1, \zeta, \gamma, \Delta}}$, $\vol\paren{S_{1, \zeta, \gamma, \Delta}},/\vol\paren{S_{1, 0, \gamma, \Delta}}$, $\vol\paren{S_{1,0, \gamma, \Delta}}/ \vol\paren{\bbB_1^d}$, $\vol\paren{\bbB_1^d}/\vol\paren{\bbB_{1-\zeta}^d}$
	
	Note that $S_{1, \theta, \gamma, \Delta}$ is just a scaling of the set $S_{1, \zeta, \gamma, \Delta}$ with respect to the origin with ratio $\paren{1-\theta}/\paren{1-\zeta}$.
	Therefore
	\begin{equation*}
		\frac{
			\vol\paren{S_{1, \theta, \gamma, \Delta}}
		}{
			\vol\paren{S_{1, \zeta, \gamma, \Delta}}
		} = \Paren{\frac{1- \theta}{1-\zeta}}^d \mcom
	\end{equation*}
	and 
	\begin{equation*}
		\frac{
			\vol\paren{S_{\theta, \zeta, \gamma, \Delta}}
		}{
			\vol\paren{S_{1, \zeta, \gamma, \Delta}}
		} = 1 - \Paren{\frac{1- \theta}{1-\zeta}}^d \mper
	\end{equation*}
	Similarly,
	\begin{equation*}
		\frac{
			\vol\paren{S_{1, \zeta, \gamma, \Delta}}
		}{
			\vol\paren{S_{1, 0, \gamma, \Delta}}
		} = \Paren{1-\zeta}^d \mper
	\end{equation*}
	To bound $\vol\paren{S_{1, 0, \gamma, \Delta}}/ \vol\paren{\bbB_1^d}$, 
	note that this ratio is equal to the ratio of the surface area of the corresponding cap of $S_{1,0, \gamma, \Delta}$ to the surface area of the unit sphere, therefore by \Cref{fact:lb-spherical-cap},
	\begin{equation*}
		\frac{\vol\paren{S_{1,0, \gamma, \Delta}}}{\vol\paren{\bbB_1^d}}
		\ge 
		\frac{1}{2} \Paren{\frac{\sqrt{1 - \paren{1 - \gamma}^2}}{2}}^{d-1} \mper
	\end{equation*}
	Now it remains to lower bound $\vol\paren{\bbB_{1}^d}/\vol\paren{\bbB_{1-\zeta}^d}$. Note that $\bbB_1^d$ is a scaling of the set $\bbB_{1-\zeta}^d$ with respect to the origin with ratio $1/\paren{1-\zeta}$, therefore,
	\begin{equation*}
		\frac{\vol\paren{\bbB_1^d}}{\vol\paren{\bbB_{1-\zeta}^d}} = \paren{1-\zeta}^{-d} \mper
	\end{equation*}
	Multiplying the bounds gives us the desired inequality.
\end{proof}
\begin{corollary}[size of the good set] 
	\label{cor:size-of-the-good-set}    
	Under the same assumptions as \Cref{lem:size-good-set},
	\begin{equation*}
		\frac{\vol\paren{S_{\theta,\zeta, \gamma, \Delta}}}{\vol\paren{\bbB_{1-\zeta}^d}} \ge \frac{\theta - \zeta}{2} \cdot \paren{\gamma^{1/2}/2}^{d-1} \mcom
	\end{equation*}
	More specifically, for every $\theta, \zeta, \gamma$, there exists some $C_{\theta, \zeta, \gamma} \ge 0$, such that
	\begin{equation*}
		\log \frac{\vol\paren{\bbB_{1-\zeta}^d}}{\vol\paren{S_{\theta,\zeta, \gamma, \Delta}}} \le C_{\theta, \zeta, \gamma} d \mper
	\end{equation*}
\end{corollary}
\subsubsection{Accuracy Analysis of the Exponential Mechanism for $\SDPVAL$}
\label{subsubsec:sdp-val-exp-acc}
\begin{theorem}[accuracy guarantee of the exponential mechanism for $\SDPVAL$]
	\label{thm:accuracy-exponential-sdpval}
	Let $Z_1 \dots Z_k, \mu, \tmu \in \R^d$, $r^* = 2\sqrt{k/n}$, $ 1/(2M) \ge \zeta \ge 0$. Suppose  $M\ge 1000$, $d \le 1.15 \normt{\mu - \tmu}$, $r = d/1.2$, and that \Cref{ass:when-at-mu-sdp-str} holds.
	Let $\cM_E$ be the exponential mechanism when instantiated with the score function $f\paren{y} = \SDPVAL\paren{y; \tmu, r, Z}$ and privacy budget $\epsilon/T$, over the ball of radius $1-\zeta$ in $\R^d$, as in \Cref{alg:gradient-estimation}. Let $y_0$ denote the vector returned by the exponential mechanism, then there exists a constant $C_M$ depending only on $M$ such that
	\begin{equation*}
		\SDPVAL\paren{y} \ge \paren{1 - \frac{1}{M}} k - \frac{2T}{\epsilon} \Paren{C_M d + \log \paren{T/\beta}} \mcom
	\end{equation*}
	with probability $1-\beta/T$.
	More specifically, if $k \ge 2T M\paren{C_{M} d + \log\paren{T/\beta}}/\eps$,
	then
	\begin{equation*}
		\SDPVAL\paren{y} \ge \paren{1 - \frac{2}{M}} k \mcom
	\end{equation*}
	with probability $1 - \beta/T$.
\end{theorem}
\begin{proof}
	This is a direct result of the volume based exponential mechanism, and our arguments about the volume of the good set and sensitivity i.e. \Cref{thm:volume-based-exponential-mechanism}, \Cref{lem:volume-ratio-sdp-val} and \Cref{lem:sdp-val-sensitivity}.
\end{proof}

\subsubsection{Runtime Analysis of the Exponential Mechanism for $\SDPVAL$}
\label{subsubsec:sdpval-runtime-analysis}
\begin{lemma}[runtime of the exponential mechanism]
	\label{lem:runtime-exponential-sdpval}
	There exists a sampling scheme for the exponential mechanism for $\SDPVAL$ in \Cref{alg:gradient-estimation}
	that takes time $\poly\paren{n, d, 1/\epsilon}$.
\end{lemma}
This lemma is a direct application of \Cref{thm:efficient-sampling}'s sampling algorithm. Note that we have shown that $\SDPVAL$ has Lipschitzness (\Cref{lem:sdpval-lipschitzness}), concavity (\Cref{cor:sdp-val-concavity}), and low sensitivity (\Cref{lem:sdp-val-sensitivity}), therefore we are allowed to use the mentioned result.

\subsubsection{Proof of \Cref{thm:deciding-between-2}}
\label{subsubsec:proof-deciding-between-2}
First let's prove the following lemma. 

\begin{lemma}
\label{lem:good-direction}
Let $Z_1\dots, Z_k, \mu, \tmu,y_0 \in \R^d$, $\tilde{d} = \normt{\mu - \tmu}, \Delta = \Abs{\mu - \tmu}/\tilde{d}$. Suppose $\Abs{\iprod{y, \Delta}} > 0.1$, $\iprod{y, \Delta} > 0$, $\normt{\mu - \tmu} \ge M \paren{5 M r^*}$, $M \ge 1000$, and that \Cref{ass:when-at-mu-simple-str} holds.
Then
\begin{equation*}
	\Abs{\Set{i \suchthat \iprod{Z_i - \tmu, y} > 0}} \ge \paren{1 - 1/M} k \mper
\end{equation*}
\end{lemma}
\begin{proof}
The proof is similar to \cite{CherapanamjeriFB19}. 
By \Cref{ass:when-at-mu-simple-str}, we know that there exists a set $G$, such that $\abs{G} \ge \paren{1 - 1/M} k$ and 
\begin{equation*}
	\forall \in G \; : \; \iprod{Z_i - \mu, -y} \le 5Mr^* \mper
\end{equation*}
Therefore, for every $i \in G$,
\begin{align*}
	\iprod{Z_i - \tmu, y} &= \iprod{Z_i - \mu, y} + \iprod{\mu - \tmu, y} \\ 
	&\ge 
	- 5Mr^* + \iprod{\Delta, y}  \normt{\mu - \tmu} \\
	&\ge 
	-5Mr^* + 0.1 M \paren{5 M r^*}  \\
	& > 0 \mcom
\end{align*}
which completes the proof.
\end{proof}
Now we use this lemma to prove \Cref{thm:deciding-between-2}.
\begin{proof}[Proof of \Cref{thm:deciding-between-2}]
Since either $\iprod{y_0, \Delta} > 0$, or $\iprod{-y_0, \Delta} > 0$, by \Cref{lem:good-direction}, and the guarantee of the exponential mechanism,
the exponential mechanism returns $\wh{y}$ such that 
\begin{equation*}
	\Score\paren{\wh{y}} \ge \paren{1 - 1/M} k - 2T/\epsilon \paren{\log 2+ \log T/\beta} \mcom
\end{equation*}
with probability $1- \beta/T$. Since $k \ge 2MT\paren{\log 2+ \log T/\beta}/\eps$,
the exponential mechanism returns $\wh{y}$ such that
\begin{equation*}
	\Score\paren{\wh{y}} \ge \paren{1 - 2/M} k \mcom
\end{equation*}
with probability $1- \beta / T$.
If $\Score\paren{\wh{y}} \ge \paren{1 - 2/M} k$, then by  \Cref{lem:good-direction}, and the assumption that $M \ge 1000$, $\Score\paren{\wh{y}}\ge \paren{1 - 2/M} k$ implies that $\iprod{\wh{y}, \Delta} > 0$, as desired.
\end{proof}
\subsection{Gradient Descent}
\label{subsec:gradient-descent}
We show that
if the algorithm has not stopped in a round of the algorithm and that every step of the algorithm i.e. \textsc{halt-estimation}, \textsc{distance-estimation}, and \textsc{gradient-estimation}, have been successful, then we move closer to $\mu$ (the central point) linearly (\Cref{lem:gd-step}). 

Finally we conclude that if all of the steps of the algorithm are successful then either we halt and return a point with optimal distance or we get closer to the central point $\mu$, linearly (\Cref{thm:gradient-descent})
\begin{lemma}[Gradient Descent Step]
\label{lem:gd-step}
Let $Z_1, \dots, Z_k, \tmu_{t-1}, \mu \in \R^d$, $\Delta_{t-1} = \paren{\mu - \tmu_{t-1}}/\normt{\mu - \tmu_{t-1}}$, the step-size $\eta = 0.075$.
Suppose in the $t$-th round inside the for loop in \Cref{alg:fine-estimation}, we have not halted, and all of
\textsc{halt-estimation}, \textsc{distance-estimation}, and \textsc{gradient-estimation}, have been successful, i.e. the guarantees in \Cref{thm:halt-estimation}, \Cref{thm:distance-estimation}, and \Cref{thm:gradient-estimation} are satisfied and
\begin{equation*}
0.99 \normt{\tmu_{t-1} - \mu} \le d_t \le 1.15\normt{\tmu_{t-1} - \mu}, \quad
\iprod{g_t, \Delta_{t-1}} > 0.1, \quad \normt{g_t} \le 1 \mper
\end{equation*}
Then
\begin{equation*}
\normt{\tmu_{t} - \mu} \le 0.999 \normt{\tmu_{t-1} - \mu} \mper
\end{equation*}
\end{lemma}
\begin{proof}
The proof is similar to \cite{CherapanamjeriFB19}.
Let $\tilde{d}_{t-1} = \normt{\tmu_{t-1} - \mu}$. Then
\begin{align*}
\snormt{\tmu_t - \mu}
&= 
\snormt{\tmu_t - \tmu_{t-1} +\tmu_{t-1}  - \mu}\\
&=
\snormt{g_t d_t \eta} + \tilde{d}_{t-1}^2 + 2\iprod{g_2 d_t \eta, \tmu_{t-1} - \mu} \\
&\le
\tilde{d}_{t-1}^2 \paren{1 + 1.15^2 \eta^2 } + 2\iprod{g_2 d_t \eta, -\Delta_{t-1}} \\
&\le
\tilde{d}_{t-1}^2 \paren{1 + 1.15^2 \eta^2 - 2 \times 0.99 \times 0.1 \times \eta }
\\
&=
0.993
\snormt{\tmu_{t-1} - \mu} \mper
\end{align*}
Taking square root gives us the final result.
\end{proof}
\begin{theorem}[Gradient Descent]
\label{thm:gradient-descent}
Let $Z_1, \dots Z_k, \mu_0 \in \R^d$ $M\ge 1000$, $r^* = 2\sqrt{k/n}$.
Suppose the guarantees in \Cref{thm:halt-estimation}, \Cref{thm:distance-estimation}, and \Cref{thm:gradient-estimation},
have been successful in every round $t$.
Then two outcomes are possible, either \Cref{alg:fine-estimation} doesn't halt and returns $\tmu^*$ such that
\begin{equation*}
\normt{\tmu^* - \mu} \le \paren{0.999}^T \normt{\mu_0 - \mu}\mcom
\end{equation*}
or \Cref{alg:fine-estimation}, halts and returns $\tmu^*$ such that 
\begin{equation*}
\normt{\tmu^* - \mu} \le \paren{4M \paren{5M + 1} + 15M}r^* \mper
\end{equation*}
\end{theorem}
\begin{proof}
Putting together \Cref{lem:gd-step}, and \Cref{thm:halt-estimation}.
\end{proof}
\subsection{Robustness}
\label{subsec:fine-robust}
\begin{remark}[Robustness of Fine Estimation]
	\label{rem:fine-robustness}
	Note that  $\SDP$, and $\SDPVAL$, and our quadratic programs all have sensitivity $1$ by \Cref{lem:sdp-val-sensitivity}, and \Cref{cor:sdp-sensitiviy}. There exists a universal constant $C_0$ (say $C_0 \le 1/1000$), such that our our arguments would still hold under $C_0k$-fraction corruption of the bucketed means $\paren{Z_1, \dots, Z_k}$. Now suppose that $m = C_0 /\eta_0$, then under $\eta \le \eta_0$ corruption of the samples $\paren{X_1, \dots X_n}$,  $\SDP$, $\SDPVAL$, and our quadratic programs are perturbed by at most  
	\begin{equation*}
		\eta n \le \eta_0 mk \le C_0 k \mper
	\end{equation*}
	Therefore, the guarantees of the algorithm will still hold. Moreover, in this setting, by \Cref{thm:fine-estimation}, we obtain an estimate  $\tmu^*$, where $\normt{\mu - \tmu^*} \le \cO\paren{\frac{1}{\sqrt{m}}} = \cO\paren{\sqrt{\eta_0}}$, as desired.
\end{remark}

%% file: lbs.tex
\newcommand*{\normtv}[1]{\Norm{#1}_{\mathrm{TV}}}

\section{Lower Bounds}
In this section, we prove lower bounds for mean estimation.
While lower bounds for these settings were known previously~\cite{HardtT10, BarberD14, BunKSW19, KamathSU20}, in contrast to prior work, we focus explicitly on the failure probability $\beta$. 
We first establish a general purpose lower bound framework, and then apply it to heavy-tailed and Gaussian distributions.
Technically, we follow the standard packing lower bound technique.

\begin{theorem}
  \label{thm:lb}
  Let $\cP = \{P_1, \dots, P_m\}$ be a set of distributions, and $P_O$ be a distribution such that for every $P_i \in \cP$, $\normtv{ P_i - P_O} \leq \gamma$. 
  Let $\cG = \{G_1, \dots, G_m\}$ be a collection of disjoint subsets of some set $\cY$.
  If there is an $\e$-DP algorithm $M$ such that $\Pr_{X \sim P_i^n}[M(X) \in G_i] \geq \beta$ for all $i \in [m]$, then
  \[ n \geq \Omega \left(\frac{\log m + \log (1/\beta)}{\gamma \left(e^{2\e} - 1\right)}\right)\mper\]
\end{theorem}
\noindent Note that for the usual regime $\e \leq 1$, we can replace the $e^{2\e} - 1$ in the denominator with $\e$.

\begin{proof}
Assume there exists some mechanism $M$ that satisfies the theorem's assumptions. 
  By the definition of $\epsilon$-DP and group privacy, we know that for any two datasets $D, D'$, and any $i \in [m]$ we have
\begin{equation}
\label{eq:nonr}
  \Pr \Brac{M(D) \in  G_i} \le  e^{\epsilon \# \Brac{D \neq D'}} \Pr \Brac{M(D') \in  G_i} \mcom
\end{equation}
where $\#\Brac{D \neq D'}$ denotes the Hamming distance between $D$ and $D'$.
  Note that for the time being, $D$ and $D'$ are fixed datasets (as opposed to random variables), and $0 \le \#\Brac{D \neq D'} \le n$.

  Basic facts about coupling imply that, for any two distributions $P_i, P_j \in \cP$, there exists a coupling $\omega_{i,j}^*$ of $P_i^n$ and $P_j^n$ such that if $(X_i, X_j) \sim \omega_{i,j}^*$, we have 
\begin{equation}
\label{eq:binom}
  \# \Brac{X_i \neq X_j} \sim \operatorname{Binomial}\Paren{n, \normtv{P_i - P_j}} \mper
\end{equation}
  Note that by the triangle inequality (through $P_O$), we have that $\normtv{P_i - P_j} \leq 2\gamma$.

By taking expected value of both sides of \Cref{eq:nonr}, we get
\begin{equation}
\label{eq:mid}
  \E_{(X_{i}, X_{j}) \sim \omega^*} \Brac{\Pr \Brac{M(X_j) \in  G_j} }
\le
  \E_{(X_{i}, X_{j}) \sim \omega^*}
\Brac{
  e^{\epsilon \# \Brac{X_i \neq X_j}} \cdot \Pr \Brac{M(X_i) \in  G_j}
} \mcom
\end{equation}
for any two $P_i, P_j$.

For every $t \in [m]$, define $Z_t^{(i)}= \Pr_{M}\Brac{M(X_i) \in G_t}$. 
Note that the probability here is only over $M$, $Z_t^{(i)}$ remains a random variable which depends on the sampling $X_i \sim P_i^n$. 
  Stated differently, to evaluate the random variable $Z_t^{(i)}$, one first samples $X_i \sim P_i^n$, and then computes $\Pr_{M}\Brac{M(X_i) \in G_t}$ for the realized dataset $X_i$. 
  Since the $G_t$'s are disjoint, the law of total probability implies $\sum_t Z_t^{(i)} \le 1$ for any $i \in [m]$.
By the accuracy guarantee of the mechanism $M$, we know that
\begin{equation*}
  \E_{X_i \sim  P_i^n}\Brac{Z_i^{(i)} } = \Pr_{X_i \sim P_i^n, M} \Brac{M \Paren{X_i} \in G_t} \ge 1- \beta \mcom
\end{equation*}
therefore
  $  \E _{X_i \sim P_i^n} \Brac{\sum_{t \neq i} Z_t^{(i)}} \le \beta$. Since $Z_t^{(i)} \le 1$,
\begin{equation*}
  \sum_{t \neq i}  \E \Brac{ \left(Z_t^{(i)}\right)^2 }  \le \sum_{t \neq i} \E \Brac{  Z_t^{(i)} } \le \beta \mper
\end{equation*}
Therefore by the pigeonhole principle, for every $i$, there exists some $\ell$ such that
\begin{equation}
\label{eq:php}
  \E_{X_i \sim P_i^n} \Brac{\Pr\Brac{M\Paren{X_i}\in G_\ell}^2 }= \E_{X_i \sim P_i^n} \Brac{\left(Z_\ell^{(i)}\right)^2} \le \frac{\beta}{m-1} \mper
\end{equation}
Now we take $j = \ell$ in \Cref{eq:mid}. We get
\begin{equation*}
  \E_{(X_{i}, X_{\ell}) \sim \omega^*} \Brac{\Pr \Brac{M\Paren{X_\ell} \in  G_\ell} }
\le
  \E_{(X_{i}, X_{\ell}) \sim \omega^*}
\Brac{
	e^{\epsilon \# \Brac{X_i \neq X_\ell}} \cdot \Pr \Brac{M\Paren{X_i} \in  G_\ell}
} \mper
\end{equation*}
By the accuracy guarantee of $M$ we know that we can lower bound the left-hand side by $1 - \beta $.
We upper bound the right-hand side by H\"older's inequality.
\begin{equation*}
  \E_{(X_{i}, X_{\ell}) \sim  \omega^*}
\Brac{
	e^{\epsilon \# \Brac{X_i \neq X_\ell}} \cdot \Pr \Brac{M\Paren{X_i} \in  G_\ell}
}
\le
  \E_{(X_{i}, X_{\ell}) \sim \omega^*}
\Brac{
	e^{2 \epsilon \# \Brac{X_i \neq X_\ell}}
}^{1/2}
 \cdot
  \E_{(X_{i}, X_{\ell}) \sim \omega^*}
\Brac{\Pr\Brac{M\Paren{X_i}\in G_\ell}^2 }^{1/2}
\mper
\end{equation*}
For the second term, \Cref{eq:php} implies that
\begin{equation*}
  \E_{(X_{i}, X_{\ell}) \sim \omega^*}
 \Brac{\Pr\Brac{M\Paren{X_i}\in G_\ell}^2 }
  =
\E_{X_i \sim P_i^n} \Brac{\Pr\Brac{M\Paren{X_i}\in G_\ell}^2 }
\le
\frac{\beta}{m-1} \mper
\end{equation*}
It remains to bound the first term.
By \Cref{eq:binom},
we know that $\#\Brac{X_i \neq X_\ell} \sim \operatorname{Binomial}\Paren{n, \normtv{P_i - P_\ell} }$, and thus, this term is just the moment-generating function of the Binomial distribution, evaluated at the point $2\e$.
Since the MGF of $\operatorname{Binomial}(n,p)$ is $(1 - p + pe^t)^n$, 
\begin{equation*}
  \E_{(X_{i}, X_{\ell}) \sim \omega^*}
\Brac{
	e^{2 \epsilon \# \Brac{X_i \neq X_\ell}}
}^{1/2}
=
\Paren{
1 - p + pe^{2\epsilon}
}^{n/2}
\le
  \exp \Paren{\frac{n}{2} \cdot  \log \Paren{1+ 2\gamma \Paren{e^{2\epsilon} - 1}}} \mper
\end{equation*}
Note that $x \ge \log \Paren{1  + x} $,  therefore
\begin{equation*}
  \E_{(X_{i}, X_{\ell}) \sim \omega^*}
	\Brac{
		e^{2 \epsilon \# \Brac{X_i \neq X_\ell}}
	}^{1/2}
	\le
  \exp \Paren{n\gamma \cdot    \Paren{e^{2\epsilon} - 1}} \mper
\end{equation*}
Combining these derivations, 
\begin{equation*}
  1-\beta \le  \exp \Paren{ n \gamma \Paren{e^{2\epsilon} - 1}} \cdot \sqrt{\frac{\beta}{m-1}} \mper
\end{equation*}
Therefore,
\begin{equation*}
  n \ge \Omega \Paren{\frac{\log \Paren{\frac{\Paren{1-\beta}^2}{\beta}} + \log m}{ \gamma \Paren{e^{2\epsilon} - 1}}} \mper
\end{equation*}
Since $\beta < 1$, we can conclude
\begin{equation*}
  n \ge \Omega \Paren{\frac{\log m + \log \frac{1}{\beta}}{\gamma \left( e^{2\epsilon} - 1\right)}} \mper
\end{equation*}
\end{proof}

Now, we apply this theorem to derive lower bounds for mean estimation for distributions with bounded moments. 
\begin{theorem}
  \label{thm:lb-bdd-mom}
  Suppose there exists an $\e$-DP algorithm $M$ such that for every distribution $D$ on $\mathbb{R}^d$ such that $\|\E_{X \sim D} X \| \leq R$ and $\E_{X \sim D} [ |\langle X - \E_{X \sim D} [X] , v  \rangle|^k ]\leq 1$ for some $k \geq 2$ and every unit vector $v \in \mathbb{R}^d$, given $X_1, \dots, X_n \sim D$, with probability $1 - \beta$ the algorithm $M$ outputs $\hat \mu$ such that $\| \hat \mu - \E_{X \sim D} X\| \leq \alpha$ where $\alpha$ is smaller than some absolute constant. 
  Then we have that 
  \[n \geq \Omega\left(\frac{d + \log (1/\beta)}{\alpha^2} + \frac{d + \log (1/\beta)}{\alpha^{\frac{k}{k-1}} \e} + \frac{d \log R + \log(1/\beta)}{\e}\right) \mper \]
\end{theorem}
\begin{proof}
  The first term is folklore from the non-private setting and obtains even for estimating the mean of a Gaussian.

  To obtain the last term, consider a maximum $1$-packing of the ball with $\ell_2$-radius $R$.
  Standard results imply that this will be of size $R^{\Omega(d)}$.
  Consider a point mass distribution at each of these points. 
  Each distribution clearly satisfies the mean and bounded moment condition in the theorem statement.  
  We apply Theorem~\ref{thm:lb} with $\cP$ being this set of distributions, noting that $\gamma \leq 1$ trivially.
  Combining the $1$-packing property with our accuracy requirement that $\alpha \leq 1/2$, we get that the set of sets $\cG$ are indeed disjoint.
  This gives the third term.

  To obtain the second term, we employ a construction of Barber and Duchi~\cite{BarberD14}.
  We consider a set of distributions parameterized by unit vectors $v \in \mathbb{R}^d$. 
  The distribution parameterized by $v$ has mass $1-25\alpha^{\frac{k}{k-1}}$ at the origin and mass $25\alpha^\frac{k}{k-1}$ on the point $\frac{1}{6}\alpha^{-\frac{1}{k-1}}v$.

  We show that each of these distributions satisfies the conditions of the theorem statement.
  The mean of this distribution is $25\alpha^\frac{k}{k-1} \frac{1}{6}\alpha^{-\frac{1}{k-1}}v = \frac{25}{6}\alpha v$, which has $\ell_2$ norm $\leq R$ for $R$ larger than an absolute constant.
  We check the bounded moment condition for an arbitrary unit vector $u$:
  \begin{align*}
    \E [ |\langle X - \E [X] , u  \rangle|^k ] &\leq \E [ |\langle X - \E [X] , v  \rangle|^k ] \\
    &\leq (1 - 25\alpha^\frac{k}{k-1}) \left( \frac{25}{6}\alpha\right)^k + 25\alpha^\frac{k}{k-1}\left(\frac{1}{6}\alpha^{-\frac{1}{k-1}}\right)^k \\
    &\leq 1 \mper
  \end{align*}
  The last inequality uses the bounds on $\alpha$ and $k$.
  Therefore, each such distribution satisfies the theorem conditions. 

  Now, we look to apply Theorem~\ref{thm:lb}.
  Let $\cP$ be the set of distributions parameterized by a maximum set of points $V$ on the unit sphere such that for each $v, v' \in V$, we have that $\|v - v'\| \geq 1/2$.
  Standard results imply that this set will be of size at least $2^{\Omega(d)}$.
  All distributions are within distance $\gamma = 25\alpha^\frac{k}{k-1}$ from the point mass distribution on the origin.
  Finally, since for each pair $v, v'$ we have $\|v - v'\| \geq 1/2$, this implies that the means of any two distributions are at distance $\geq \frac{25}{12}\alpha$ from each other, and the balls of radius $\alpha$ around each mean are disjoint.
  Letting these be the $G_i$'s, we conclude the theorem statement using Theorem~\ref{thm:lb}.
\end{proof}

We can similarly derive a lower bound for mean estimation of Gaussian distributions.
\begin{theorem}
  \label{thm:lb-gaussian}
  Suppose there exists an $\e$-DP algorithm $M$ such that for every Gaussian distribution $\mathcal{N}(\mu, I)$ on $\mathbb{R}^d$ such that $\|\mu\| \leq R$, given $X_1, \dots, X_n \sim D$, with probability $1 - \beta$ the algorithm $M$ outputs $\hat \mu$ such that $\| \hat \mu - \E_{X \sim D} X\| \leq \alpha$ where $\alpha$ is smaller than some absolute constant. 
  Then we have that 
  \[n \geq \Omega\left(\frac{d + \log (1/\beta)}{\alpha^2} + \frac{d + \log (1/\beta)}{\alpha \e} + \frac{d \log R + \log(1/\beta)}{\e}\right) \mper \]
\end{theorem}
\begin{proof}
  The first term is again standard from the non-private literature, and the third term can be obtained akin to in the proof of Theorem~\ref{thm:lb-bdd-mom} (though we consider a Gaussian with each center, rather than a point mass).

  The second term is also similar to that of Theorem~\ref{thm:lb-bdd-mom}.
  The construction will involve the set of Gaussians with mean $4\alpha v$ where $v$ is a unit vector. 
  These are Gaussian with bounded mean, thus satisfying the conditions of the theorem statement. 

  As before, we try to apply Theorem~\ref{thm:lb}.
  We let $\cP$ be the set of $2^{\Omega(d)}$ distributions associated with a packing of unit vectors $V$ such that for each $v, v' \in V$, we have that $\|v - v'\| \geq 1/2$. 
  Each pair of means is at distance $\geq 2 \alpha$, and thus the balls of radius $\alpha$ around each mean are disjoint.
  The proof is concluded by noting that all these distributions are at total variation distance $\gamma = O(\alpha)$ from $N(0, I)$, due to the standard fact that $\ell_2$ distance between means and total variation distance are equivalent (up to constant factors) for identity-covariance Gaussians when the means are at distance smaller than an absolute constant.
\end{proof}

%% file: deferred-proofs.tex
\section{Deferred Proofs}
\begin{lemma}[$\SDP$ decreasing]
\label{lem:sdp-decreasing-r}
$\SDP$  (\Cref{def:sdp-relax}) is decreasing in $r$.
\end{lemma}
\begin{proof}
Suppose $r > r'$. Our goal is to prove that $\SDP\paren{\tmu, r, Z} \le \SDP\paren{\tmu, r', Z}$.
Suppose $X$ is the solution that maximizes $\SDP\paren{\tmu, r, Z}$. Then $X$ is a feasible solution to $\SDP\paren{\tmu, r', Z}$. Therefore $\SDP\paren{\tmu, r, Z} \le \SDP\paren{\tmu, r', Z}$.
\end{proof}

\begin{lemma}
	\label{lem:b-2-less-than-1-sos}
	$\Set{b^2 = b} \proves_2 b^2 \le 1$.
\end{lemma}
\begin{proof}
	\begin{equation*}
		\paren{1 - b}^2 + 2\paren{b - b^2} = 1 - b^2 \implies b^2 = b \proves_2 b^2 \le 1 \mper
	\end{equation*}
\end{proof}

\begin{lemma}[boundedness of the solutions of \Cref{def:coarse-sdp} ]
	\label{lem:boundedness-coarse-optimization-problem}
	The constraints of $\textsc{coarse-sdp}$, as defined in \Cref{def:coarse-sdp}, imply that $\snormt{\paren{b, v}} \le \paren{2R + R/1000}^2 + k$.  
\end{lemma}
\begin{proof}
From the constraints we have that $\snormt{v} \le \paren{2R + R/1000}^2$. Furthermore by \Cref{lem:b-2-less-than-1-sos}, we know that $\snormt{b} \le k$. Therefore $\snormt{\paren{b, v}} \le \paren{2R + R/1000}^2  +k.$
\end{proof}

\begin{lemma}[boundedness of the solutions of \Cref{def:SDPVAL}]
	\label{lem:boundedness-quadratic-optimization-problem}
	The constraints of $\SDPVAL$, as defined in (\Cref{def:quadratic-optimization-problem}) imply that $\snormt{\paren{b, v}} \le k+1$.  
\end{lemma}
\begin{proof}
	From the constraints we have that $\snormt{v} = 1$. Furthermore by \Cref{lem:b-2-less-than-1-sos}, we know that $\snormt{b} \le k$. Therefore $\snormt{\paren{b, v}} \le k+1$
\end{proof}

\begin{lemma}[$\QUADVAL\paren{y} \le \SDPVAL\paren{y}$]
	\label{lem:quad-val-sdp-val}
	For every $Z_1, \dots Z_k, \tmu \in \R^d$, and $r\ge 0$,
	\begin{equation*}
	\QUADVAL\paren{y; \tmu, r, Z} \le \SDPVAL\paren{y; \tmu, r, Z} \mper
	\end{equation*}
\end{lemma}
\begin{proof}
	Suppose $\paren{b, v}$ is a solution that maximizes the quadratic optimization problem of $\QUADVAL\paren{y}$.
	Therefore $\sum_{i=1}^k b_i^2 = \QUADVAL\paren{y}$.
	Let $X = \paren{1, b, v} \tensor \paren{1, b, v}$, then $X$ will be PSD.
	Suppose
	\begin{equation*}
		X = \begin{bmatrix}
			1 & \transpose{b} & \transpose{v} \\
			b & B & W \\
			v & W^T & V
		\end{bmatrix} \mcom
	\end{equation*}
	then $X$ has $\Tr\paren{B} = \sum_{i=1}^k b_i^2 = \QUADVAL\paren{y}$, and satisfies every constraint in the $\SDPVAL$ optimization problem (\Cref{def:SDPVAL}), except for $\Tr\paren{V} = 1$. This is because $\normt{y}$ may be smaller than $1$. Let
	\begin{equation*}
		X' = \begin{bmatrix}
			1 & \transpose{b} & \transpose{v} \\
			b & B & W \\
			v & W^T & V'= V/\Tr\paren{V}
		\end{bmatrix} \mper
	\end{equation*}
	Now the objective function's value and all of the constraints will be the same as $X$, except for the $\Tr\paren{V} = 1$ constraint, which is now satisfied. It remains to argue that $X'$ is PSD, which is a direct result of \Cref{lem:sdp-scaling}.
\end{proof}

\begin{lemma}
	\label{lem:sdp-scaling}
	If $M = \begin{bmatrix}
		A & B \\
		\transpose{B} & C
	\end{bmatrix}$ is PSD, then for any $k\ge 1$, 
	$M' = 
	\begin{bmatrix}
		A & B \\
		\transpose{B} & kC
	\end{bmatrix}$  is also PSD.
\end{lemma}
\begin{proof}
	We prove the third condition in \Cref{lem:psd-block-matrix} for $M'$. Note that $\trsp{kC} = \trsp{C}/k $.
	From $M \succcurlyeq 0$, 
	\begin{itemize}
		\item $A \succcurlyeq 0$,
		\item $C \succcurlyeq 0$,
		\item $ (I - C\trsp{C})\transpose{B} = 0$,
		\item $A - B \trsp{C} \transpose{B} \succcurlyeq 0$,
	\end{itemize}
	and we want to prove that
	\begin{itemize}
		\item $kC \succcurlyeq 0$,
		\item $ (I - kC\trsp{\paren{kC}})\transpose{B} = 0$,
		\item $A - B \trsp{\paren{kC}} \transpose{B} \succcurlyeq 0$.
	\end{itemize}
	The first two are obvious. For the last one, note that 
	\begin{equation*}
		A \succcurlyeq 0 \implies \paren{k - 1} A \succcurlyeq 0 \implies k A \succcurlyeq A \mper
	\end{equation*}
	Therefore
	\begin{equation*}
		kA \succcurlyeq A \succcurlyeq B \trsp{C} \transpose{B} \mcom
	\end{equation*}
	and
	\begin{equation*}
		A - B\trsp{C}\transpose{B}/k \succcurlyeq 0 \mcom
	\end{equation*}
	as desired.
\end{proof}

%% file: lemmata.tex
\section{Lemmata}

\begin{lemma}
\label{lem:nulling-psd-matrices}
Suppose $M \succcurlyeq 0$ is a matrix in $\R^{d\times d}$, $1 \le k \le d$. Let $N$ be a $d \times d$ matrix such that for every $1\le i, j \le d$, $N_{i, j} = M_{i, j}$, if $i \neq k$ and $j \neq k$, and $N_{i, j} = 0$, if $i = k$, or $j = k$.
Then $N \succcurlyeq 0$.
\end{lemma}
\begin{proof}
\begin{align*}
\min_{\normt{z} = 1} \transpose{z} N z &=
\min_{\normt{z} = 1, z_k = 0} z \transpose{N} z 
\\
&=
\min_{\normt{z} = 1, z_k = 0} z \transpose{M} z \\
&\ge
\min_{\normt{z} = 1} z \transpose{M} z  \\
&\ge 
0
\end{align*}
\end{proof}

\begin{lemma}
	\label{lem:convacity-maximization}
	Suppose $\cK$ is a closed convex set, then $f(y) = \max_{\substack{x \in \cK\\ Ax = y}} \iprod{x, c}$ is concave.
\end{lemma}
\begin{proof}
	We need to prove that for $\alpha \in \Brac{0, 1}$
	\begin{equation*}
		\alpha f(y_1) + \paren{1- \alpha} f(y_2) \le f(\alpha y_1 + \paren{1-\alpha} y_2) \mper
	\end{equation*}
	Suppose $f(y_1)$ is maximized at $x_1$ and $f(y_2)$ is maximized at $x_2$. Let $z = \alpha x_1 + \paren{1- \alpha} x_2 \in \cK$, then  $Az = \alpha y_1 + \paren{1-\alpha}y_2$
	and
	\begin{equation*}
		f(\alpha y_1 + \paren{1- \alpha} y_2) \ge \iprod{z, c} \mcom
	\end{equation*}
	which implies the desired inequality.
\end{proof}

The following lemma is the ordinary binary search.
\begin{lemma}[Non Private Binary Interval Search]
	\label{lem:non-private-binary-interval-search}
	Suppose an interval $\brac{0, D}$ is given.
	For an interval $\brac{s, e} \subseteq \brac{0, D}$, and for an accuracy parameter $a$, the \textit{binary interval search} algorithm returns a point $x_T$ after $T = \log\paren{D / a}$ steps of the following binary search, where $x_i$ is the parameter of the binary search at step $i$.
	\begin{enumerate}
		\item If $x_i < s$, move right.
		\item If $ s \le x_i \le e$, move in an arbitrary direction. Note that naturally in a binary search we would terminate in this step, but here we assume that we might go to any direction.
		\item if $x_i > e$, move left.
	\end{enumerate}
	Then $x_T \in [s - a, e + a]$.
\end{lemma}
\begin{lemma}[Cauchy-Schwarz inequality]
	\label{lem:sos-cauchy-schwarz}
	If $\paren{f, g}$ is a level-$2$ fictitious random variable over $\R^{\cU} \times \R^{\cU}$, then 
	\begin{equation*}
		\pE_{f, g} \iprod{ f, g} \le 
		\sqrt{ \pE_f \normt{f}^2}
		\cdot
		\sqrt{\pE_g \normt{g}^2} \mper
	\end{equation*}
\end{lemma}

\begin{lemma}[pseudo-covariance]
	\label{lem:pseudo-covariance}
	If $\cX$ is a degree-$2$ pseudo-distribution over $\R^d$, then
	\begin{equation*}
		\pE_\cX \Brac{\Normt{x - \pE_\cX \Brac{x}}^2} = 
		\pE_\cX \Brac{\snormt{x}} -
		\snormt{\pE_\cX \Brac{x}} \mper
	\end{equation*}
\end{lemma}
\begin{proof}
	\begin{align*}
		\pE_\cX \Brac{\Normt{x - \pE_\cX \Brac{x}}^2} &= 
		\sum_i \pE_\cX \Brac{ \Paren{x_i - \pE_\cX \Brac{x_i}}^2}
		\\
		&=
		\sum_i \pE_\cX \Brac{x_i^2} - \Paren{\pE_\cX \Brac{x_i}}^2 \\
		&=
		\pE_\cX \Brac{\snormt{x}} -
		\snormt{\pE_\cX \Brac{x}}
	\end{align*}
\end{proof}

\begin{lemma}[PSD $2 \times 2$ block matrix]
	\label{lem:psd-block-matrix}
	Given any matrix, 
	$M = 
	\begin{pmatrix}
		A & B \\
		\transpose{B} & C
	\end{pmatrix}$,
	the following conditions are equivalent:
	\begin{enumerate}
		\item $M \succcurlyeq 0$ ($M$ is positive semidefinite).
		\item $A \succcurlyeq 0, \quad (I - A\trsp{A})B = 0, \quad C - \transpose{B} \trsp{A} B \succcurlyeq 0$.
		\item $C \succcurlyeq 0, \quad (I - C\trsp{C})\transpose{B} = 0, \quad A - B \trsp{C} \transpose{B} \succcurlyeq 0$.
	\end{enumerate}
\end{lemma}

\begin{fact}[SoS Triangle Inequality]
\label{fact:sos-triangle-inequality}
Let $x, y$ be indeterminates. Let $t$ be a power of $2$. Then
\begin{equation*}
\proves_{t} \paren{x+y}^2 \le 2^{t-1} \paren{x^t +y^t} \mper
\end{equation*}
\end{fact}

\begin{fact}[lower bound for spherical caps \cite{Ball1997AnEI}]
	\label{fact:lb-spherical-cap}
	For $0 \le r \le 2$, a cap of radius $r$ on $\bbS^{n-1}$ has measure at least $\frac{1}{2} \paren{r / 2}^{n-1}$. 
\end{fact}